\documentclass[11pt]{article}
\usepackage{amsfonts}
\usepackage{amsmath}
\usepackage{amssymb}
\usepackage{xparse}
\usepackage{graphicx}
\usepackage[colorlinks=true,linkcolor=blue,citecolor=red,plainpages=false,pdfpagelabels]%
{hyperref}
\usepackage{fullpage}
\usepackage[caption=false]{subfig}%
\usepackage{color}
\providecommand{\U}[1]{\protect\rule{.1in}{.1in}}
\newtheorem{theorem}{Theorem}

\newtheorem{condition}[theorem]{Condition}

\newtheorem{definition}[theorem]{Definition}

\newtheorem{lemma}{Lemma}

\newtheorem{proposition}[theorem]{Proposition}
\newtheorem{remark}[theorem]{Remark}

\newenvironment{proof}[1][Proof]{\noindent\textbf{#1.} }{\ \rule{0.5em}{0.5em}}

\newcommand{\ket}[1]{| #1 \rangle}
\newcommand{\bra}[1]{\langle #1 |}

\def\U{\mathrm{U}}
\def\A{\mathcal{A}}
\def\B{\mathcal{B}}

\def\D{\mathcal{D}}
\def\E{\mathcal{E}}

\def\H{\mathcal{H}}
\def\U{\mathcal{U}}

\def\L{\mathcal{L}}
\def\M{\mathcal{M}}
\def\N{\mathcal{N}}

\def\P{\mathcal{P}}

\def\T{\mathcal{T}}

\def\S{\mathcal{S}}

\def\V{\mathcal{V}}
\def\W{\mathcal{W}}
\newcommand{\tr}{\operatorname{Tr}}

\numberwithin{equation}{section}

\begin{document}

\title{\textbf{Bounding the energy-constrained quantum and private capacities of phase-insensitive bosonic Gaussian channels}}

\author{Kunal Sharma\thanks{Hearne Institute for Theoretical Physics, Department of Physics and Astronomy, Louisiana State University,
	Baton Rouge, Louisiana 70803, USA}
\and 
Mark M. Wilde\footnotemark[1] \,
\thanks{Center for Computation and Technology, Louisiana State University, Baton
	Rouge, Louisiana 70803, USA}
\and Sushovit Adhikari\footnotemark[1]
\and Masahiro Takeoka
\thanks{National Institute of Information and Communications Technology,
	Koganei, Tokyo 184-8795, Japan}}

\date{\today}

\maketitle

\begin{abstract}
 We establish several upper bounds on the energy-constrained quantum and private capacities of all single-mode phase-insensitive bosonic Gaussian channels. 
 The first upper bound, which we call the "data-processing bound," is the simplest and is obtained by decomposing a phase-insensitive channel as a pure-loss channel followed by a quantum-limited amplifier channel. We prove that the data-processing bound can be at most 1.45 bits larger than a known lower bound on these capacities of the phase-insensitive Gaussian channel. We discuss another data-processing upper bound as well. Two other  upper bounds, which we call the ``$\varepsilon$-degradable bound'' and the ``$\varepsilon$-close-degradable bound,'' are established using the notion of approximate degradability along with energy constraints. We find a strong limitation on any potential superadditivity of the coherent information of any phase-insensitive Gaussian channel in the low-noise regime, as the data-processing bound is very near to a known lower bound in such cases. 
We also find improved achievable rates of private communication through bosonic thermal channels, by employing coding schemes that make use of displaced thermal states. 
We end by proving that an optimal Gaussian input state for the energy-constrained, generalized channel divergence of two particular Gaussian channels  is  the two-mode squeezed vacuum state that saturates the energy constraint. What remains open for several interesting channel divergences, such as the diamond norm or the R\'enyi channel divergence, is to determine whether, among all input states, a Gaussian state is optimal.
\end{abstract}

\setcounter{tocdepth}{1}

\tableofcontents

\section{Introduction}

One of the main aims of quantum information theory is to characterize the capacities of quantum communication channels \cite{H13book,MH06,W15book}. A quantum channel is a model for a communication link between two  parties. The properties of a quantum channel and its coupling to an environment govern the evolution of a quantum state that is sent through the channel. 

The quantum capacity $Q(\N)$ of a quantum channel $\N$ is the maximum rate at which quantum information (qubits) can be reliably transmitted from a sender to a receiver by using the channel  many times. The private capacity $P(\N)$ of a quantum channel $\N$ is defined to be the maximum rate at which a sender can reliably communicate classical messages to a receiver by using the channel  many times, such that the environment of the channel gets negligible information about the transmitted message. In general, the best known characterization of quantum or private capacity of a quantum channel is given by the optimization of regularized information quantities over an unbounded number of uses of the channel \cite{PhysRevA.55.1613,capacity2002shor,1050633,ieee2005dev}. Since these information quantities are additive for a special class of channels called degradable channels \cite{cmp2005dev,S08}, the capacities of these channels can be calculated without any regularization. However, for the channels that are not degradable, these information quantities can be superadditive \cite{DSS98,SS07,smith:170502,Supadd15,ES15}, and quantum capacities can be superactivated for some of these channels \cite{SY08,GSY11}. Hence, it is difficult to determine the quantum or private capacity of channels that are not degradable, and the natural way to characterize such channels is to bound these capacities from above and below. 

An important class of channels called phase-insensitive, bosonic Gaussian channels act as a good model for the transmission of light through optical fibers or free space (see, e.g., \cite{AS17} for a review). Within the past two decades, there have been advances in finding quantum and private capacities of bosonic channels. In particular, when there is no constraint on the  energy available at the transmitter, the quantum and private capacities of single-mode quantum-limited attenuator and amplifier channels were given in \cite{HW01,WPG07,Mark2012tradeoff,QW16,MH16}. However, the availability of an unbounded amount of energy at the transmitter is not practically feasible, and it is thus natural to place energy constraints on any communication protocol. Recently, a general theory of energy-constrained quantum and private communication has been developed in \cite{MH16}, by building on notions developed in the context of  other energy-constrained information-processing tasks \cite{H04}. For the particular case of bosonic Gaussian channels,
formulas for 
the energy-constrained quantum and private capacities of the single-mode pure-loss channel were conjectured in \cite{GSE08} and proven in \cite{Mark2012tradeoff,MH16}. Also, for a single-mode quantum-limited amplifier channel, the energy-constrained quantum and private capacities have been established in \cite{QW16,MH16}.

What remains a pressing open question in the theory of Gaussian quantum information \cite{AS17}
is to determine formulas for or bounds on the quantum and private capacities of non-degradable bosonic Gaussian
channels. Of particular interest are phase-insensitive bosonic Gaussian channels, which serve as models for several physical processes. In this article, we address this query by providing several bounds on the energy-constrained quantum and private capacities of all phase-insensitive Gaussian channels. 

To motivate the thermal channel model, consider that almost all communication systems are affected by thermal noise \cite{CC82}. Even though the pure-loss channel has relevance in free-space communication \cite{YS78,S09}, it represents an ideal situation in which the environment of the channel is prepared in a vacuum state. Instead, consideration of a thermal state with a fixed mean photon number $N_B$ as the state of the environment is more realistic, and such a channel is called a bosonic Gaussian thermal channel \cite{S09,RGRKVRHWE17}. Hence, quantum thermal channels model free-space communication with background thermal radiation affecting the input state in addition to transmission loss. Additionally, the dark counts of photon detectors can also be modeled as arising from thermal photons in the environment \cite{RGRKVRHWE17,S09}. In the context of private communication, a typical conservative model is to allow an eavesdropper access to the environment of a channel, and in particular, 
tampering by an eavesdropper
can be modeled as the excess noise realized by a thermal channel \cite{NH04,LDTG05}.

Interestingly, quantum amplifier channels model spontaneous parametric down-conversion in a nonlinear optical system \cite{CDGMS10}, along with the dynamical Casimir effect in superconducting circuits \cite{GM70}, the Unruh effect \cite{WU76}, and Hawking radiation \cite{SH72}. Moreover,  an additive-noise channel is ubiquitous in quantum optics due to the fact that the aggregation of many independent random disturbances will typically have a Gaussian distribution \cite{MH94}.   

\section{Summary of results}

Our main contribution is to establish several bounds on the energy-constrained quantum and private capacities of single-mode, phase-insensitive bosonic Gaussian channels.
We start by summarizing our upper bounds on the energy-constrained quantum capacity of thermal channels. A first upper bound is established by decomposing a thermal channel as a pure-loss channel followed by a quantum-limited amplifier channel \cite{dp2006,dp2012} and using a data-processing argument. We note that the same method was employed in \cite{smith2013}, in order to establish an upper bound on the classical capacity of the thermal channel (note that the general idea for the data-processing argument comes from the earlier work in \cite{WP07,SS08}). Throughout, we call this first upper bound the ``data-processing bound.'' We also prove that this upper bound can be at most 1.45 bits larger than a known lower bound \cite{HW01,Mark2012tradeoff} on the energy-constrained quantum and private capacity of a thermal channel. Moreover, the data-processing bound is very near to a known lower bound for the case of low thermal noise and both low and high transmissivity. 

We then prove that any phase-insensitive channel that is not entanglement-breaking \cite{HSR03} can be decomposed as the concatenation of a quantum-limited amplifier channel followed by a pure-loss channel. This theorem was independently proven in \cite{RMG18, NAJ18} (see also \cite{notesSWAT17}). It has been used to bound the unconstrained quantum capacity of a thermal channel in \cite{RMG18}, via a data-processing argument. We use this technique to prove an upper bound on the energy-constrained quantum and private capacities of a thermal channel. This technique has also been used most recently in \cite{NAJ18} in similar contexts. In particular, we find that this upper bound is very near to a known lower bound for the case of low thermal noise and both low and high transmissivity. Furthermore, this alternate data-processing upper bound and the data-processing bound mentioned in the previous paragraph are incomparable, as one is better than the other for certain parameter regimes. 

Recently, the notion of approximate degradability of quantum channels was developed in \cite{sutter2017approx}, and upper bounds on the quantum and private capacities of approximately degradable channels were established for quantum channels with
finite-dimensional input and output systems. In our paper, we establish general upper bounds on the energy-constrained quantum and private capacities of approximately degradable channels for infinite-dimensional systems. These general upper bounds can be applied to any quantum channel that is approximately degradable with energy constraints on the input and output states of the channels. In particular, we apply these general upper bounds to  bosonic Gaussian thermal and amplifier channels. 

Our second upper bound is based on the notion of $\varepsilon$-degradability of thermal channels, and we call this bound the ``$\varepsilon$-degradable bound.'' In this method, we first construct a degrading channel, such that a complementary channel of the thermal channel is close in diamond distance \cite{K97} to the serial concatenation of the thermal channel followed by this degrading channel. In general, it seems to be computationally hard to determine the diamond distance between two quantum channels if the optimization is over input density operators acting on an infinite-dimensional Hilbert space. However, in our setup, we address this difficulty by constructing a simulating channel, which simulates the serial concatenation of the thermal channel and the aforementioned degrading channel. Using this technique, an upper bound on the diamond distance reduces to the calculation of the quantum fidelity between the environmental states of the thermal channel and the simulating channel.
Based on the fact that, for certain parameter regimes, the resulting capacity upper bound is better than all other upper bounds reported here,
we believe that our aforementioned choice of a degrading channel is a good choice. 

A third upper bound on the energy-constrained quantum capacity of thermal channels is established using the concept of $\varepsilon$-close-degradability of a thermal channel, and we call this bound the ``$\varepsilon$-close-degradable bound.'' In particular, we show that a low-noise thermal channel is $\varepsilon$-close degradable, given that it is close in diamond distance to a pure-loss channel. We find that the $\varepsilon$-close-degradable bound is very near to the data-processing bound for the case of low thermal noise.

 We compare these different upper bounds with a known lower bound on the quantum capacity of a  thermal channel \cite{HW01,Mark2012tradeoff}. We find that the data-processing bound is very near to a known capacity lower bound for low thermal noise and for both medium and high transmissivity. Moreover, we show that the maximum difference between the data-processing bound and a known lower bound never exceeds $1/\ln2 \approx 1.45$ bits for all possible values of parameters, and this maximum difference is attained in the limit of infinite input mean photon number. This result places a strong limitation on any possible  superadditivity of coherent information of the thermal channel. We note here that this kind of result was suggested without proof by the heuristic developments in \cite{SS13}.
 Next, we plot these  upper bounds as well as a known lower bound versus input mean photon number for different values of the channel transmissivity $\eta$ and thermal noise~$N_B$. In particular, we find that the $\varepsilon$-close-degradable bound is very near to the data-processing bound for low thermal noise and for both medium and high transmissivity. Moreover, all of these upper bounds are very near to a known lower bound for low thermal noise and high transmissivity. We also examine different parameter regimes where the $\varepsilon$-close-degradable bound is tighter than the $\varepsilon$-degradable bound and vice versa. 
In particular, we find that the $\varepsilon$-degradable bound is tighter than the $\varepsilon$-close degradable bound for the case of high thermal noise. 

We find an interesting parameter regime where the $\varepsilon$-degradable bound is tighter than all other upper bounds, as it becomes closest to a known lower bound for the case of high noise and high input mean photon number. However, for the same parameter regime, if the input mean photon number is low, then the data-processing bound is tighter than the $\varepsilon$-degradable bound. This suggests that the upper bounds based on the notion of approximate degradability are good for the case of high input mean photon number. We suspect that these bounds could be further improved for the case of low input mean photon number if it were possible to compute or tightly bound the energy-constrained diamond norm \cite{Sh17,Wetal17} (see also Section~\ref{sec:generalized-channel-divergence} for some developments in this direction).

Similar to our bounds on the energy-constrained quantum capacity, we establish several upper bounds on the energy-constrained \textit{private} capacity of bosonic thermal channels. We also develop an improved  lower bound on the energy-constrained  private capacity of a bosonic thermal channel. In particular, we find that for  certain values of the channel transmissivity, a higher private communication rate can be achieved by using displaced thermal states as information carriers instead of coherent states.

Related to our bounds on energy-constrained quantum and private capacities of thermal channels, we establish several upper bounds on the same capacities of quantum amplifier channels. We also establish upper bounds on the energy-constrained quantum and private capacities of an additive-noise channel.

As one of the last technical developments of our paper, we address the question of computing energy-constrained channel distances  in a very broad sense, by considering the energy-constrained, generalized channel divergence of two quantum channels, as an extension of the generalized channel divergence developed in \cite{LKDW17}. In particular, we prove that an optimal Gaussian input state for the energy-constrained, generalized channel divergence of two particular Gaussian channels is the two-mode squeezed vacuum state that saturates the energy constraint. It is an interesting open question to determine whether the two-mode squeezed vacuum is optimal among all input states, but we leave this for future work, simply noting for now that an answer would lead to improved upper bounds on the energy-constrained quantum and private capacities of the thermal and amplifier channels. At the least, we have proven that the optimal input state for the particular Gaussian channels is such that its reduction to the channel input system is diagonal in the photon number basis.

The rest of our paper is structured as follows. In Section \ref{sec:prelim}, we summarize definitions and prior results relevant to our paper. We provide general upper bounds on the energy-constrained quantum and private capacities of approximately degradable channels in Section~\ref{sec:bounds-approximate-degradable}. We use these tools to establish several upper bounds on the energy-constrained quantum and private capacities of a thermal channel in Sections~\ref{sec:upper-bound-quantum-cap} and  \ref{sec:upper-bound-private-cap}, respectively. A comparision of these different upper bounds on energy-constrained quantum capacity of a thermal channel is discussed in Section~\ref{sec:comparision-of-bounds}. We present an improvement on the achievable rate of private communication through thermal channels, in Section~\ref{sec:lower-bound-private-cap}. We establish bounds on energy-constrained capacities of a quantum amplifier channel and an additive-noise channel in Sections \ref{sec:bounds-amp} and \ref{sec:bound-additive-noise}, respectively. In Section \ref{sec:recent-dev}, we discuss recent developments
from \cite{RMG18}
on the unconstrained quantum capacity of a thermal channel, and we then provide another upper bound on the energy-constrained quantum and private capacities of a thermal channel. We discuss the optimization of the Gaussian energy-constrained generalized channel divergence in Section~\ref{sec:generalized-channel-divergence}.
Finally, we summarize our results and conclude in Section~\ref{sec:conclusion}.

\bigskip

\section{Preliminaries}\label{sec:prelim}
Background on quantum information in infinite-dimensional systems is available in \cite{H13book} (see also \cite{H04,SH08,HS10,HZ12,S15,S15squashed}). In this section, we explain our notations and discuss prior results relevant for our paper. 

\bigskip
\noindent\textbf{Quantum states and channels.}
 Let $\H$ denote a separable Hilbert space, let $\B(\H)$ denote the set of bounded operators acting on $\H$, and let $\P(\H)$ denote the subset of $\B(\H)$ that consists of positive semi-definite operators. Let $\T(\H)$ denote the set of trace-class operators, defined such that their trace norm is finite: $\left\Vert A \right\Vert_1 \equiv \tr\{\left\vert A \right\vert\} < \infty$, where $\left\vert A \right\vert \equiv \sqrt{A^{\dagger}A}$. Let $\D(\H)$ denote the set of density operators (positive semi-definite with unit trace) acting on $\H$. A quantum channel $\N : \T(\H_A) \rightarrow \T(\H_B)$ is a completely positive, trace-preserving linear map. Using the Stinespring dilation theorem \cite{S55}, a quantum channel can be expressed in terms of a linear isometry: i.e., there exists another Hilbert space $\H_E$ and a linear isometry $U:\H_A \rightarrow \H_B \otimes \H_E$ such that for all $\omega_A \in \T(\H_A)$, the following equality holds: $\N(\omega_A) = \tr_E\{U \omega_A U^{\dagger}\}$. A complementary channel $\hat{\N}_{A \rightarrow E}$ of $\N_{A \rightarrow B}$ is defined as $\hat{\N}_{A \to E}= \tr_B\{U \omega_A U^{\dagger}\}$. A quantum channel $\N_{A \rightarrow B}$ is degradable \cite{cmp2005dev} if there exists a quantum channel $\D_{B\rightarrow E}$ such that $\D_{B \rightarrow E}(\N_{A\rightarrow B}(\omega_A)) = \hat{\N}_{A\rightarrow E}(\omega_A)$, for all $\omega_A \in \T(\H_A)$.  

\bigskip
\noindent\textbf{Quantum entropies and information.}
The quantum entropy of a state $\rho \in \D(H)$ is defined as $H(\rho) \equiv -\tr\{\rho\log_2 \rho\}$. It is a non-negative, concave, lower semicontinuous function \cite{W76} and not necessarily finite \cite{BV13}.
The binary entropy function is defined for $x\in[0,1]$ as
\begin{equation}
h_2(x) \equiv -x\log_2x - (1-x)\log_2(1-x).
\end{equation} 
 Throughout the paper we use a function $g(x)$, which is the entropy of a bosonic thermal state with mean photon number $x\geq0$: 
\begin{equation}\label{eq:g-function}
g(x)\equiv(x+1)\log_2(x+1)-x\log_2 x.
\end{equation}
By continuity, we have that $h_2(0) = \lim_{x \to 0} h_2(x) = 0$ and $g(0) = \lim_{x \to 0} g(x) = 0$.
The quantum relative entropy $D(\rho\Vert\sigma)$ of $\rho,\sigma
\in\mathcal{D}(\mathcal{H})$ is defined as \cite{F70,Lindblad1973} 
\begin{equation}D(\rho\Vert\sigma)\equiv \sum_{i} \bra{i}\rho\log_{2}\rho-\rho\log_{2}\sigma + \sigma - \rho \ket{i},
\end{equation}
where $\{|i\rangle\}_{i=1}^{\infty}$ is an orthonormal basis of eigenvectors
of the state $\rho$, if $\operatorname{supp}(\rho)\subseteq\operatorname{supp}%
(\sigma)$ and $D(\rho\Vert\sigma)=\infty$ otherwise. The quantum relative
entropy $D(\rho\Vert\sigma)$ is non-negative for $\rho,\sigma\in
\mathcal{D}(\mathcal{H})$ and is monotone with respect to a quantum channel \cite{Lindblad1975}
$\mathcal{N}:\mathcal{T}(\mathcal{H}_{A})\rightarrow\mathcal{T}(\mathcal{H}%
_{B})$: 
\begin{equation}D(\rho\Vert\sigma)\geq D(\mathcal{N}(\rho)\Vert\mathcal{N}(\sigma)).
\end{equation}
The quantum mutual information $I(A;B)_{\rho}$ of a bipartite state $\rho
_{AB}\in\mathcal{D}(\mathcal{H}_{A}\otimes\mathcal{H}_{B})$ is defined
as~\cite{Lindblad1973} \begin{equation}I(A;B)_{\rho}
\equiv
D(\rho_{AB}\Vert\rho_{A}\otimes\rho_{B}).
\end{equation}
The coherent information $I(A\rangle B)_{\rho}$ of $\rho_{AB}$ is defined as \cite{PhysRevA.54.2629,HS10,K11}
\begin{equation}
I(A\rangle B)_{\rho}\equiv I(A;B)_{\rho}-H(A)_{\rho},
\end{equation}
when $H(A)_{\rho}<\infty$. This expression reduces to
\begin{equation}
I(A\rangle B)_{\rho}=H(B)_{\rho}-H(AB)_{\rho},%
\end{equation}
if $H(B)_{\rho}<\infty$.

\bigskip
\noindent\textbf{Quantum fidelity, trace distance, and diamond distance.}
The fidelity of two quantum states $\rho,\sigma\in\mathcal{D}(\mathcal{H})$ is
defined as \cite{U76} $F(\rho,\sigma)\equiv\left\Vert \sqrt{\rho}\sqrt{\sigma}\right\Vert _{1}^{2}$. The trace distance between two density operators $\rho, \sigma \in \D(\H)$ is equal to $\left\Vert \rho - \sigma \right\Vert_1$. The operational interpretation of trace distance is that it is linearly related to the maximum success probability in distinguishing two quantum states. The diamond norm of a Hermiticity preserving linear map $\S$ is defined as  $\left\Vert \S \right\Vert_{\diamond} \equiv \sup_{\rho_{RA} \in \D(\H_R \otimes \H_A)} \left\Vert (\operatorname{id}_R \otimes \S_{A \rightarrow B})(\rho_{RA}) \right\Vert_1$, where $\operatorname{id}_R$ is the identity map acting on a Hilbert space~$\H_R$ corresponding to an arbitrarily large reference system \cite{K97}. It suffices to optimize with respect to input states $\rho$ that are pure. The diamond-norm distance $\left\Vert \N  - \M \right\Vert_{\diamond}$ is a measure of the distinguishability of two quantum channels $\N$ and $ \M$.

\bigskip
\noindent\textbf{Approximate degradability.}
The concept of approximate degradability was introduced in \cite{sutter2017approx}. The following two definitions of approximate degradability will be useful in our paper.

\begin{definition}
[$\varepsilon$-degradable
\cite{sutter2017approx}]\label{def:eps-approx} A channel $\N_{A\to B}$ is  $\varepsilon$-degradable if there exists a channel $\D_{B \to E}$ such that $\frac{1}{2}\left\Vert \hat{\N} - \D\circ\N \right\Vert_{\diamond} \leq \varepsilon$, where $\hat{\N}$ denotes a complementary channel of $\N$.
\end{definition}

\begin{definition}
	[$\varepsilon$-close-degradable
	\cite{sutter2017approx}]\label{def:eps-close} A channel $\N_{A\to B}$ is  $\varepsilon$-close-degradable if there exists a degradable channel $\M_{A\to B}$ such that $\frac{1}{2}\left\Vert \N - \M \right\Vert_{\diamond} \leq \varepsilon$.
\end{definition}

\begin{remark}
	Let $\N_{A \rightarrow B}$ be a quantum channel that is $\varepsilon$-close-degradable. Then $\N_{A \rightarrow B}$ is $\varepsilon+2\sqrt{\varepsilon}$-degradable by \cite[Proposition~A.5]{sutter2017approx}.
	A converse implication is not known to hold.
\end{remark}

\bigskip

\noindent\textbf{Energy-constrained continuity bounds.}
Next, we recall the definition of an energy observable and a Gibbs observable \cite{H13book,Winter15}. We also review the uniform continuity of conditional quantum entropy with energy constraints \cite{Winter15}. When defining a Gibbs observable, we follow \cite{H13book,Winter15}.
\begin{definition}
	[Energy observable] Let G be a positive semi-definite operator. We assume that it has discrete spectrum and that it is bounded from below. In particular, let $\{\ket{e_k} \}_k$ be an orthonormal basis for a Hilbert space $\H$, and let $\{g_k\}_k$ be a sequence of non-negative real numbers. Then 
	\begin{equation}
	G = \sum_{k=1}^{\infty} g_k \ket{e_k}\bra{e_k}
	\end{equation}
	is a self-adjoint operator that we call an energy observable. 
\end{definition}

\begin{definition}[Extension of energy observable] The nth extension $\overline{G}_n$ of an energy observable G is defined as 	
	\begin{equation}	
	\overline{G}_n = \frac{1}{n}( G\otimes I\otimes\cdots\otimes I + \cdots + I\otimes \cdots \otimes I \otimes G), 	
	\end{equation}
	where $n$ is the number of factors in each tensor product above.
\end{definition}

\begin{definition}
	[Gibbs Observable]\label{def:Gibbs} An energy observable $G$ is a Gibbs observable if for all $\beta > 0$, we have $\tr\{\exp(-\beta G)\} < \infty$, so that the partition function $Z(\beta):=\tr\{\exp(-\beta G)\}$ has a finite value and hence $\exp(-\beta G)/ \tr\{\exp(-\beta G)\}$ is a well defined thermal state. 
\end{definition}
For a Gibbs observable $G$, let us consider a quantum state $\rho$ such that $\tr\{G \rho \} \leq W$. There exists a unique state that maximizes the entropy $H(\rho)$, and this unique maximizer has the Gibbs form $\gamma(W) = \exp(-\beta(W)G)/Z(\beta(W))$, where $\beta(W)$ is the solution of the equation: 
\begin{align}
\tr\{\exp(-\beta G) (G-W)\} = 0. 
\end{align}
In particular, for the Gibbs observable $G = \hbar \omega \hat{n}$, where $\hat{n}=\hat{a}^{\dagger}\hat{a}$ is the photon number operator, a thermal state (mean photon number $\bar{n}$) that saturates the energy constrained inequality $\tr\{G\rho \} \leq W$, gives the maximum value of the entropy:
\begin{equation}
H(\gamma(W)) = g(\bar{n}) = (\bar{n}+1)\log_2(\bar{n}+1) - \bar{n}\log_2\bar{n}~.%
\end{equation}
Here, we have fixed the ground-state energy to be equal to zero. In some parts of our paper, we take the Gibbs observable to be the number operator, and we use the terminology ``mean photon number" and ``energy" interchangeably. 

The following lemma is a uniform continuity bound for the conditional quantum entropy with energy constraints \cite{Winter15}:
\begin{lemma}[Meta-Lemma 17, \cite{Winter15}]\label{thm:cond-entropy}
	{For a Gibbs observable $G \in \P(\H_A)$, and states $\omega_{AB}, \tau_{AB} \in \D({\H_{A}\otimes \H_B})$, such that $\frac{1}{2} \left\Vert \omega_{AB} -\tau_{AB} \right\Vert_1 \leq \varepsilon < \varepsilon' \leq 1$, $\tr\{(G\otimes I_B) \omega_{AB} \}$, $\tr\{(G\otimes I_B)\tau_{AB} \} \leq W$, where $W\in\lbrack0,\infty)$ and $\delta = (\varepsilon' - \varepsilon)/(1+ \varepsilon')$,  the following inequality holds
		\begin{equation}
		\left\vert H(A|B)_{\omega} - H(A|B)_{\tau}\right\vert \leq (2\varepsilon' + 4\delta) H(\gamma(W/\delta)) + g(\varepsilon') + 2h_2(\delta)~. 
		\end{equation}}
\end{lemma}

\bigskip

Throughout the paper, we consider only those quantum channels that satisfy the following finite output entropy condition:
\begin{condition}
	[Finite output entropy]\label{cond:finite-out-entropy}Let $G$ be a Gibbs
	observable and $W\in\lbrack0,\infty)$. A quantum channel $\mathcal{N}$
	satisfies the finite-output entropy condition with respect to $G$ and $W$ if%
	\begin{equation}
	\sup_{\rho\, :\, \operatorname{Tr}\{G\rho\}\leq W}H(\mathcal{N}(\rho))<\infty,
	\label{eq:finiteness-cap}%
	\end{equation}

\end{condition}

\bigskip 

\noindent\textbf{Gaussian states and channels.} We now deliver a brief review of Gaussian states and channels, and we point to
\cite{AS17} for more details. Gaussian
channels model natural physical processes such as photon loss, photon
amplification, thermalizing noise, or random kicks in phase space. They
satisfy Condition~\ref{cond:finite-out-entropy} when the Gibbs observable for
$m$ modes is taken to be%
\begin{equation}
\hat{E}_{m}\equiv\sum_{j=1}^{m}\omega_{j}\hat{a}_{j}^{\dag}\hat{a}_{j},
\label{eq:photon-num-op-freqs}%
\end{equation}
where $\omega_{j}>0$ is the frequency of the $j$th mode and $\hat{a}_{j}$ is
the photon annihilation operator for the $j$th mode, so that $\hat{a}%
_{j}^{\dag}\hat{a}_{j}$ is the photon number operator for the $j$th mode.

 Let%
\begin{equation}
\hat{x}\equiv\left[  \hat{q}_{1},\ldots,\hat{q}_{m},\hat{p}_{1},\ldots,\hat
{p}_{m}\right]  \equiv\left[  \hat{x}_{1},\ldots,\hat{x}_{2m}\right]
\end{equation}
denote a vector of position- and momentum-quadrature operators, satisfying the
canonical commutation relations:%
\begin{equation}
\left[  \hat{x}_{j},\hat{x}_{k}\right]  =i\Omega_{j,k},\quad\text{where}%
\quad\Omega\equiv%
\begin{bmatrix}
0 & 1\\
-1 & 0
\end{bmatrix}
\otimes I_{m},
\end{equation}
and $I_{m}$ denotes the $m\times m$ identity matrix. We take the annihilation
operator for the $j$th mode as $\hat{a}_{j}=(\hat{q}_{j}+i\hat{p}_{j}%
)/\sqrt{2}$. For $\xi\in\mathbb{R}^{2m}$, we define the unitary displacement
operator $D(\xi)\equiv\exp(i\xi^{T}\Omega\hat{x})$. Displacement operators
satisfy the following relation:%
\begin{equation}
D(\xi)^{\dag}D(\xi^{\prime})=D(\xi^{\prime})D(\xi)^{\dag}\exp(i\xi^{T}%
\Omega\xi^{\prime}).
\end{equation}
Every state $\rho\in\mathcal{D}(\mathcal{H})$ has a corresponding Wigner
characteristic function, defined as%
\begin{equation}
\chi_{\rho}(\xi)\equiv\operatorname{Tr}\{D(\xi)\rho\},
\end{equation}
and from which we can obtain the state $\rho$ as%
\begin{equation}
\rho=\frac{1}{\left(  2\pi\right)  ^{m}}\int d^{2m}\xi\ \chi_{\rho}%
(\xi)\ D^{\dag}(\xi).
\end{equation}
A quantum state $\rho$ is Gaussian if its Wigner characteristic function has a
Gaussian form as%
\begin{equation}
\chi_{\rho}(\xi)=\exp\left(  -\frac{1}{4}\left[  \Omega\xi\right]  ^{T}%
V^{\rho}\Omega\xi+\left[  \Omega\mu^{\rho}\right]  ^{T}\xi\right)  ,
\end{equation}
where $\mu^{\rho}$ is the $2m\times1$ mean vector of $\rho$, whose entries are
defined by $\mu_{j}^{\rho}\equiv\left\langle \hat{x}_{j}\right\rangle _{\rho}$
and $V^{\rho}$ is the $2m\times2m$ covariance matrix of $\rho$, whose entries
are defined as%
\begin{equation}
V_{j,k}^{\rho}\equiv\langle\{\hat{x}_{j}-\mu_{j}^{\rho},\hat{x}_{k}-\mu
_{k}^{\rho}\}\rangle.
\end{equation}
The following condition holds for a valid covariance matrix: $V+i\Omega\geq0$,
which is a manifestation of the uncertainty principle \cite{PhysRevA.49.1567}.

A thermal Gaussian state $\theta_{\beta}$ of $m$ modes with respect to
$\hat{E}_{m}$ from \eqref{eq:photon-num-op-freqs}\ and having inverse
temperature $\beta>0$ thus has the following form:%
\begin{equation}
\theta_{\beta}=e^{-\beta\hat{E}_{m}}/\operatorname{Tr}\{e^{-\beta\hat{E}_{m}%
}\}, \label{eq:thermal-E-m-op}%
\end{equation}
and has a mean vector equal to zero and a diagonal $2m\times2m$ covariance
matrix. One can calculate that the photon number in this state is equal to%
\begin{equation}
\sum_{j}\frac{1}{e^{\beta\omega_{j}}-1}.
\end{equation}
A single-mode thermal state with mean photon number $\bar{n} = 1/(e^{\beta \omega} -1)$ has the following representation in the photon number basis:
\begin{equation}\label{eq:thermalstate}
\theta(\bar{n}) \equiv \frac{1}{1+\bar{n}}\sum_{n=0}^{\infty} \left(\frac{\bar{n}}{\bar{n}+1}\right)^n \ket{n}\bra{n}~.
\end{equation}
It is also well known that thermal states can be written as a Gaussian mixture
of displacement operators acting on the vacuum state:
\begin{equation}
\theta_{\beta}=\int d^{2m}\xi\ p(\xi)\ D(\xi)\left[  |0\rangle\langle
0|\right]  ^{\otimes m}D^{\dag}(\xi),
\end{equation}
where $p(\xi)$ is a zero-mean, circularly symmetric Gaussian distribution.
From this, it also follows that randomly displacing a thermal state in such a
way leads to another thermal state of higher temperature:%
\begin{equation}
\theta_{\beta}=\int d^{2m}\xi\ q(\xi)\ D(\xi)\theta_{\beta^{\prime}}D^{\dag
}(\xi), \label{eq:displaced-thermal-is-thermal}%
\end{equation}
where $\beta^{\prime}\geq\beta$ and $q(\xi)$ is a particular circularly
symmetric Gaussian distribution.

In our paper, we employ the two-mode squeezed vacuum state with parameter $\bar{n}$, which is equivalent to a purification of the thermal state in \eqref{eq:thermalstate} and is defined as
\begin{equation}\label{eq:tms}
\ket{\psi_{\operatorname{TMS}}(\bar{n})} \equiv \frac{1}{\sqrt{\bar{n}+1}}\sum_{n=0}^{\infty}\sqrt{\left(\frac{\bar{n}}{\bar{n}+1}\right)^n} \ket{n}_R\ket{n}_A~.
\end{equation}

A $2m\times2m$ matrix $S$ is symplectic if it preserves the symplectic form:
$S\Omega S^{T}=\Omega$. According to Williamson's theorem \cite{W36}, there is
a diagonalization of the covariance matrix $V^{\rho}$ of the form
\begin{equation}
V^{\rho}=S^{\rho}\left(  D^{\rho}\oplus D^{\rho}\right)  \left(  S^{\rho
}\right)  ^{T},
\end{equation}
where $S^{\rho}$ is a symplectic matrix and $D^{\rho}\equiv\operatorname{diag}%
(\nu_{1},\ldots,\nu_{m})$ is a diagonal matrix of symplectic eigenvalues such
that $\nu_{i}\geq1$ for all $i\in\left\{  1,\ldots,m\right\}  $. Computing
this decomposition is equivalent to diagonalizing the matrix $iV^{\rho}\Omega$
\cite[Appendix~A]{WTLB16}.

The entropy $H(\rho)$\ of a quantum Gaussian state $\rho$\ is a direct
function of the symplectic eigenvalues of its covariance matrix $V^{\rho}$
\cite{AS17}:%
\begin{equation}
H(\rho)=\sum_{j=1}^{m}g((\nu_{j}-1)/2),
\end{equation}
where $g(\cdot)$ is defined in \eqref{eq:g-function}.

The Hilbert--Schmidt adjoint of a Gaussian quantum channel $\mathcal{N}_{X,Y}$\ from $l$ modes to $m$ modes
has the following effect on a displacement operator $D(\xi)$ \cite{AS17}:%
\begin{equation}
D(\xi)\mapsto D(X^T\xi)\exp\left(  -\frac{1}{2}\xi^{T}Y\xi+i\xi^{T}\Omega
d\right)  ,
\end{equation}
 where $X$ is a real $2m\times2l$ matrix, $Y$ is a real $2m\times2m$ positive
semi-definite matrix, and $d\in\mathbb{R}^{2m}$, such that they satisfy%
\begin{equation}
Y+i\Omega-iX\Omega X^{T}\geq0.
\end{equation}
The effect of the channel on the mean vector $\mu^{\rho}$ and the covariance
matrix $V^{\rho}$\ is thus as follows:%
\begin{align}
\mu^{\rho}  &  \longmapsto X\mu^{\rho}+d,\\
V^{\rho}  &  \longmapsto XV^{\rho}X^{T}+Y.
\end{align}
A phase-insensitive, single-mode bosonic Gaussian channel adds an equal amount of noise to each quadrature of the electromagnetic field, such that  
\begin{align}
X & = \operatorname{diag}(\sqrt{\tau}, \sqrt{\tau}),\\
Y & = \operatorname{diag}(\nu, \nu),\\
d & = 0, 
\end{align}
where $\tau \in [0, 1]$ corresponds to attenuation, $\tau \geq 1$ amplification, and $\nu$ is the variance of an additive noise. Moreover, the following inequalities should hold
\begin{align}
& \nu \geq 0, \label{eq:cptp-cond-1}\\
& \nu^2 \geq (1 - \tau)^2 \label{eq:cptp-cond}, 
\end{align}
in order for the map to be a legitimate completely positive and trace preserving map. The channel is entanglement breaking
\cite{HSR03}
 if the following inequality holds \cite{Holevo2008}
 \begin{equation}
 \nu \geq \tau + 1.
 \end{equation}

All Gaussian channels are covariant with respect to displacement operators.
That is, the following relation holds%
\begin{equation}
\mathcal{N}_{X,Y}(D(\xi)\rho D^{\dag}(\xi))=D(X\xi)\mathcal{N}_{X,Y}%
(\rho)D^{\dag}(X\xi), \label{eq:covariance-gaussian}%
\end{equation}
and note that $D(X\xi)$ is a tensor product of local displacement operators.

Just as every quantum channel can be implemented as a unitary transformation
on a larger space followed by a partial trace, so can Gaussian channels be
implemented as a Gaussian unitary on a larger space with some extra modes
prepared in the vacuum state, followed by a partial trace \cite{CEGH08}. Given
a Gaussian channel $\mathcal{N}_{X,Y}$\ with $Z$ such that $Y=ZZ^{T}$ we can
find two other matrices $X_{E}$ and $Z_{E}$ such that there is a symplectic
matrix
\begin{equation}
S=%
\begin{bmatrix}
X & Z\\
X_{E} & Z_{E}%
\end{bmatrix}
, \label{eq:gaussian-dilation}%
\end{equation}
which corresponds to the Gaussian unitary transformation on a larger space.
The complementary channel $\mathcal{\hat{N}}_{X_{E},Y_{E}}$\ from input to the
environment then effects the following transformation on mean vectors and
covariance matrices:%
\begin{align}
\mu^{\rho}  &  \longmapsto X_{E}\mu^{\rho},\\
V^{\rho}  &  \longmapsto X_{E}V^{\rho}X_{E}^{T}+Y_{E},
\end{align}
where $Y_{E}\equiv Z_{E}Z_{E}^{T}$.

\bigskip

\noindent\textbf{Quantum thermal channel.} \label{def:thermal channel}
A quantum thermal channel is a Gaussian channel that can be characterized by a beamsplitter of transmissivity $\eta \in (0,1)$, coupling the signal input state with a thermal state with mean photon number $N_B\geq 0$, and followed by a partial trace over the environment. In the Heisenberg picture, the beamsplitter transformation is given by the following Bogoliubov transformation:
\begin{align}
&\hat{b} = \sqrt{\eta} \hat{a} - \sqrt{1-\eta} \hat{e},\\
&\hat{e}' = \sqrt{1-\eta}\hat{a} + \sqrt{\eta} \hat{e},
\end{align}
where $\hat{a}, \hat{b}$, $\hat{e}$, and $\hat{e}'$  are the annihilation operators representing the sender's input mode, the receiver's output mode, an environmental input mode, and an environmental output mode of the channel, respectively. Throughout the paper, we represent the thermal channel by $\L_{\eta, N_B}$. If the mean photon number at the input of a thermal channel is no larger than $N_S$, then the total number of photons that make it through the channel to the receiver is no larger than $\eta N_S + (1-\eta)N_B$.

\bigskip 

\noindent \textbf{Quantum amplifier channel}. A quantum amplifier channel is a Gaussian channel that can be characterized by a two-mode squeezer with parameter $G > 1$, coupling the signal input state with a thermal state with mean photon number $N_B \geq 0$, and followed by a partial trace over the environment. In the Heisenberg picture, the two-mode squeezer implementing a quantum amplifier channel has the following Bogoliubov transformation:
\begin{align}\label{eq:amp-unitary}
&\hat{b} = \sqrt{G} \hat{a} + \sqrt{G-1} {\hat{e}}^{\dagger},\\
&\hat{e}' = \sqrt{G-1} \hat{a}^{\dagger}+ \sqrt{G} \hat{e}~,
\end{align}
where $\hat{a}, \hat{b}$, $\hat{e}$, and $\hat{e}'$ correspond to the same parties as discussed above. Throughout the paper, we represent the noisy amplifier channel by $\A_{G, N_B}$, and the quantum-limited amplifier channel (with $N_B = 0$) by $\A_{G,0}$. 
\bigskip

\noindent \textbf{Additive-noise channel}. An additive-noise channel is specified by the following completely positive and trace preserving map:
\begin{equation}
\N_{\bar{n}}(\rho) \equiv \int d^2 \alpha \ P_{\bar{n}}(\alpha) D(\alpha) \rho D^{\dagger}(\alpha),
\end{equation}
where $P_{\bar{n}} = \exp(-\vert\alpha\vert^2/\bar{n})/(\pi \bar{n}) $ and $D(\alpha)$ is a displacement operator for the input mode. The variance $\bar{n} > 0 $ completely characterizes the channel $\N_{\bar{n}}$, and it roughly represents the number of noise photons added to the input mode by the channel. 
\bigskip

\noindent\textbf{Continuity of output entropy.}
The following theorem on continuity of output entropy for infinite-dimensional systems with finite average energy constraints is a direct consequence of \cite[Theorem~11]{Debbi15} and Lemma \ref{thm:cond-entropy}.

\begin{theorem}\label{thm:continuity-output-entropy}
	Let $\N_{A\to B}$ and $\M_{A\to B}$ be quantum channels, $G\in \P(\H_{B})$ be a Gibbs observable, such that 
		\begin{equation}
		\tr\{\overline{G}_n \N^{\otimes n}(\rho_{A^{n}}) \},\, \tr\{\overline{G}_n\M^{\otimes n}(\rho_{A^{n}}) \} \leq W~,
		\end{equation}
		where $W\in\lbrack0,\infty)$ and $\rho_{RA^n} \in \D(\H_{R} \otimes \H^{\otimes n}_{A})$. If $\frac{1}{2}\left\Vert \N - \M \right\Vert_{\diamond} \leq \varepsilon <\varepsilon' \leq1$ and $\delta = (\varepsilon' - \varepsilon)/(1+ \varepsilon')$, then the following inequality holds
		\begin{multline}
		\left\vert H((\operatorname{id}_R \otimes \N^{\otimes n}_{A\to B} )(\rho_{RA^n})) - H ((\operatorname{id}_R \otimes \M^{\otimes n}_{A\to B}) (\rho_{RA^n})) \right\vert
		\\ \leq n\lbrack(2\varepsilon' + 4\delta) H(\gamma(W/\delta)) + g(\varepsilon')+ 2h_2(\delta)\rbrack.
		\end{multline}
\end{theorem}
\begin{proof}
 Let
 \begin{equation}
 \rho^j = (\operatorname{id}_R \otimes \M^{\otimes j}_{A \to B} \otimes \N^{\otimes (n-j)}_{A\to B} )(\rho_{RA^{n}})~,
 \end{equation}
 and consider the following chain of inequalities:
 \begin{align}
& \left\vert H(RB^n)_{\rho^0} -H(RB^n)_{\rho^n}\right\vert \nonumber \\
& = \left\vert \sum_{j=1}^{n} H(RB^n)_{\rho^{j-1}} - H(RB^n)_{\rho^j} \right\vert\\
 & \leq \sum_{j=1}^{n}\left\vert H(RB^n)_{\rho^{j-1}} - H(RB^n)_{\rho^{j}}\right\vert\\
 &= \sum_{j=1}^{n} \left\vert H(B_j|RB_1\cdots B_{j-1}B_{j+1} \cdots B_n)_{\rho^{j-1}} -H(B_j|RB_1\cdots B_{j-1}B_{j+1}\cdots B_n)_{\rho^j}\right\vert\\
 &\leq n\lbrack(2\varepsilon' + 4\delta) \left(\sum_{j=1}^{n} \frac{1}{n}H(\gamma(W_j/\delta))\right) + g(\varepsilon') + 2h_2(\delta)\rbrack\\
 &\leq n\lbrack(2\varepsilon' + 4\delta) H\! \left(\frac{1}{n}\sum_{j=1}^{n} \gamma(W_j/\delta)\right) + g(\varepsilon') + 2h_2(\delta)\rbrack\\
 &\leq n\lbrack(2\varepsilon' + 4\delta) H\left( \gamma(W/\delta)\right) + g(\varepsilon') + 2h_2(\delta)\rbrack~.
 \end{align}
 The first inequality follows from the triangle inequality. The second equality follows from the fact that the states $\rho^j$ and $\rho^{j-1}$ are the same except for the $j$th output system. Let $W_j$ denote an energy constraint on the $j$th output state of both the channels $\N$ and $\M$, i.e., $\tr\{G \N(\rho_{A_j})\}$, $\tr\{G \M(\rho_{A_j})\} \leq W_j$
 and $\frac{1}{n} \sum_j W_j \leq W$. Then the second inequality follows because $\frac{1}{2}\Vert \rho^j - \rho^{j-1}\Vert_1 \leq \varepsilon$ for the given channels, and we use Lemma~\ref{thm:cond-entropy} for the $j$th output system. The third inequality follows from concavity of entropy. The final inequality follows because
 \begin{equation}
 \tr\left\{ \frac{1}{n}\sum_{j=1}^{n} G~\gamma(W_j/\delta) \right\} = \frac{1}{n}\sum_{j=1}^{n} \tr\{G~\gamma(W_j/\delta)  \} \leq  W/\delta,
 \end{equation}
 and $\gamma(W/\delta)$ is the Gibbs state that maximizes the entropy corresponding to the energy $W/\delta$.
\end{proof}

\bigskip
\noindent\textbf{Continuity of capacities for channels.}
The continuity of various capacities of quantum channels has been discussed in \cite[Lemma~12]{Debbi15}. The general form for the classical, quantum, or private capacity of a channel $\N$ can be defined as $F(\N) = \lim_{n \rightarrow \infty} \frac{1}{n}\sup_{P^{(n)}} f_n(\N^{ \otimes n}, P^{(n)})$, where $\{f_n\}_n$ denotes a family of functions, and $P^{(n)}$ represents states or parameters over which an optimization is performed. Then the following lemma holds \cite{Debbi15}.
\begin{lemma}[Lemma 12, \cite{Debbi15}]\label{thm:continuous-cap}
{If $F(\N) = \lim_{n \rightarrow \infty}  \frac{1}{n} \sup_{P^{(n)}} f_n(\N^{\otimes n}, P^{(n)})$ for a channel $\N$ and $\forall$ $n, P^{(n)}$, $\left\vert  f_n(\N^{\otimes n}, P^{(n)}) - f_n(\M^{\otimes n}, P^{(n)}) \right\vert\leq n c$, then $\left\vert F(\N) - F(\M) \right\vert \leq c$.}
\end{lemma}

\bigskip
\noindent\textbf{Energy-constrained quantum and private capacities.}
The energy-constrained quantum and private capacities of quantum channels have been defined in \cite[Section~III]{MH16}. In what follows, we review the definition of quantum communication and private communication codes, achievable rates, and regularized formulas for energy-constrained quantum and private capacities.

\bigskip

\noindent\textbf{Energy-constrained quantum capacity.}
An $(n,M,G,W,\varepsilon)$ code for
energy-constrained
quantum communication   consists of an encoding channel $\mathcal{E}^{n}:\mathcal{T} (\mathcal{H}_{S})\rightarrow\mathcal{T}(\mathcal{H}_{A}^{\otimes n})$ and a
decoding channel $\mathcal{D}^{n}:\mathcal{T}(\mathcal{H}_{B}^{\otimes n})\rightarrow\mathcal{T}(\mathcal{H}_{S})$, where $M=\dim(\mathcal{H}_{S})$. The energy constraint is such that the following bound holds for all states resulting from the output of the encoding channel $\mathcal{E}^{n}$:
 \begin{equation}
 \tr\{ \overline{G}_n \E^n(\rho_S) \}\leq W~, \label{eq:energy-constraint-q-cap}
 \end{equation}
 where $\rho_S \in \D(\H_S)$. Note that%
 \begin{equation}
 \operatorname{Tr}\left\{  \overline{G}_{n}\mathcal{E}^{n}(\rho_{S})\right\}
 =\operatorname{Tr}\left\{  G\overline{\rho}_{n}\right\}  ,
 \end{equation}
 where%
 \begin{equation}
 \overline{\rho}_{n}\equiv\frac{1}{n}\sum_{i=1}^{n}\operatorname{Tr}%
 _{A^{n}\backslash A_{i}}\{\mathcal{E}^{n}(\rho_{S})\}.
 \end{equation}
 due to the i.i.d.~nature of the observable $\overline{G}_{n}$. Furthermore, the quantum communication code satisfies the following reliability condition such that for all pure states $\phi_{RS}\in\mathcal{D}(\mathcal{H}_{R} \otimes\mathcal{H}_{S})$,
\begin{equation}
F(\phi_{RS},(\operatorname{id}_{R}\otimes\lbrack\mathcal{D}^{n}\circ
\mathcal{N}^{\otimes n}\circ\mathcal{E}^{n}])(\phi_{RS}))\geq1-\varepsilon~,
\label{eq:q-code-fidelity}%
\end{equation}
where $\mathcal{H}_{R}$ is isomorphic to$~\mathcal{H}_{S}$. A rate $R$ is achievable for quantum communication over $\mathcal{N}$ subject
to the energy constraint $W$\ if for all $\varepsilon\in(0,1)$,
$\delta>0$, and sufficiently large $n$, there exists an $(n,2^{n[R-\delta
	]},G,W,\varepsilon)$ energy-constrained quantum communication code. The energy-constrained quantum capacity $Q(\mathcal{N},G,W)$\ of $\mathcal{N}$ is equal to the supremum of all achievable rates.
	
If the channel $\N$ satisfies Condition \ref{cond:finite-out-entropy}
and $G$ is a Gibbs observable, then the quantum capacity $Q(\N, G, W)$ is equal to the regularized energy-constrained coherent information of the channel $\N$ \cite{MH16}
\begin{equation}
Q(\N, G, W) = \lim_{n \rightarrow \infty} \frac{1}{n} I_{c}(\N^{\otimes n}, \overline{G}_n, W),
\label{eq:reg-en-constr-cap-equality}
\end{equation}
where the energy-constrained coherent information of the channel is defined as \cite{MH16}
\begin{equation}
I_c(\N, G, W) \equiv \sup_{\rho: \tr\{\rho G\}\leq W} H(\N(\rho)) - H(\hat{\N}(\rho)),\label{eq:energy-cons-co-inf}
\end{equation} 
and $\hat{\N}$ denotes a complementary channel of $\N$. Note that another definition of energy-constrained quantum communication is possible, but it leads to the same value for the capacity in the asymptotic limit of many channel uses \cite{MH16}.

\bigskip
\noindent\textbf{Energy-constrained private capacity.}
An $(n,M,G,W,\varepsilon)$ code for private communication consists of a set
$\{\rho_{A^{n}}^{m}\}_{m=1}^{M}$\ of quantum states, each in $\mathcal{D}%
(\mathcal{H}_{A}^{\otimes n})$, and a POVM\ $\{\Lambda_{B^{n}}^{m}\}_{m=1}%
^{M}$ such that%
\begin{align}
\operatorname{Tr}\left\{  \overline{G}_{n}\rho_{A^{n}}^{m}\right\}   &  \leq
W,\label{eq:energy-constraint}\\
\operatorname{Tr}\{\Lambda_{B^{n}}^{m}\mathcal{N}^{\otimes n}(\rho_{A^{n}}%
^{m})\}  &  \geq1-\varepsilon,\label{eq:private-good-comm}\\
\frac{1}{2}\left\Vert \hat{\N}^{\otimes n}(\rho_{A^{n}}^{m}%
)-\omega_{E^{n}}\right\Vert _{1}  &  \leq\varepsilon, \label{eq:security-cond}%
\end{align}
for all $m\in\left\{  1,\ldots,M\right\}  $, with $\omega_{E^{n}}$ some fixed
state in $\mathcal{D}(\mathcal{H}_{E}^{\otimes n})$. In the above,
$\hat{\N}$ is a channel complementary to $\N$. 
A rate $R$ is achievable for private communication over $\mathcal{N}$ subject
to energy constraint $W$ if for all $\varepsilon\in(0,1)$, $\delta
>0$, and sufficiently large $n$, there exists an $(n,2^{n[R-\delta
	]},G,W,\varepsilon)$ private communication code. The energy-constrained private capacity
$P(\mathcal{N},G,W)$ of $\mathcal{N}$  is equal
to the supremum of all achievable rates.

An upper bound on the energy-constrained private capacity of a channel has been established in \cite{MH16}, but the lower bound still needs a detailed proof. However, the results in \cite{MH16} suggest the validity of the following form. If the channel $\N$ satisfies Condition~\ref{cond:finite-out-entropy} and $G$ is a Gibbs observable, then the energy-constrained private capacity $P(\mathcal{N},G,W)$ is given by the regularized energy-constrained private information of the channel:
\begin{equation}\label{eq:reg-en-constr-pcap-equality}
P(\N, G, W) = \lim_{n \rightarrow \infty} \frac{1}{n} P^{(1)}(\N^{\otimes n}, \overline{G}_n, W),
\end{equation} 
where the energy-constrained private information is defined as 
\begin{equation}
P^{(1)}(\N, G, W) \equiv \sup_{\bar{\rho}_{\E_A}: \tr\{G\bar{\rho}_{\E_A}\}\leq W}\int dx~ p_X(x)[D(\N(\rho^x_A)\Vert \N(\bar{\rho}_{\E_A}) )- D(\hat{\N}(\rho^x_A)\Vert \hat{\N}(\bar{\rho}_{\E_A})) ],\label{eq:private-information}
\end{equation} 
and $\bar{\rho}_{\E_A} \equiv \int dx~p_X(x) \rho^x_A$ is an average state of the ensemble 
\begin{equation}
\E_{A} \equiv \{p_X(x), \rho^x_A   \},
\end{equation}
and $\hat{\N}$ denotes a complementary channel of $\N$. 
Note that another definition of energy-constrained private communication is possible, but it leads to the same value for the capacity in the asymptotic limit of many channel uses \cite{MH16}.

\begin{remark} \label{rem:unconstrained-cap}
	The unconstrained quantum and private capacities of a quantum channel $\N$ are defined in the same way as above but without the energy constraints demanded in \eqref{eq:energy-constraint-q-cap} and \eqref{eq:energy-constraint}. As a consequence of these definitions and the fact that the set of states with finite but arbitrarily large energy is dense in the set of all states, for channels satisfying the finite output-entropy condition for every energy $W \geq 0$, the unconstrained quantum and private capacities are respectively given by
\begin{equation}
\sup_{W \geq 0} Q(\N, G, W), \qquad \sup_{W \geq 0} P(\N, G, W).
\end{equation}	
\end{remark}

\bigskip
\section{Bounds on energy-constrained quantum  and private capacities  of approximately degradable channels}

\label{sec:bounds-approximate-degradable}

In this section, we derive upper bounds on the energy-constrained quantum and private capacities of approximately degradable channels. We derive these bounds for both $\varepsilon$-degradable (Definition~\ref{def:eps-approx}) and $\varepsilon$-close-degradable (Definition~\ref{def:eps-close}) channels. This general form for the upper bounds on the energy-constrained quantum and private capacities of approximately degradable channels will be directly used in establishing bounds on the capacities of quantum thermal channels. 

We begin by defining the \textit{conditional entropy of degradation}, which will be useful for finding upper bounds on the energy-constrained quantum and private capacities of an $\varepsilon$-degradable channel. A similar quantity has been defined for the finite-dimensional case in \cite{sutter2017approx}.

\begin{definition}[Conditional entropy of degradation] Let $\N_{A\to B}$ and $\D_{B\to E}$ be quantum channels, and let $G \in \P(\H_{A})$ be a Gibbs observable. We define the conditional entropy of degradation as follows:
\begin{equation}\label{eq:Ud-def}
{U_{\D}(\N,G,W) =  \sup_{\rho\, :\, \operatorname{Tr}\{G\rho\}\leq W} \lbrack H(\N(\rho))- H(\D\circ\N(\rho)) \rbrack}~,
\end{equation} 
where $W\in[0,\infty)$.	For a Stinespring dilation $\V: \T(B) \to \T(E) \otimes \T(F)$ of the channel $\D$,
\begin{align}\label{eq:Ud}
U_{\D}(\N,G,W) =  \sup_{\rho\, :\, \operatorname{Tr}\{G\rho\}\leq W} \lbrack H(F\vert E)_{\V\circ\N(\rho)} \rbrack~.
\end{align} 
\end{definition}

We note that the conditional entropy of degradation can be understood as the negative entropy gain of the channel $\D_{B\to E}$ \cite{A04,Holevo2010,Holevo2011,H11ISIT}, with the optimization over input states $\N(\rho)$ restricted to being in the image of $\N$ and obeying the energy constraint $\operatorname{Tr}\{G\rho\}\leq W$.
Next, we show that the conditional entropy of degradation in \eqref{eq:Ud} is additive. 

\begin{lemma}\label{thm:add-Ud}
	Let $\N_{A\to B}$ and $\D_{B\to E}$ be quantum channels,  let $G \in \P(\H_{A})$ be a Gibbs observable, and let $W\in[0,\infty)$. 
	Then for all integer $n\geq 1$
	\begin{equation}
	U_{\D^{\otimes n}}(\N^{\otimes n}, \overline{G}_n, W) = n\lbrack U_{\D}(\N, G, W)\rbrack~.
	\end{equation}
\end{lemma}
\begin{proof}
	The following inequality  
	\begin{equation}
	U_{\D^{\otimes n}}(\N^{\otimes n},\overline{G}_n,W) \geq n \lbrack U_{\D}(\N, G,W)\rbrack
	\end{equation}
	follows trivially because a product input state is a particular state of the form required in the optimization of $U_{\D^{\otimes n}}(\N^{\otimes n},\overline{G}_n,W)$. We now prove the  less trivial inequality 
	\begin{equation}
	U_{\D^{\otimes n}}(\N^{\otimes n},\overline{G}_n,W) \leq n \lbrack U_{\D}(\N, G,W)\rbrack~.
	\end{equation}
	Consider the following chain of inequalities:
	\begin{align}
	H(F^n\vert E^n)_{(\V^{\otimes n}\circ \N^{\otimes n}) (\rho_{A^n})} 
	&\leq  \sum_{i=1}^{n} H(F_i \vert E_i)_{(\V\circ\N)(\rho_{A_i})} \\
	&\leq  n\lbrack  H(F\vert E)_{(\V\circ \N)(\bar{\rho}_n)} \rbrack\\
	& \leq n\lbrack U_{\D}(\N, G,W)\rbrack~,
	\end{align}
	where $\bar{\rho}_n = \frac{1}{n} \sum_{i=1}^{n} \rho_{A_i}$. The first inequality follows from several applications of strong subadditivity \cite{LR73,PhysRevLett.30.434}. The second inequality follows from concavity of conditional entropy
	\cite{LR73,PhysRevLett.30.434}. The last inequality follows because
	$\operatorname{Tr}\{\overline{G}_n \rho_{A^n}\}
	=
	\operatorname{Tr}\{G \bar{\rho}_n\}\leq W
	$  and the conditional entropy of degradation $U_{\D}(\N, G,W)$ involves an optimization over all input states obeying this energy constraint. Since the chain of inequalities is true for all input states $\rho_{A^n}$ satisfying the input energy constraint, the desired result follows. 
\end{proof}

\subsection{Bound on the energy-constrained quantum capacity of an $\varepsilon$-degradable channel}
 
An upper bound on the quantum capacity of an $\varepsilon$-degradable channel was established as \cite[Theorem 3.1(ii)]{sutter2017approx} for the finite-dimensional case. Here, we prove a related bound for the infinite-dimensional case with finite average energy constraints on the input and output states of the channels.

\begin{theorem}\label{thm:qcbound-eps-approx}
	Let $\N_{A\to B}$ be an $\varepsilon$-degradable channel with a degrading channel $\D_{B\to E'}$, and let $G \in \P(\H_{A})$ and $G' \in \P(\H_{E'}) $ be Gibbs observables, such that for all input states $\rho_{A^n} \in \D(H^{\otimes n}_{A})$ satisfying input average energy constraints $\tr\{\overline{G}_n\rho_{A^{n}}\}\leq W$, the following output average energy constraints are satisfied: 
	\begin{equation}
	\tr\{\overline{G}'_n\hat{\N}^{\otimes n}(\rho_{A^{n}}) \}, \quad \tr\{\overline{G}'_n (\D^{\otimes n}\circ\N^{\otimes n})(\rho_{A^{n}})\} \leq W'~,  
	\end{equation}
	where $\hat{\N}_{A \to E}$ is a complementary channel of $\N$ and $E' \simeq E$. Then the energy-constrained quantum capacity $	Q(\N, G, W)$ is bounded from above as
	\begin{equation}
	Q(\N, G, W) \leq U_{\D}(\N,G,W) +(2\varepsilon' + 4\delta) H(\gamma(W'/\delta)) + g(\varepsilon') + 2 h_2(\delta)~,
	\end{equation}
	with $\varepsilon'\in(\varepsilon,1]$, $W, W'\in\lbrack0,\infty)$, and $\delta = (\varepsilon' - \varepsilon)/(1+ \varepsilon')$.
\end{theorem}
\begin{proof}
	Let 
	\begin{align}
	&\sigma_{B^n} = \N^{\otimes n} (\rho_{A^{n}})~,\nonumber\\
	& \rho^j_{E'^jE^{(n-j)}} = (\D^{\otimes j} \circ \N^{\otimes j})\otimes \hat{\N}^{\otimes (n-j)} (\rho_{A^{n}})\nonumber~,
	\end{align}
	and consider the following chain of inequalities:
	\begin{align}
	&
	H(B^n)_{\sigma} - H(E^n)_{\rho^0} \nonumber \\
	& = 
	H(B^n)_{\sigma} -H(E'^n)_{\rho^n} + H(E'^n)_{\rho^n} - H(E^n)_{\rho^0}
	\\
	& \leq U_{\D^{\otimes n}}(\N^{\otimes n}, \overline{G}_n, W) + H(E'^n)_{\rho^n} - H(E^n)_{\rho^0}\\
&= n~U_{\D}(\N, G, W) \nonumber \\
	& \qquad+ \sum_{j=1}^{n} [H(E'_j\vert E'_1\dots E'_{j-1}E_{j+1}\dots E_n)_{\rho^j} - H(E_j\vert E'_1\dots E'_{j-1}E_{j+1}\dots E_n)_{\rho^{j-1}} ] \\
	&  \leq n \lbrack U_{\D}(\N, G, W)+ (2\varepsilon' + 4\delta)  \left(\sum_{j=1}^{n} \frac{1}{n}H(\gamma(W'_j/\delta))\right) + g(\varepsilon') + 2 h_2(\delta) \rbrack\\
	&\leq n\lbrack U_{\D}(\N, G, W)+  (2\varepsilon' + 4\delta)    H\left(\frac{1}{n}\sum_{j=1}^{n} \gamma(W'_j/\delta)\right) + g(\varepsilon') + 2 h_2(\delta) \rbrack\\
	&\leq n\lbrack U_{\D}(\N, G, W)  + (2\varepsilon' + 4\delta)    H\left( \gamma(W'/\delta)\right)+ g(\varepsilon') + 2 h_2(\delta)\rbrack\label{eq:coh-inf-eps-approx}~,
	\end{align}
	The first inequality follows from the definition in \eqref{eq:Ud-def}. The second equality follows from Lemma \ref{thm:add-Ud} and the telescoping technique.
	Let $W'_j$ denote the energy constraint on the $j$th output state of both the channels $\D\circ\N$ and $\hat{\N}$, i.e., $\tr\{G' (\D\circ\N)(\rho_{A_j})\}, \tr\{G'	 \hat{\N}(\rho_{A_j})\} \leq W'_j$
	where $\frac{1}{n} \sum_j W_j' \leq W'$. Then the second inequality holds because $\frac{1}{2}\Vert \rho^j - \rho^{j-1}\Vert_1 \leq \varepsilon$ for the given channels, and we use Lemma \ref{thm:cond-entropy} for the $j$th output system. The third inequality follows from concavity of entropy. The last inequality follows because $\tr\{ \frac{1}{n}\sum_{j=1}^{n} G \gamma(W'_j/\delta) \} = \frac{1}{n}\sum_{j=1}^{n} \tr\{G\gamma(W'_j/\delta)  \}\leq  W'/\delta$, and $\gamma(W'/\delta)$ is the Gibbs state that maximizes the entropy corresponding to the energy $W'/\delta$. Since the chain of inequalities is true for all $\rho_{A^{n}}$ satisfying the input average energy constraint, from \eqref{eq:energy-cons-co-inf} and the above, we get that
	\begin{align}
	\frac{1}{n}I_{c}(\N^{\otimes n}, \overline{G}_n, W) &\leq U_{\D}(\N, G, W)  + (2\varepsilon' + 4\delta)    H\left( \gamma(W'/\delta)\right)+ g(\varepsilon') + 2 h_2(\delta)~.
	\end{align} 
	Since the last inequality holds for all $n$, we obtain the desired result by taking the limit $n \to \infty$ and applying \eqref{eq:reg-en-constr-cap-equality}.
\end{proof}

\bigskip
\subsection{Bound on the energy-constrained quantum capacity of an $\varepsilon$-close-degradable channel}

An upper bound on the quantum capacity of an $\varepsilon$-close-degradable channel was established as \cite[Proposition~A.2(i)]{sutter2017approx} for the finite-dimensional case. Here, we provide a bound for the infinite-dimensional case with finite average energy constraints on the input and output states of the channels.

\begin{theorem}\label{thm:qcbound-eps-close}
	{Let $\N_{A\to B}$ be an $\varepsilon$-close-degradable channel, i.e., $\frac{1}{2} \left\Vert \N -\M \right\Vert_{\diamond} \leq \varepsilon <\varepsilon' \leq 1$, where $\M_{A\to B}$ is a degradable channel.  Let $G \in \P(\H_{A})$, $G' \in \P(\H_{B})$ be Gibbs observables, such that for all input states $\rho_{RA^n} \in \D(\H_R \otimes \H^{\otimes n}_A)$ satisfying the input average energy constraint $\tr\{\overline{G}_n\rho_{A^{n}}\}\leq W$, the following output average energy constraints are satisfied: 
		\begin{equation}
		\tr\{ \overline{G}'_n\N^{\otimes n}(\rho_{A^{n}})\}, \ \tr\{ \overline{G}'_n\M^{\otimes n}(\rho_{A^{n}})\} \leq W'~,
		\end{equation}
		where $W, W'\in\lbrack0,\infty)$. Then
		the energy-constrained quantum capacity $	Q(\N, G, W)$ is bounded from above as
	\begin{equation}
	Q(\N, G, W) \leq I_c(\M, G, W)	+(4\varepsilon' + 8\delta) H(\gamma(W'/\delta)) + 2g(\varepsilon') + 4h_2(\delta)~,
	\end{equation}}
with $\varepsilon'\in(\varepsilon,1]$ and $\delta = (\varepsilon' - \varepsilon)/(1+ \varepsilon')$.
\end{theorem}

\begin{proof}
Let $\omega_{RB^n} = (\operatorname{id}_R \otimes \N^{\otimes n})(\rho_{RA^{n}})$ and $\tau_{RB^n} =  (\operatorname{id}_R \otimes \M^{\otimes n})(\rho_{RA^{n}})$, and consider the following chain of inequalities:
	\begin{align}
 & \!\!\!\! H(B^n)_{\omega} - H(RB^n)_{\omega} - H(B^n)_{\tau} + H(RB^n)_{\tau} \notag \\
 &= H(B^n)_{\omega} -H(B^n)_{\tau}+  H(RB^n)_{\tau} - H(RB^n)_{\omega} \\
	& \leq 2n\lbrack (2\varepsilon' + 4\delta) H( \gamma(W/\delta)) + g(\varepsilon') + 2h_2(\delta)\rbrack~,
	\end{align}
The first inequality follows from  applying Theorem \ref{thm:continuity-output-entropy} twice.
Then from Lemma \ref{thm:continuous-cap},
\begin{equation}
Q(\N, G, W) \leq Q(\M,G,W) + (4\varepsilon' + 8\delta) H(\gamma(W'/\delta)) + 2g(\varepsilon') + 4h_2(\delta).
\end{equation}
The desired result follows from the fact that the energy-constrained quantum capacity of a degradable channel is equal to the energy-constrained coherent information of the channel \cite{MH16}.
\end{proof}

\bigskip

\subsection{Bound on the energy-constrained private capacity of an $\varepsilon$-degradable channel}

In this section, we first derive an upper bound on the private capacity of an $\varepsilon$-degradable channel for the finite-dimensional case, which is different from any of the bounds presented in \cite{sutter2017approx}.  
Then, we generalize this bound to the infinite-dimensional case with finite average energy constraints on the input and output states of the channels. 

\begin{theorem}\label{thm:eps-deg-finite}
Let $\N_{A\to B}$ be a finite-dimensional $\varepsilon$-degradable channel with a degrading channel $\D_{B\to E'}$, and let $\hat{\N}:\T(A)\to \T(E)$ be a complementary channel of $\N$, such that $E' \simeq E$. If
\begin{equation} \label{eq:ud-finite}
U_{\D}(\N) = \max_{\rho \in \D(\H_A)} [H(\N(\rho)) - H((\D\circ\N)(\rho))],
\end{equation}
then the private capacity $P(\N)$ of $\N$ is bounded from above as
\begin{equation}
P(\N) \leq U_{\D}(\N) + 6 \varepsilon \log_2 \operatorname{dim}(\H_E)+ 3g(\varepsilon)~.
\end{equation}
\end{theorem}
\begin{proof}
	Consider Stinespring dilations $\U:\T(A)\to\T(B)\otimes \T(E)$ and $\V:\T(B) \to \T(E')\otimes \T(F)$ of the channel $\N$ and the degrading channel $\D$, respectively. Let $\rho_{XA^n}$ be a classical--quantum state in correspondence with an ensemble $\{p_X(x),\rho^x_{A^n} \}$:
	\begin{equation} \label{eq:cq-state-finite}
\rho_{XA^n} = \sum_{x} p_X(x) \ket{x}\bra{x}_X \otimes \rho^x_{A^n}~,
	\end{equation}
	and let
	\begin{equation}
	\omega_{XE^nE'^nF^n} =\sum_{x} p_X(x)\ket{x}\bra{x}_X \otimes (\operatorname{id}^{\otimes n}_E\otimes \V^{\otimes n})\circ \U^{\otimes n}(\rho^x_{A^n})~.
	\end{equation}
	Consider the following extension of $\omega_{XE^nE'^nF^n}$:
	\begin{equation}
	\sigma_{XYE^nE'^nF^n} = \sum_{x,y} p_X(x)p_{Y\vert X}(y\vert x)\ket{x}\bra{x}_X \otimes \ket{y}\bra{y}_Y \otimes (\operatorname{id}^{\otimes n}_E \otimes \V^{\otimes n})\circ \U^{\otimes n}(\psi^{x,y}_{A^n})~,
	\end{equation}
	where $\psi^{x,y}_{A^n}$ is a pure state, and let $\sigma^{x,y}_{E^nE'^nF^n} = (\operatorname{id}^{\otimes n}_E \otimes \V^{\otimes n})\circ \U^{\otimes n}(\psi^{x,y}_{A^n})$. Consider the following chain of inequalities:
	\begin{align}
	I(X;B^n)_{\omega} - I(X;E^n)_{\omega}& = I(X;F^n\vert E'^n)_{\omega}+ I(X;E'^n)_{\omega} -I(X;E^n)_{\omega}\\
	&= I(X;F^n\vert E'^n)_{\omega} + H(E'^n)_{\omega} - H(E^n)_{\omega}+ H(E^n\vert X)_{\omega} - H(E'^n\vert X)_{\omega}\\
	&\leq I(X;F^n\vert E'^n)_{\omega} + 2n[2\varepsilon \log_2\operatorname{dim}(\H_E)+g(\varepsilon)]\\
	&\leq I(XY;F^n\vert E'^n)_{\sigma}+n[4\varepsilon \log_2\operatorname{dim}(\H_E)+2g(\varepsilon)]\\
	& = 
	H(F^n\vert E'^n)_{\sigma} - H(F^nE'^n|XY)_\sigma + H(E'^n|XY)_\sigma \nonumber \\
	& \qquad +n[4\varepsilon \log_2\operatorname{dim}(\H_E)+2g(\varepsilon)]\\
&=		H(F^n\vert E'^n)_{\sigma} - H(E^n|XY)_\sigma + H(E'^n|XY)_\sigma \nonumber \\
	& \qquad +n[4\varepsilon \log_2\operatorname{dim}(\H_E)+2g(\varepsilon)]\\
	&\leq n[U_{\D}(\N)+6\varepsilon \log_2\operatorname{dim}(\H_E)+3g(\varepsilon)].
    \end{align}
    The first two equalities follow from entropy identities.
    The first inequality follows by  applying the telescoping technique twice and using the continuity result of the conditional quantum entropy for finite-dimensional quantum systems \cite{Winter15}. The second inequality follows from the quantum data processing inequality for conditional quantum mutual information. The last two equalities follow from entropy identities and by using that  $\sigma^{x,y}_{E^nE'^nF^n}$ is a pure state, so that $H(F^nE'^n)_{\sigma^{x,y}} = H(E^n)_{\sigma^{x,y}}$. The last inequality follows from the definition in \eqref{eq:ud-finite}, and additivity of $U_{\D}(\N)$ \cite{sutter2017approx}. Also, we applied the telescoping technique for each $\sigma^{x,y}$ in the summation, and used the continuity result of the conditional quantum entropy for finite-dimensional systems \cite{Winter15}. Since the chain of inequalities is true for any ensemble $\{p_X(x), \rho^x_{A^n}\}$, the final result follows from the definition of private information of the channel, dividing by $n$, taking the limit $n\to\infty$, and noting that the regularized private information is equal to the private capacity of any channel.
\end{proof}

 Next, we derive an upper bound on the energy-constrained private capacity of an $\varepsilon$-degradable channel. 
 
    \begin{theorem}\label{thm:pu2gen}
	Let $\N_{A\to B}$ be an $\varepsilon$-degradable channel with a degrading channel $\D_{B\to E'}$, and let $G \in \P(\H_{A})$, $G' \in \P(\H_{E'}) $ be Gibbs observables, such that for all input states $\rho_{A^n} \in \D(H^{\otimes n}_{A})$ satisfying input average energy constraints $\tr\{\overline{G}_n\rho_{A^{n}}\}\leq W$, the following output average energy constraints are satisfied: 
	\begin{equation}
	\tr\{\overline{G}'_n\hat{\N}^{\otimes n}(\rho_{A^{n}}) \}, \ \tr\{\overline{G}'_n (\D^{\otimes n}\circ\N^{\otimes n})(\rho_{A^{n}})\} \leq W'~,  
	\end{equation}
	where $\hat{\N}_{A \to E}$ is a complementary channel of $\N$, and $E' \simeq E$. Then the energy-constrained private capacity is bounded from above as
	\begin{equation}
	P(\N, G, W) \leq U_{\D}(\N,G,W) +(6\varepsilon' + 12\delta) H(\gamma(W'/\delta)) + 3g(\varepsilon') + 6 h_2(\delta)~,
	\end{equation}
	with $\varepsilon'\in(\varepsilon,1]$, $W, W'\in\lbrack0,\infty)$, and $\delta = (\varepsilon' - \varepsilon)/(1+ \varepsilon')$.
\end{theorem}
\begin{proof}
Since the proof is similar to the above one and previous ones, we just summarize it briefly below.
	 Consider Stinespring dilations $\U:\T(A)\to\T(B)\otimes \T(E)$ and $\V:\T(B) \to \T(E')\otimes \T(F)$ of the channel $\N$ and the degrading channel $\D$, respectively.
Then the action of $\U^{\otimes n}$ followed by $\V^{\otimes n}$ on the  ensemble $\{p_X(x), \rho^x_{A^n}\}$ leads to the following ensemble:
\begin{equation}
\{ p_X(x) , \omega^x_{E^nE'^nF^n} \equiv (\operatorname{id}^{\otimes n}_E\otimes \V^{\otimes n})\circ \U^{\otimes n}(\rho^x_{A^n})\}.
\end{equation}
Similar to the above proof, from applying the telescoping technique three times  and using Lemma~\ref{thm:cond-entropy}, concavity of entropy, and Lemma~\ref{thm:add-Ud}, we get the following bound:
\begin{equation}
I(X;B^n)_{\omega} - I(X;E^n)_{\omega} \leq n [U_{\D}(\N,G,W) +(6\varepsilon' + 12\delta) H(\gamma(W'/\delta)) + 3g(\varepsilon') + 6 h_2(\delta)]~.
\end{equation}
The desired result follows from dividing by $n$,  taking the limit $ n \to \infty$, the definition of the energy-constrained private information of the channel, and using the fact that the regularized energy-constrained private information is an upper bound on the energy-constrained private capacity of a quantum channel \cite{MH16}. 
\end{proof}
\bigskip

\subsection{Bound on the energy-constrained private capacity of an $\varepsilon$-close-degradable channel}

An upper bound on the private capacity of an $\varepsilon$-close-degradable channel was established as \cite[Proposition~A.2(ii)]{sutter2017approx} for the finite-dimensional case. Here, we provide a bound for the infinite-dimensional case with finite average energy constraints on the input and output states of the channels. 
\begin{theorem} \label{thm:eps-close-priv-gen-b}
Let $\N_{A\to B}$ be an $\varepsilon$-close-degradable channel, i.e., $\frac{1}{2} \left\Vert \N -\M \right\Vert_{\diamond} \leq \varepsilon <\varepsilon' \leq 1$, where $\M_{A\to B}$ is a degradable channel.  Let $G \in \P(\H_{A})$, $G' \in \P(\H_{B})$ be Gibbs observables, such that for all input states $\rho_{A^n} \in \D(\H^{\otimes n}_A)$ satisfying input average energy constraints $\tr\{\overline{G}_n\rho_{A^{n}}\}\leq W$, the following output average energy constraints are satisfied: 
\begin{equation}
\tr\{ \overline{G}'_n\N^{\otimes n}(\rho_{A^{n}})\}, \, \tr\{ \overline{G}'_n\M^{\otimes n}(\rho_{A^{n}})\} \leq W'~,
\end{equation}
where $W, W'\in\lbrack0,\infty)$. Then
\begin{equation}
P(\N, G, W) \leq I_c(\M, G, W)	+(8\varepsilon' + 16\delta) H(\gamma(W'/\delta)) + 4g(\varepsilon') + 8h_2(\delta)~,
\end{equation}
with $\varepsilon'\in(\varepsilon,1]$, and $\delta = (\varepsilon' - \varepsilon)/(1+ \varepsilon')$.
\end{theorem}
\begin{proof}
We follow the proof of \cite[Corollary~15]{Debbi15} closely, but incorporate energy constraints.
Consider Stinespring dilations $\U: \T(A) \to \T(B)\otimes \T(E)$ and $\V:\T(A) \to \T(B)\otimes \T(E)$ of the channels $\N$ and $\M$, respectively. Consider an input ensemble $\{p_X(x), \rho^x_{A^n}\}$, which leads to the output ensembles 
\begin{align}
&\{ p_X(x), \omega^x  \equiv \U^{\otimes n} (\rho^x_{A^n})\}, \\
&\{ p_X(x), \tau^x  \equiv \ \V^{\otimes n} (\rho^x_{A^n})\}.
\end{align}
Supposing at first that the index $x$ is discrete, from four times applying Theorem~\ref{thm:continuity-output-entropy} and employing the same expansions as in the proof of \cite[Corollary~15]{Debbi15} , we get
\begin{equation}
I(X;B^n)_{\omega} - I (X;E^n)_{\omega} - [I(X;B^n)_{\tau} - I(X;E^n)_{\tau}]\leq 4 n [ (2\varepsilon' + 4\delta) H\left( \gamma(W/\delta)\right) + g(\varepsilon') + 2h_2(\delta)].
\end{equation}
The upper bound is uniform and has no dependence on the particular ensemble except via the energy constraints. Thus, by approximation, the same bound applies to ensembles for which the index $x$ is continuous.
Then from Lemma \ref{thm:continuous-cap}, we find that
\begin{align}
P(\N, G, W) &\leq P(\M, G, W) + (8\varepsilon' + 16\delta) H(\gamma(W'/\delta)) + 4g(\varepsilon') + 8h_2(\delta)\\
& = I_c(\M, G, W)+ (8\varepsilon' + 16\delta) H(\gamma(W'/\delta)) + 4g(\varepsilon') + 8h_2(\delta).
\end{align}
The equality in the last line follows from the fact that the energy-constrained private capacity of a degradable channel is equal to the energy-constrained coherent information of the channel \cite{MH16}.
\end{proof}

\bigskip

\section{Upper bounds on energy-constrained quantum capacity of bosonic thermal channels}

 \label{sec:upper-bound-quantum-cap}

In this section, we establish three different upper bounds on the energy-constrained quantum capacity of a thermal channel:
\begin{enumerate}
\item  We establish a first upper bound using the theorem that any thermal channel can be decomposed as the concatenation of a pure-loss channel followed by a quantum-limited amplifier channel \cite{dp2006,dp2012}. We call this bound the data-processing bound and denote it by $Q_{U_1}$.

\item Next, we show that a thermal channel is an $\varepsilon$-degradable channel for a particular choice of degrading channel. Then an upper bound on the energy-constrained quantum capacity of a thermal channel directly follows from Theorem \ref{thm:qcbound-eps-approx}. We call this bound the $\varepsilon$-degradable bound and denote it by $Q_{U_2}$.

\item We establish a third upper bound on the energy-constrained quantum capacity of a thermal channel using the idea of $\varepsilon$-close-degradability. We show that the thermal channel is $\varepsilon$-close to a pure-loss bosonic channel for a particular choice of $\varepsilon$. Since a pure-loss bosonic channel is a degradable channel \cite{WPG07}, the bound on the energy-constrained quantum capacity of a thermal channel follows directly from Theorem \ref{thm:qcbound-eps-close}. We call this bound the $\varepsilon$-close-degradable bound and denote it by $Q_{U_3}$.
 \end{enumerate}
\noindent
In Section \ref{sec:comparision-of-bounds},
we compare,  for different parameter regimes, the closeness of these upper bounds with a known lower bound on the quantum capacity of thermal channels.

\subsection{Data-processing bound on the energy-constrained quantum capacity of bosonic thermal channels}

\label{sec:data-proc}

 In this section, we provide an upper bound using the theorem that any thermal channel $\L_{\eta, N_B}$ can be decomposed as the concatenation of a pure-loss channel $\L_{\eta', 0}$ with transmissivity $\eta'$ followed by a quantum-limited amplifier channel $\A_{G,0}$ with gain $G$ \cite{dp2006,dp2012}, i.e., 
\begin{equation}
\L_{\eta, N_B} = \A_{G,0} \circ \L_{\eta', 0}, \label{eq:decomp-gaussian}% 
\end{equation}
where $G= (1-\eta) N_B + 1$, and $\eta' = \eta/G$. In Theorem \ref{thm:1.45bits}, we prove that the data-processing bound can be at most  1.45 bits larger than a known lower bound. 
\begin{theorem}\label{thm:qu1}
	An upper bound on the quantum capacity of a thermal channel $\L_{\eta, N_B}$ with transmissivity $\eta\in\lbrack1/2,1\rbrack$, environment photon number $N_B$, and input mean photon number constraint $N_S$  is given by
	\begin{align} 
	Q(\L_{\eta, N_B}, N_S )  &\leq    \max\{0,
	Q_{U_1}(\L_{\eta, N_B}, N_S )\}~,\\
	Q_{U_1}(\L_{\eta, N_B}, N_S ) &\equiv g(\eta' N_{S})-g[(1-\eta')N_{S}], \label{eq:qu1}
	\end{align}
	with $\eta' = \eta/((1-\eta)N_B+1)$.
\end{theorem}
\begin{proof}
	An upper bound on the energy-constrained quantum capacity can be established by using \eqref{eq:decomp-gaussian} and a data-processing argument. We find that
	\begin{align}
	Q(\L_{\eta, N_B}, N_S )  &  =Q(\mathcal{\mathcal{A}}_{G}\circ
	\L_{\eta', 0},N_S)\\
	&  \leq Q(\L_{\eta', 0},N_S)\\
	&  =\max\{0,g(\eta' N_{S})-g[(1-\eta')N_{S}]\} ~. 
	\end{align}
	The first inequality follows from definitions and data processing---the energy-constrained capacity of $\mathcal{\mathcal{A}}_{G}\circ
	\L_{\eta', 0}$ cannot exceed that of
	$\L_{\eta', 0}$. The second equality follows from the formula for the energy-constrained quantum capacity of a pure-loss bosonic channel with transmissivity $\eta'$ and input mean photon number $N_S$ \cite{Mark2012tradeoff,MH16}. 
\end{proof}

\begin{remark} \label{rem:unconstrained-qcap-thermal}
	Applying Remark \ref{rem:unconstrained-cap}, we find the following data-processing bound
	$Q_{U_1}(\L_{\eta, N_B})$ on the unconstrained quantum capacity of bosonic thermal channels: 
		\begin{align}
		Q(\L_{\eta, N_B}) \leq Q_{U_1}(\L_{\eta, N_B}) &= \sup_{N_S: N_S\in[0,\infty]}Q_{U_1}(\L_{\eta, N_B}, N_S )\\
	&=\lim_{N_S \to \infty} Q_{U_1}(\L_{\eta, N_B}, N_S) \\
	&= \log_{2}(\eta/(1-\eta)) - \log_2(N_B+1),\label{eq:unconstrained-qcap-thermal}
	\end{align}
	where the second equality follows from the monotonicity of $g(\eta N_S) -g[(1-\eta)N_S]$ with respect to $N_S$ for $\eta\geq 1/2$ \cite{GSE08}.
	
	The bound 
	\begin{equation}
	Q(\L_{\eta, N_B}, N_S ) \leq -\log_2([1-\eta] \eta^{N_B}) - g(N_B)
	\label{eq:thermal-bnd-PLOBWTB}
	\end{equation}
	 was found in  \cite{PLOB17,WTB17}. Moreover, the following bound was established quite recently in \cite[Eq.~(40)]{RMG18}:
	 \begin{align}
	 Q(\L_{\eta, N_B}, N_S) \leq \max\left\{0, \log_{2}\frac{\eta - N(1-\eta)}{(1+N)(1-\eta)}  \right\}~. \label{eq:thermal-bnd-RMG}
	 \end{align}
	 As discussed in \cite{RMG18}, a comparison of \eqref{eq:unconstrained-qcap-thermal} with the bounds from \eqref{eq:thermal-bnd-PLOBWTB} and \eqref{eq:thermal-bnd-RMG} leads to the conclusion that the bound given in \eqref{eq:thermal-bnd-RMG} is always tighter than \eqref{eq:unconstrained-qcap-thermal}. However, \eqref{eq:unconstrained-qcap-thermal} and the bound in \eqref{eq:thermal-bnd-PLOBWTB} are incomparable as one is better than the other for certain parameter regimes. Also, \eqref{eq:thermal-bnd-PLOBWTB} is tighter than \eqref{eq:thermal-bnd-RMG} for certain parameter regimes. 
	 
	 We note that the upper bound in \eqref{eq:thermal-bnd-RMG} was indepedently established in \cite{NAJ18}.
 \end{remark}

\begin{remark}
The data-processing bound $Q_{U_1}(\L_{\eta, N_B},N_S)$ on the energy-constrained quantum capacity
$Q(\L_{\eta, N_B},N_S)$
 places a strong restriction on the channel parameters $\eta$ and $N_B$. Since the quantum capacity of a pure-loss channel with transmissivity $\eta'$ is non-zero only for $\eta' > 1/2$, the energy-constrained quantum capacity $Q(\L_{\eta, N_B},N_S)$ is non-zero only for 
\begin{equation}\label{eq:ent-break-thermal}
1\geq \eta > \frac{N_B+1}{N_B+2}.
\end{equation}
However, \cite[Section~4]{dp2006} provides a stronger restriction on $\eta$ and $N_B$  than \eqref{eq:ent-break-thermal} does. 
\end{remark}

\subsection{$\varepsilon$-degradable bound on the energy-constrained quantum capacity of bosonic thermal channels}\label{sec:eps-deg-q-cap}

In this section, we provide an upper bound on the energy-constrained quantum capacity of a thermal channel using the idea of $\varepsilon$-degradability. In Theorem \ref{thm:qcbound-eps-approx}, we established a general upper bound on the energy-constrained quantum capacity of an $\varepsilon$-degradable channel. Hence, our first step is to construct the degrading channel $\D$ given in \eqref{eq:deg-map}, such that the concatenation of a thermal channel $\L_{\eta, N_B}$ followed by $\D$ is close in diamond distance to the complementary channel $\hat{\L}_{\eta, N_B}$ of the thermal channel $\L_{\eta, N_B}$.

We start by motivating the reason for choosing the particular degrading channel in \eqref{eq:deg-map}, 
which is depicted in Figure~\ref{fig:beamsplitter-transformations}, and then we find an upper bound on the diamond distance between $\D\circ\L_{\eta,N_B}$ and $\hat{\L}_{\eta, N_B}$. In general, it is computationally hard to perform the optimization over an infinite dimensional space required in the calculation of the diamond distance between Gaussian channels. However, we address this problem in this particular case  by introducing a channel that simulates the serial concatenation of the thermal channel and the degrading channel, and we call it the simulating channel, as given in \eqref{eq:simulating-ch}. This allows us to bound the diamond distance between the channels from above by the trace distance between the environment states of the complementary channel and the simulating channel (Theorem~\ref{thm:eps-diamond}). Next, we argue that, for a given input mean photon-number constraint $N_S$, a thermal state with mean photon number $N_S$ maximizes the conditional entropy of degradation defined in \eqref{eq:Ud}, which also appears in the general upper bound established in Theorem~\ref{thm:qcbound-eps-approx}. We finally provide an upper bound on the energy-constrained quantum capacity of a thermal channel by using all these tools and invoking Thereom \ref{thm:qcbound-eps-approx}.

We now establish an upper bound on the diamond distance between the complementary channel of the thermal channel and the concatenation of the thermal channel followed by a particular degrading channel. 
Let $\B$ and $\B'$ represent beamsplitter transformations with transmissivity $\eta$ and $(1-\eta)/\eta$, respectively. In the Heisenberg picture, the beamsplitter transformation $\B_{C_1 D_1 \to C_2 D_2}$ is given by
\begin{align}
&\hat{c}_2 = \sqrt{\eta} \hat{c}_1 -\sqrt{1-\eta} \hat{d}_1,\\
&\hat{d}_2 = \sqrt{1-\eta}\hat{c}_1+\sqrt{\eta} \hat{d}_1~.
\end{align}
Similarly, the beamsplitter transformation $\B'_{C_1 D_1 \to C_2 D_2}$ is given by
\begin{align}
&\hat{c}_2 = \sqrt{(1-\eta)/\eta} \hat{c}_1 +\sqrt{(2\eta-1)/\eta} \hat{d}_1,\label{eq:B'1}\\
&\hat{d}_2 =-\sqrt{(2\eta-1)/\eta}\hat{c}_1 +\sqrt{(1-\eta)/\eta} \hat{d}_1 \label{eq:B'2}~,
\end{align}
where $\hat{c}_1, \hat{c}_2$, $\hat{d}_1$, and $\hat{d}_2$  are annihilation operators representing various modes involved in the beamsplitter transformations. Here, $\eta \in [1/2,1]$. It is important to stress that there is a difference in phase between $\B$ and $\B'$ beamsplitter transformations, which is crucial in our development. 

\begin{figure*}
	\begin{center}
		{\includegraphics[width=.80\columnwidth]{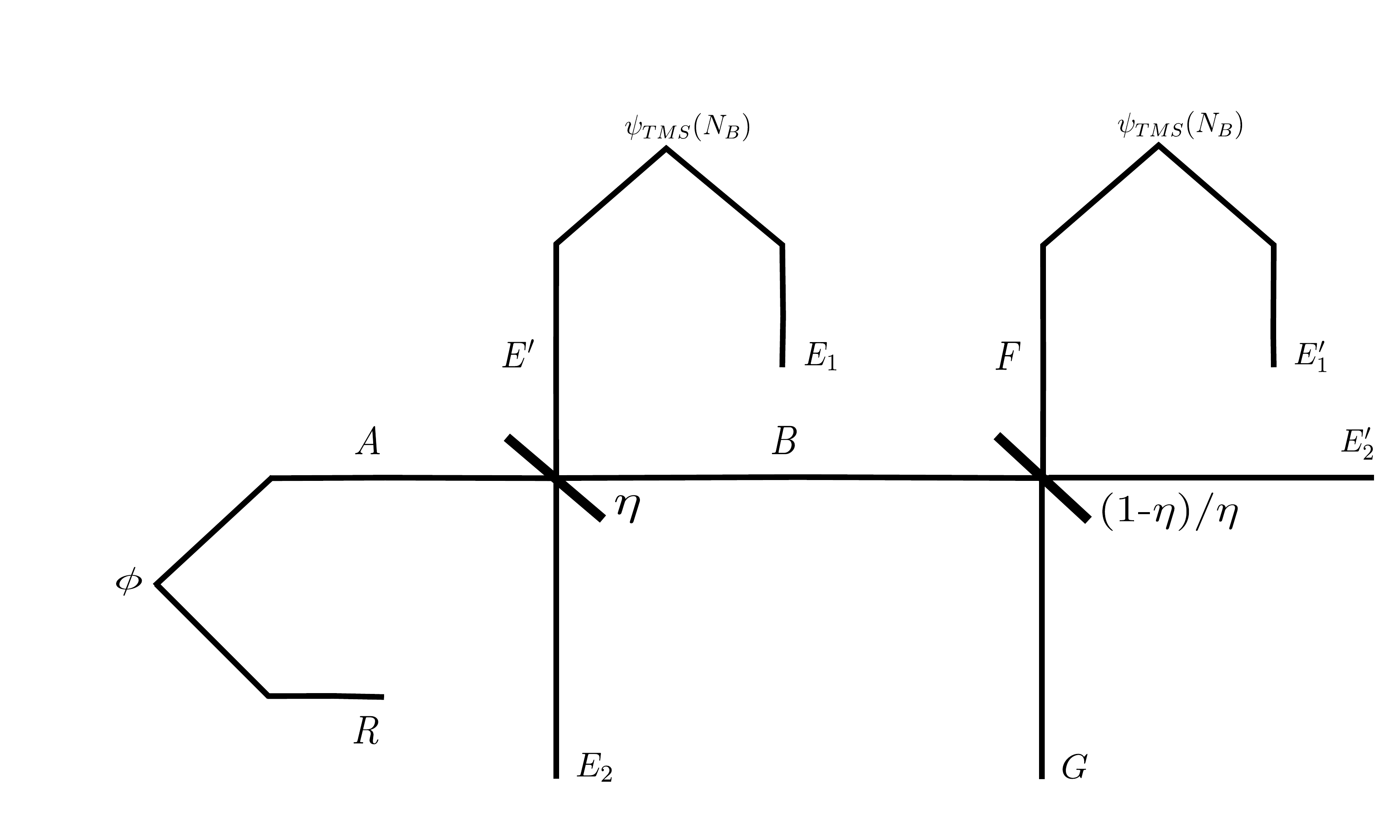}}
	\end{center}
	\caption{The figure plots a thermal channel with transmissivity $\eta\in[1/2,1]$ and a degrading channel as described in \eqref{eq:deg-map}. $\phi_{RA}$ is an input state to the beamsplitter $\B$ with transmissivity $\eta$ and $\psi_{\operatorname{TMS}}(N_B)$ represents a two-mode squeezed vacuum state with parameter $N_B$.  System $B$ is the output of the thermal channel, and systems $E_1E_2$ are the  outputs of the complementary channel. The second beamsplitter $\B'$ has transmissivity $(1-\eta)/\eta$, and system $B$ acts as an input to $\B'$. Systems $E'_1E'_2$ represent the output systems of the degrading channel, whose action is to tensor in the state $\psi_{\operatorname{TMS}}(N_B)_{F E_1'}$, interact the input system $B$ with $F$ according to $\B'$, and then trace over system $G$.  }%
	\label{fig:beamsplitter-transformations}%
\end{figure*}

Consider the following action of the thermal channel $\L_{\eta, N_B}$ on an input state $\phi_{RA}$:
\begin{equation}
\label{eq:th-ch}
(\operatorname{id}_R\otimes \L_{\eta, N_B})(\phi_{RA}) = \tr_{E_1E_2}\{ \B_{AE'\to BE_2}(\phi_{RA}\otimes \psi_{\operatorname{TMS}}(N_B)_{E'E_1})  \}~, 
\end{equation}
where $R$ is a reference system and $\psi_{\operatorname{TMS}}(N_B)_{E'E_1}$ is a two-mode squeezed vacuum state with parameter $N_B$, as defined in \eqref{eq:tms}.

Here and what remains in the proof, we consider the action of various transformations on the covariance matrices of the states involved, and we furthermore track only the submatrices corresponding to the position-quadrature operators of the covariance matrices. It suffices to do so because all channels involved in our discussion are phase-insensitive Gaussian channels.

The submatrix corresponding to the position-quadrature operators of the covariance matrix of $\psi_{\operatorname{TMS}}(N_B)_{E'E_1}$ has the following form:
\begin{equation}\label{eq:tms-cov}
V = %
\begin{bmatrix}
2 N_B+1 & 2\sqrt{N_B(1+N_B)} \\
2\sqrt{N_B(1+N_B)} & 2N_B+1 
\end{bmatrix}
~.\end{equation}
The action of a complementary channel $\hat{\L}_{\eta, N_B}$ on an input state $\phi_{RA}$ is given by
\begin{equation}
\label{eq:comp-ch}
(\operatorname{id}_R \otimes \hat{\L}_{\eta,N_B})(\phi_{RA}) = \tr_{B}\{ \B_{AE'\to BE_2}(\phi_{RA}\otimes \psi_{\operatorname{TMS}}(N_B)_{E'E_1})  \}~.
\end{equation}
 It can be understood from Figure~\ref{fig:beamsplitter-transformations} that the system $R$ is correlated with the input system $A$ for the channel, and the system $E'$ is the environment's input. The beamsplitter transformation $\B$ then leads to systems $B$ and $E_2$. Hence, the output of the thermal channel $\L_{\eta, N_B}$ is system $B$, and the outputs of the complementary channel $\hat{\L}_{\eta, N_B}$ are systems $E_1$ and $E_2$. 
 
Our aim is to introduce a degrading channel $\D$, such that the combined state of $R$ and the output of $\D\circ\L_{\eta, N_B}$ emulate the combined state of $R, E_1$, and $E_2$, to an extent. This will then allow us to bound the diamond distance between $\D\circ\L_{\eta, N_B}$ and $\hat{\L}_{\eta, N_B}$ from above. For the case when there is no thermal noise, i.e., $N_B = 0 $, a thermal channel reduces to a pure-loss channel. Moreover, we know that a pure-loss channel is a degradable channel and the corresponding degrading channel can be realized by a beamsplitter with transmissivity $(1-\eta)/\eta$ \cite{GSE08}. Hence, we consider a degrading channel, such that it also satisfies the conditions for the above described special case.
 
  Consider a beamsplitter with transmissivity $(1-\eta)/\eta$ and the beamsplitter transformation  $\B'$ from \eqref{eq:B'1}-\eqref{eq:B'2}. As described in Figure \ref{fig:beamsplitter-transformations}, the output $B$ of the thermal channel $\L_{\eta, N_B}$ becomes an input to the beamsplitter $\B'$. We consider one mode ($F$ in Figure \ref{fig:beamsplitter-transformations}) of the two-mode squeezed vacuum state $\psi_{\operatorname{TMS}}(N_B)_{FE'_1}$ as an environmental input for $\B'$, so that the subsystem $E'_1$ mimics $E_1$. Hence, our choice of degrading channel seems reasonable, as the combined state of system $R$ and output systems $E'_1$, $E'_2$ of $\D\circ\L_{\eta, N_B}$ emulates the combined state of $R, E_1$, and $E_2$, to an extent. We suspect that our choice of degrading channel is a good choice because an upper bound on the energy-constrained quantum capacity of a thermal channel using this technique outperforms all other upper bounds for certain parameter regimes. 
  We denote our choice of degrading channel by $\D_{(1-\eta)/\eta,N_B}:\T(B) \to \T(E'_1)\otimes \T(E'_2)$. More formally,  $\D_{(1-\eta)/\eta,N_B}$ has the following action on the output state  $\L_{\eta,N_B}(\phi_{RA})$:
\begin{equation}
\label{eq:deg-map}
(\operatorname{id}_R \otimes[\D_{(1-\eta)/\eta,N_B}\circ\L_{\eta,N_B}])(\phi_{RA}) = \tr_{G}\{\B'_{BF\to E'_2G}(\L_{\eta,N_B}(\phi_{RA})\otimes \psi_{\operatorname{TMS}}(N_B)_{FE'_1})  )  \}~.
\end{equation}

Next, we provide a strategy to bound the diamond distance between $\D_{(1-\eta)/\eta,N_B}\circ \L_{\eta, N_B}$ and $\hat{\L}_{\eta, N_B}$. Consider the following submatrix corresponding to the position-quadrature operators of the covariance matrix of an input state $\phi_{RA}$:
\begin{equation}
\gamma=%
\begin{bmatrix}
a & c\\
c & b%
\end{bmatrix}
.\label{eq:input-cov}%
\end{equation}
where $a,b,c\in\mathbb{R}$ are such that the above is the position-quadrature part of a legitimate covariance matrix.
Let $\xi_{RE'_2E_2E_1GE'_1}$ denote the state after the beamsplitter transformations act on an input state $\phi_{RA}$:
\begin{equation}\label{eq:full-dilation-state}
\xi_{RE'_2E_2E_1GE'_1} = \B'_{BF\to E'_2G}\lbrack\B_{AE'\to BE_2}[\phi_{RA}\otimes \psi_{\operatorname{TMS}}(N_B)_{E'E_1}] \otimes \psi_{\operatorname{TMS}}(N_B)_{FE'_1})\rbrack~.
\end{equation}  
Then the submatrix corresponding to the position-quadrature operators of the covariance matrix of the output state in \eqref{eq:deg-map} is given by \cite{Mathematica}:
\begin{equation}\label{eq:deg-cov}
\gamma' = %
\begin{bmatrix}
a & c\sqrt{1-\eta} & 0\\
c\sqrt{1-\eta} & b +\eta(1-b+2 N_B) & 2 \sqrt{N_B(1+N_B)(2-1/\eta)}\\
0 & 2 \sqrt{N_B(1+N_B)(2-1/\eta)} & 2 N_B+1
\end{bmatrix}~.
\end{equation}

\begin{figure*}
	\begin{center}
		{\includegraphics[width=.80\columnwidth]{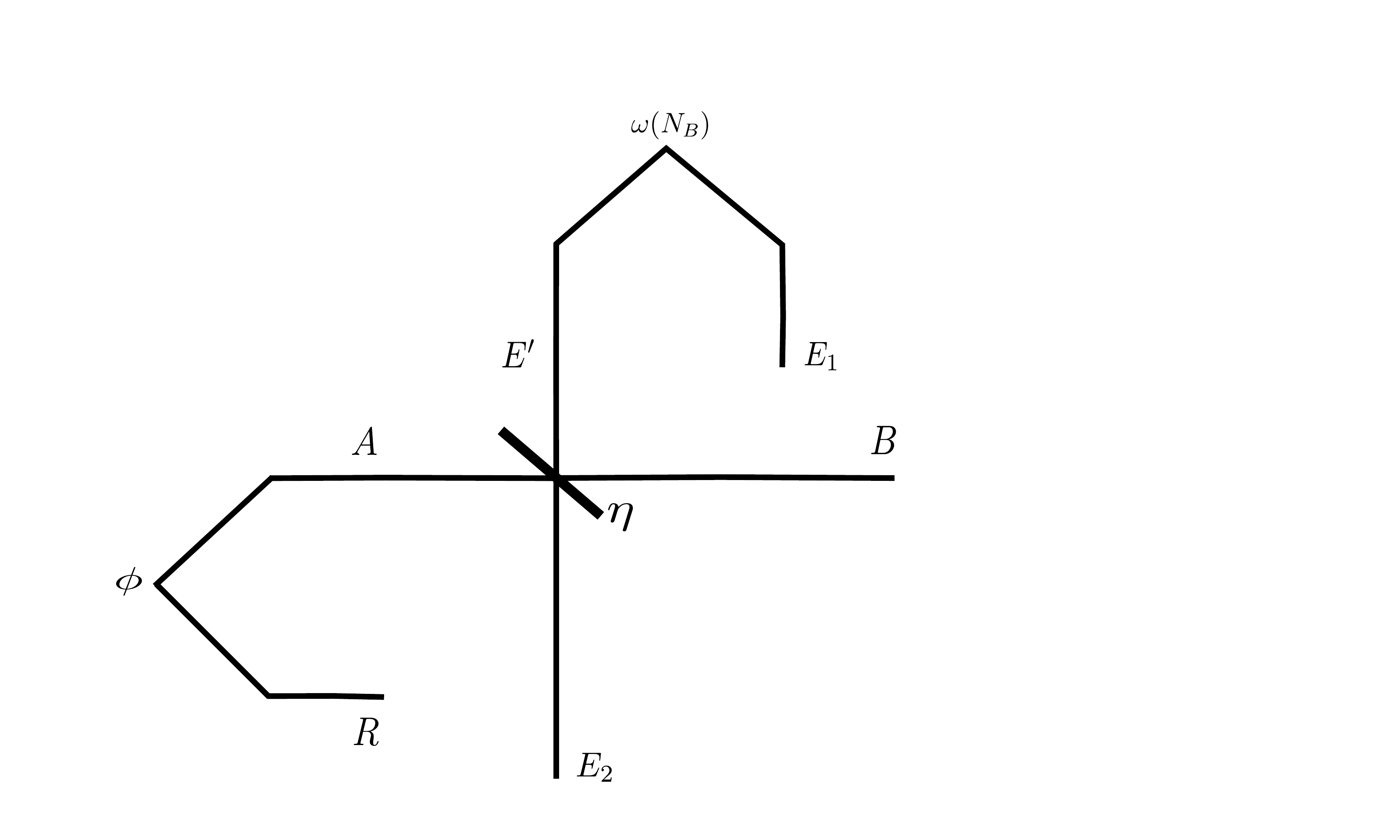}}
	\end{center}
	\caption{The figure plots the simulating channel $\Xi$ described in \eqref{eq:simulating-ch}. $\phi_{RA}$ is an input state to a beamsplitter $\B$ with transmissivity $\eta$ and $\omega(N_B)$ represents a noisy version of a two-mode squeezed vacuum state with parameter $N_B$ (see \eqref{eq:noisy-tms-cov}), one mode of which is an input to the environment mode of the beamsplitter. The simulating channel is such that system $B$ is traced over, so that the channel outputs are $E_1$ and $E_2$. Finally, the simulating channel is exactly the same as the channel from system $A$ to systems  $E'_1E'_2$ in Figure \ref{fig:beamsplitter-transformations}.
	}%
	\label{fig:simulating-channel}%
\end{figure*}

Now, we introduce a particular channel that simulates the action of $\D_{(1-\eta)/\eta,N_B}\circ \L_{\eta, N_B}$ on an input state $\phi_{RA}$. We denote this channel by $\Xi$, and it has the following action on an input state $\phi_{RA}$:
\begin{equation}\label{eq:simulating-ch}
(\operatorname{id}_R \otimes \Xi)(\phi_{RA}) = \tr_{B}\{\B_{AE' \rightarrow BE_2} (\phi_{RA}\otimes \omega(N_B)_{E'E_1}) \}~,
\end{equation}
where $\omega(N_B)_{E'E_1}$ represents a noisy version of a two-mode squeezed vacuum state with parameter $N_B$ and has the following submatrix corresponding to the position-quadrature operators of the covariance matrix:
\begin{equation}\label{eq:noisy-tms-cov}
V' = %
\begin{bmatrix}
2 N_B+1 & 2\sqrt{[N_B(1+N_B)(2\eta-1)]/\eta^2} \\
2\sqrt{[N_B(1+N_B)(2\eta-1)]/\eta^2} & 2N_B+1 
\end{bmatrix}
~.\end{equation}
The matrix $V'$ in \eqref{eq:noisy-tms-cov} is a well defined submatrix of the covariance matrix for the noisy version of a two-mode squeezed vacuum state, because $(2\eta-1)/\eta^2 \in [0,1]$ for $\eta \in [1/2,1]$.
The submatrix of the covariance matrix corresponding to the state in \eqref{eq:simulating-ch} is the same as the submatrix  in \eqref{eq:deg-cov} \cite{Mathematica}. In other words, the covariance matrix for the sytems $R$, $E'_1$, and $E'_2$ in Figure~\ref{fig:beamsplitter-transformations} is exactly the same as the covariance matrix for the systems $R$, $E_1$, and $E_2$ in Figure~\ref{fig:simulating-channel}.  This equality of covariance matrices is sufficient to conclude that the following equivalence holds for any quantum input state $\phi_{RA}$ (see \cite[Chapter 5]{AS17} for a proof):
\begin{equation}\label{eq:equiv}
(\operatorname{id}_R \otimes [ \D_{(1-\eta)/\eta,N_B}\circ\L_{\eta, N_B}])(\phi_{RA}) = (\operatorname{id}_R \otimes \Xi)(\phi_{RA})~.
\end{equation}
Thus, the channels $\D_{(1-\eta)/\eta,N_B}\circ\L_{\eta, N_B}$ and $\Xi$ are indeed the same. 

From \eqref{eq:comp-ch}, \eqref{eq:simulating-ch}, and \eqref{eq:equiv}, the action of both $\hat{\L}_{\eta,N_B}$ and $\Xi$ can be understood as tensoring the state of the environment with the input state of the channel, performing the beamsplitter transformation $\B$, and then tracing out the output of the channels.
Using these techniques, we now establish an upper bound on the diamond distance between the complementary channel in \eqref{eq:comp-ch} and the concatenation of the thermal channel followed by the degrading channel in \eqref{eq:deg-map}. 

\begin{theorem}\label{thm:eps-diamond}
	Fix $\eta \in [1/2,1]$. Let $\L_{\eta, N_B}$ be a thermal channel with transmissivity $\eta$, and let $\D_{(1-\eta)/\eta,N_B}$ be a degrading channel as defined in \eqref{eq:deg-map}. Then
	\begin{equation}
	\frac{1}{2}\left\Vert \hat{\L}_{\eta, N_B} - \D_{(1-\eta)/\eta,N_B}\circ\L_{\eta, N_B}\right\Vert_{\diamond} \leq \sqrt{1-\eta^2/\kappa(\eta,N_B)}~,
	\end{equation}     
	with 
	\begin{equation}
	{\kappa(\eta,N_B) = \eta^2 + N_B(N_B+1)\lbrack1+3 \eta^2 - 2\eta (1+\sqrt{2\eta -1})\rbrack }~.
	\end{equation}
\end{theorem}
\begin{proof}
	Consider the following chain of inequalities:
	\begin{align}
	&\left\Vert (\operatorname{id}_R \otimes \hat{\L}_{\eta, N_B}) (\phi_{RA}) - (\operatorname{id}_R\otimes [\D_{(1-\eta)/\eta,N_B}\circ\L_{\eta, N_B}])(\phi_{RA})\right\Vert_1\nonumber \\
	\qquad &= \left\Vert (\operatorname{id}_R \otimes \hat{\L}_{\eta, N_B}) (\phi_{RA}) - (\operatorname{id}_R \otimes \Xi)(\phi_{RA}) \right\Vert_1\\
	&= \left\Vert \tr_{B} \{ \B_{AE' \rightarrow BE_2}(\phi_{RA}\otimes\psi_{\operatorname{TMS}}(N_B)_{E'E_1} ) -  \B_{AE' \rightarrow BE_2}(\phi_{RA}\otimes\omega(N_B)_{E'E_1})  \} \right\Vert_1\\
	& \leq \left\Vert \B_{AE' \rightarrow BE_2}(\phi_{RA}\otimes\psi_{\operatorname{TMS}}(N_B)_{E'E_1} ) -  \B_{AE' \rightarrow BE_2}(\phi_{RA}\otimes\omega(N_B)_{E'E_1}) \right\Vert_1\\
	&= \left\Vert \phi_{RA}\otimes\psi_{\operatorname{TMS}}(N_B)_{E'E_1}  -  \phi_{RA}\otimes\omega(N_B)_{E'E_1} \right\Vert_1\\
	&= \left\Vert \psi_{\operatorname{TMS}}(N_B)_{E'E_1}  - \omega(N_B)_{E'E_1} \right\Vert_1\\
	&\leq 2\sqrt{1-F(\psi_{\operatorname{TMS}}(N_B)_{E'E_1} , \omega(N_B)_{E'E_1})}
	\end{align}
	The first equality follows from \eqref{eq:equiv}. The second equality follows from \eqref{eq:comp-ch} and \eqref{eq:simulating-ch}. The first inequality follows from monotonicity of the trace distance. The third equality follows from invariance of the trace distance under a unitary transformation (beamsplitter). The last inequality follows from the Powers-Stormer inequality \cite{powers1970}.
	
	Next, we compute the fidelity between $\psi_{\operatorname{TMS}}(N_B)_{E'E_1}$ and $\omega(N_B)_{E'E_1}$ by using their respective covariance matrices in \eqref{eq:tms-cov} and \eqref{eq:noisy-tms-cov}, in the Uhlmann fidelity formula for two-mode Gaussian states \cite{MM12}. We find \cite{Mathematica}
\begin{equation}
F(\psi_{\operatorname{TMS}}(N_B)_{E'E_1} , \omega(N_B)_{E'E_1}) = \frac{\eta^2}{\eta^2 + N_B(N_B+1)\lbrack1+3 \eta^2 - 2\eta (1+\sqrt{2\eta -1})\rbrack }~.
\end{equation}
Since these inequalities hold for any input state $\phi_{RA}$, the final result follows from the definition of the diamond norm.
\end{proof}
\bigskip

\begin{theorem}\label{thm:qu2}
An upper bound on the quantum capacity of a thermal channel $\L_{\eta, N_B}$ with transmissivity $\eta \in [1/2,1]$, environment photon number $N_B$, and input mean photon-number constraint $N_S$ is given by
\begin{multline} \label{eq:qu2}
Q(\L_{\eta, N_B} , N_S) \leq Q_{U_2}(\L_{\eta, N_B} , N_S) \equiv g(\eta N_S+(1-\eta)N_B) -g(\zeta_{+}) -g(\zeta_{-}) \\ +(2\varepsilon' + 4\delta) g([(1-\eta)N_S+(1+\eta)N_B]/\delta)
+ g(\varepsilon') + 2 h_2(\delta)~,
\end{multline}
with 
\begin{align}
&\varepsilon = \sqrt{1-\eta^2/\left(\eta^2 + N_B(N_B+1)\lbrack1+3 \eta^2 - 2\eta (1+\sqrt{2\eta -1})\rbrack\right)}~,\\
&\zeta_{\pm} = \frac{1}{2}\left(-1+\sqrt{[ (1+2N_B)^2 - 2\varrho + (1+2\vartheta)^2 \pm 4(\vartheta -N_B)\sqrt{[1+N_B+\vartheta]^2-\varrho}]/2}\right)~,\\
&\varrho = 4N_B(N_B+1)(2\eta-1)/\eta~,\\
&\vartheta = \eta N_B+(1-\eta)N_S~,
\end{align}
$\varepsilon' \in (\varepsilon,1]$, and $\delta = (\varepsilon' - \varepsilon)/(1+ \varepsilon')$.
\end{theorem}
\begin{proof}
From Theorem \ref{thm:eps-diamond}, we have an upper bound on the diamond distance between the complementary channel of the thermal channel and the concatenation of the thermal channel followed by the degrading channel, i.e.,
	\begin{multline}
\frac{1}{2}\left\Vert \hat{\L}_{\eta, N_B} - \D_{(1-\eta)/\eta,N_B}\circ\L_{\eta, N_B}\right\Vert_{\diamond} \\
\leq \sqrt{1-\eta^2/\left(\eta^2 + N_B(N_B+1)\lbrack1+3 \eta^2 - 2\eta (1+\sqrt{2\eta -1})\rbrack\right)}  < \varepsilon' \leq 1.
\end{multline}
Due to the input mean photon number constraint $N_S$, and environment photon number $N_B$ for both $\L_{\eta, N_B}$ and $\D_{(1-\eta)/\eta,N_B}$, there is a total photon number constraint $(1-\eta)N_S+(1+\eta)N_B$ for the average output of $n$ channel uses of both $\hat{\L}_{\eta, N_B}$ and $\D_{(1-\eta)/\eta,N_B}\circ\L_{\eta, N_B}$. Using these results in Theorem \ref{thm:qcbound-eps-approx}, we find the following upper bound on the energy-constrained quantum capacity of a thermal channel:
\begin{equation}
Q(\L_{\eta, N_B} , N_S) \leq U_{\D_{(1-\eta)/\eta,N_B}}(\L_{\eta, N_B},N_S) +(2\varepsilon' + 4\delta) g([(1-\eta)N_S+(1+\eta)N_B]/\delta)
+ g(\varepsilon') + 2 h_2(\delta)~.
\end{equation}
Using Proposition \ref{thm:Ud-optimization}, we find that the thermal state with mean photon number $N_S$ optimizes the conditional entropy of degradation $U_{\D_{(1-\eta)/\eta,N_B}}(\L_{\eta, N_B},N_S)$. For the given thermal channel in \eqref{eq:th-ch} and the degrading channel in \eqref{eq:deg-map}, we find the following analytical expression \cite{Mathematica}:  
\begin{equation}
{U_{\D_{(1-\eta)/\eta,N_B}}(\L_{\eta, N_B},N_S)=g(\eta N_S+(1-\eta)N_B) -g(\zeta_{+}) -g(\zeta_{-}) }~,
\end{equation}
 with $\zeta_{\pm}$  defined as in the theorem statement. 
\end{proof}

\begin{proposition}\label{thm:Ud-optimization}
Let $\L_{\eta, N_B}$ be a thermal channel with transmissivity $\eta\in[1/2,1]$, environment photon number $N_B$, and input mean photon number constraint $N_S$. Let $\D_{(1-\eta)/\eta,N_B}$ be the degrading channel from \eqref{eq:deg-map}. Then the thermal state with mean photon number $N_S$ optimizes the conditional entropy of degradation $U_{\D_{(1-\eta)/\eta,N_B}}(\L_{\eta, N_B}, N_S)$, defined from \eqref{eq:Ud}. 
\end{proposition}
\begin{proof}
Consider the Stinespring dilation in \eqref{eq:full-dilation-state} of the degrading channel $\D_{(1-\eta)/\eta,N_B}$ from \eqref{eq:deg-map}, and denote it by $\W$. Then according to \eqref{eq:Ud},
\begin{equation} \label{eq:max-ach-state}
U_{\D_{(1-\eta)/\eta,N_B}}(\L_{\eta, N_B},N_S) = \sup_{\rho \, :\,  \tr\{\hat{n}\rho\}\leq N_S} H(G\vert E'_1E'_2)_{(\W\circ\L_{\eta, N_B})(\rho)}~.
\end{equation}
Our aim is to find an input state $\rho$ with a certain photon number $N_t \leq N_S$, such that it maximizes the conditional entropy in $\eqref{eq:max-ach-state}$. From the extremality of Gaussian states applied to the conditional entropy \cite{EW07}, it suffices to perform the optimization in $\eqref{eq:max-ach-state}$ over only  Gaussian states. 

Now, we argue that for a given input mean photon number $N_t$, a thermal state is the optimal state for the conditional output  entropy in \eqref{eq:max-ach-state}. For a thermal channel and our choice of a degrading channel, a phase rotation on the input state is equivalent to a product of local phase rotations on the outputs. Let us denote the state after the local phase rotations on the outputs by
\begin{equation}
\sigma_{E'_2GE'_1}(\phi) =  (e^{i \phi \hat{n}} \otimes  e^{i \phi \hat{n}} \otimes  e^{-i \phi \hat{n}}) (\W\circ\L_{\eta, N_B})(\rho) (e^{-i \phi \hat{n}} \otimes  e^{-i \phi \hat{n}}\otimes  e^{i \phi \hat{n}}),
\end{equation}
and let
\begin{equation}
\xi_{E'_2GE'_1} = \frac{1}{2\pi} \int_{0}^{2\pi} d\phi~(\W\circ\L_{\eta, N_B})(e^{{i \phi \hat{n}}} \rho e^{{-i \phi \hat{n}}})~.
\end{equation}
Note that the phase covariance property mentioned above is the statement that the following equality holds for all $\phi \in [0,2\pi)$  \cite{Mathematica}:
\begin{equation}
\sigma_{E'_2GE'_1}(\phi) =  (\W\circ\L_{\eta, N_B})(e^{{i \phi \hat{n}}} \rho e^{{-i \phi \hat{n}}}).
\end{equation}
Consider the following chain of inequalities
for a Gaussian input state $\rho$:
\begin{align}
H(G\vert E'_1E'_2)_{(\W\circ\L_{\eta,N_B})(\rho)}
&= 
\frac{1}{2\pi} \int_{0}^{2\pi} d\phi~H(G\vert E'_1E'_2)_{\sigma(\phi)}\\
&= \frac{1}{2\pi} \int_{0}^{2\pi} d\phi~H(G\vert E'_1E'_2)_{(\W\circ\L_{\eta, N_B})(e^{{i \phi \hat{n}}} \rho e^{{-i \phi \hat{n}}})}\\
&\leq H(G\vert E'_1E'_2)_{\xi}\\
&= H(G\vert E'_1E'_2)_{(\W\circ\L_{\eta, N_B})(\theta(N_t))}~,
\end{align}
The first equality follows from invariance of the conditional entropy under local unitaries. The second equality follows from the phase covariance property of the channel. The inequality follows from concavity of conditional entropy. The last equality follows from linearity of the channel, and 
the following identity:
\begin{equation}
\label{eq:phase-avg-t-state}
\theta(N_t) = \frac{1}{2\pi} \int_{0}^{2\pi} d\phi~e^{{i \phi \hat{n}}} \rho e^{{-i \phi \hat{n}}}~.
\end{equation}
In \eqref{eq:phase-avg-t-state}, the state after the phase averaging is diagonal in the number basis, and furthermore, the resulting state has the same photon number $N_t$ as the Gaussian state $\rho$. The thermal state $\theta(N_t)$ is the only Gaussian state of a single mode that is diagonal in the number basis with photon number equal to $N_t$.

Next, we argue that, for a given photon number constraint, a thermal state that saturates the constraint is the optimal state for the  conditional output entropy.
Let
\begin{multline}
\tau_{E'_2GE'_1}(\alpha) = \\
[D(\sqrt{1-\eta}\alpha)\otimes D(\sqrt{2\eta-1}\alpha) \otimes I] [(\W \circ \L_{\eta, N_B})(\theta(N_t))] [D^{\dagger}(\sqrt{1-\eta}\alpha)\otimes D^{\dagger}(\sqrt{2\eta-1}\alpha)\otimes I]~.
\end{multline} 
 Consider the following chain of inequalities:
\begin{align}
H(G\vert E'_1E'_2)_{(\W\circ\L_{\eta, N_B})(\theta(N_t))} &= \int d^2\alpha~q_{(N_S-N_t)}(\alpha)~H(G\vert E'_1E'_2)_{(\W\circ\L_{\eta, N_B})(\theta(N_t))}\\
&= \int d^2\alpha~q_{(N_S-N_t)}(\alpha)~H(G\vert E'_1E'_2)_{\tau(\alpha) }\\
&= \int d^2\alpha~q_{(N_S-N_t)}(\alpha)~H(G\vert E'_1E'_2)_{(\W\circ\L_{\eta, N_B})(D(\alpha)\theta(N_t)D^{\dagger}(\alpha))}\\
&\leq H(G\vert E'_1E'_2)_{(\W\circ\L_{\eta, N_B})\theta(N_S)}~,
\end{align}
where $q_{N}(\alpha) = \exp\{-\vert\alpha\vert^2/N\}/\pi N$ is a complex-centered Gaussian distribution with variance $N\geq 0$.
The first equality follows by placing a probability distribution in front, and the second follows from  invariance of the conditional entropy under local unitaries. The third equality follows because the channel is covariant with respect to displacement operators, as reviewed in \eqref{eq:covariance-gaussian}. The last inequality follows from concavity of conditional entropy, and from the fact that a thermal state with a higher mean photon number can be realized by random Gaussian displacements of a thermal state with a lower mean photon number, as reviewed in  \eqref{eq:displaced-thermal-is-thermal}. Hence, for a given input mean photon number constraint $N_S$, a thermal state with mean photon number $N_S$ optimizes the conditional entropy of degradation defined from \eqref{eq:Ud}. 
\end{proof}

\begin{remark}
The arguments used in the proof of Proposition~\ref{thm:Ud-optimization} can be employed in more general situations beyond that which is discussed there. The main properties that we need are the following, when the channel involved takes a single-mode input to a multi-mode output:
\begin{itemize}
\item The channel should be phase covariant, such that a phase rotation on the input state is equivalent to a product of local phase rotations on the output.
\item The channel should be covariant with respect to displacement operators, such that a displacement operator acting on the input state is equivalent to a product of local displacement operators on the output.
\item The function being optimized should be invariant with respect to local unitaries and concave in the input state.
\end{itemize}
If all of the above hold, then we can conclude that the thermal-state input saturating the energy constraint is an optimal input state. We employ this reasoning again in the proof of Theorem~\ref{thm:1.45bits}.
\end{remark}

\subsection{$\varepsilon$-close-degradable bound on the energy-constrained quantum capacity of bosonic thermal channels}
\label{sec:eps-close-q-cap}

In this section, we first establish an upper bound on the diamond distance between a thermal channel and a pure-loss channel. Since a pure-loss channel is a degradable channel, an upper bound on the energy-constrained quantum capacity of a thermal channel directly follows from Theorem \ref{thm:qcbound-eps-close}.
\begin{theorem} \label{thm:thermal-pureloss}
	If a thermal channel $\L_{\eta, N_B}$  and a pure-loss bosonic channel $\L_{\eta,0}$ have the same transmissivity parameter $\eta \in[0, 1]$, then
	\begin{equation}
	\frac{1}{2}\left\Vert \L_{\eta, N_B} - \L_{\eta,0} \right\Vert_{\diamond} \leq \frac{N_B}{N_B+1}~.
	\end{equation}
\end{theorem}
\begin{proof} Let $\B$ represent the beamsplitter transformation, and let $\theta_E(N_B)$ and $\theta'_E(0)$ denote the states of the environment for the thermal channel and pure-loss channel, respectively. For any input state $\psi_{RA}$ to both thermal and pure-loss channels, the following inequalities hold:
 \begin{align}
& \left\Vert (\operatorname{id}_R \otimes \L_{\eta,N_B}) (\psi_{RA}) - (\operatorname{id}_R \otimes \L_{\eta, 0}) (\psi_{RA})\right\Vert_1 
\nonumber \\
& = 
 \left\Vert \tr_{E'}\{\B_{AE \rightarrow BE'} (\psi_{RA}\otimes \theta_E(N_B))  - \B_{AE \rightarrow BE'}( \psi_{RA} \otimes \theta'_E(0))\} \right\Vert_1\\
 & \leq \left\Vert \B_{AE \rightarrow BE'} (\psi_{RA}\otimes \theta_E(N_B))  - \B_{AE \rightarrow BE'}( \psi_{RA} \otimes \theta'_E(0)) \right\Vert_1\\
 &=\left\Vert \psi_{RA} \otimes \theta_E(N_B)   -  \psi_{RA}\otimes \theta'_E(0)  \right\Vert_1\\
 &=\left\Vert \theta_E(N_B) - \theta'_E(0) \right\Vert_1\\
 &= \left\Vert \sum_{n=0}^{\infty} \frac{(N_B)^n}{(N_B+1)^{n+1}}\ \ket{n}\bra{n} - \ket{0}\bra{0}\right\Vert_1 \\
 &= \frac{2N_B}{N_B+1}~.
   \end{align}	
   The first equality follows from the definition of the channel in terms of its environment and a unitary interaction (beam splitter). The first inequality follows from monotonicity of the trace distance. The second equality follows from invariance of the trace distance under a unitary operator (beamsplitter). The last equality follows from basic algebra. Since these inequalities hold for any state $\psi_{RA}$, the final result follows from the definition of the diamond norm.
\end{proof}

\bigskip

\begin{remark} In \cite{2016channel-discrimination}, it has been shown that the optimal strategy to distinguish two quantum thermal channels $\L_{\eta,N^1_B}$ and $\L_{\eta,N^2_B}$, each having the same transmissivity parameter $\eta$, and thermal noises $N^1_B$ and $N^2_B$, respectively, is to use a highly squeezed, two-mode squeezed vacuum state $\psi_{\operatorname{TMS}}(N_S)_{RA}$ as input to the channels. According to \cite[Eq.~(35)]{2016channel-discrimination}, 
	\begin{equation}
	\lim_{N_S \rightarrow \infty} F(\sigma_{N^1_B} , \sigma_{N^2_B}) = F(\theta(N^1_B), \theta(N^2_B)),
	\end{equation}
	where  $\sigma_{N^i_B} \equiv (\operatorname{id}_R \otimes \L_{\eta,N^i_B} )(\psi_{\operatorname{TMS}}(N_S)_{RA})$, and $\theta(N^i_B)$ is a thermal state with mean photon number $N^i_B$. Hence, a lower bound on the diamond distance in Theorem \ref{thm:thermal-pureloss} is given by 
	\begin{equation}
	\frac{1}{2}\left\Vert \L_{\eta, N_B} - \L_{\eta,0} \right\Vert_{\diamond} \geq 1-\sqrt{F(\theta(N_B), \theta(0))} = 1- 1/\sqrt{N_B+1},
	\end{equation} 
	where the inequality follows from the Powers-Stormer inequality \cite{powers1970}. We also suspect that the upper bound in Theorem \ref{thm:thermal-pureloss} is achievable, but we are not aware of a method for computing the trace distance of general quantum Gaussian states, which is what it seems would be needed to verify this suspicion.
\end{remark}

\bigskip
\begin{theorem}\label{thm:qu3}
	An upper bound on the quantum capacity of a thermal channel $\L_{\eta, N_B}$ with transmissivity $\eta \in \lbrack1/2,1\rbrack$, environment photon number $N_B$, and input mean photon number constraint $N_S$ is given by
	\begin{multline} \label{eq:qu3}
	Q(\L_{\eta, N_B},N_S) \leq Q_{U_3}(\L_{\eta, N_B},N_S) \equiv g(\eta N_S)-g[(1-\eta)N_S] \\
	+ (4\varepsilon' + 8\delta) g[(\eta N_S +(1-\eta)N_B)/\delta] + 2g(\varepsilon') + 4h_2(\delta)~, 
	\end{multline} 
	with $\varepsilon = N_B/(N_B+1)$, $\varepsilon' \in (\varepsilon,1]$ and $\delta = (\varepsilon' - \varepsilon)/(1+ \varepsilon')$.
\end{theorem}
\begin{proof}
	From Theorem \ref{thm:thermal-pureloss}, we have that $\frac{1}{2}\left\Vert \L_{\eta, N_B} - \L_{\eta,0} \right\Vert_{\diamond} \leq \frac{N_B}{N_B+1} < \varepsilon' \leq 1$. Due to the input mean photon number constraint $N_S$ for $n$ channel uses, the output mean photon number cannot exceed $\eta N_S + (1-\eta)  N_B$ for the thermal channel and $\eta N_S$ for the pure-loss channel. Hence, there is a photon number constraint $\eta N_S + (1-\eta)  N_B$ for the output  of both the thermal and pure-loss channels.  Since the pure-loss channel is a degradable channel for $\eta \in [1/2,1]$ \cite{WPG07,GSE08}, the final result follows directly from Theorem~\ref{thm:qcbound-eps-close}.
\end{proof}

\bigskip

\section{Comparison of upper bounds on the energy-constrained quantum capacity of bosonic thermal channels }\label{sec:comparision-of-bounds}

In this section, we study the closeness of the three different upper bounds when compared to a known lower bound. In particular, 
we use the following lower bound on the quantum capacity of a thermal channel \cite{HW01,Mark2012tradeoff} and denote it by $Q_L$:
\begin{multline}
Q(\L_{\eta, N_B}, N_S) \geq  Q_L(\L_{\eta, N_B}, N_S ) \equiv g(\eta N_S + (1-\eta) N_B) \\
- g([D+(1-\eta)N_S - (1-\eta)N_B-1]/2)  
- g([D-(1-\eta)N_S + (1-\eta)N_B-1]/2), \label{eq:ql}
\end{multline}
where
\begin{equation}
D^2 \equiv [(1+\eta)N_S+(1-\eta)N_B+1]^2 - 4\eta N_S(N_S+1).
\end{equation}
We start by discussing how close the data-processing bound $Q_{U_1}$ is to the aforementioned lower bound. In particular, we show that the data-processing bound $Q_{U_1}$ can be at most $1.45$ bits larger than $Q_L$. 
\begin{theorem}\label{thm:1.45bits}
		Let $\L_{\eta, N_B}$ be a thermal channel with transmissivity $\eta\in[1/2,1]$, environment photon number $N_B$, and input mean photon number constraint $N_S$. Then the following relation holds between the data-processing bound $Q_{U_1}(\L_{\eta, N_B},N_S)$ in \eqref{eq:qu1} and the lower bound $Q_L(\L_{\eta, N_B},N_S) $ in \eqref{eq:ql} on the energy-constrained quantum capacity of a thermal channel:
	\begin{equation} \label{eq:1.45}
	Q_L(\L_{\eta, N_B},N_S) \leq Q_{U_1}(\L_{\eta, N_B},N_S) \leq Q_L(\L_{\eta, N_B},N_S)+ 1/\ln 2 ~.
	\end{equation}
\end{theorem}
\begin{proof}
To prove this result, we first compute the difference between the data-processing bound in \eqref{eq:qu1} and the lower bound in \eqref{eq:ql} and show that it is equal to $1/\ln 2$ as $N_S \rightarrow \infty$. Next, we prove that the difference is a monotone increasing function with respect to input mean photon number $N_S \geq 0$. Hence, the difference $Q_{U_1}(\L_{\eta, N_B},N_S) - Q_L(\L_{\eta, N_B},N_S)$ attains its maximum value in the limit  $N_S \to \infty$.
We note that a similar statement has been given in \cite{smith2013} to  bound  the classical capacity of a thermal channel, but the details of the approach we develop here are different and are likely to be more broadly applicable to related future questions.

For simplicity, we denote $(1-\eta)N_B$ as $Y$, employ the natural logarithm for $g(x)$, and omit the prefactor $1/\ln 2$ from all instances of $g(x)$. We use the following property of the function $g(x)$: For large~$x$,
\begin{equation}
g(x)=\ln(x+1)+1 + O(1/x)  \label{eq:entropy-trick},
\end{equation}
so that as $x \rightarrow \infty$, the approximation $g(x)\approx\ln(x+1)+1  $ holds. 
Using (\ref{eq:entropy-trick}),  the data-processing bound in \eqref{eq:qu1} can be expressed as follows for large $N_S$:
\begin{equation}
\ln(Y+1+ \eta N_S) - \ln(Y+1+(Y+1-\eta)N_S) + O(1/N_S) ~.
\end{equation}
Similarly, the lower bound $Q_L$ in \eqref{eq:ql} can be expressed as
\begin{equation}
\ln(1+\eta N_S + Y) - \ln([1+D+(1-\eta)N_S-Y]/2) - \ln([1+D-(1-\eta)N_S+Y]/2) + O(1/N_S) - 1 ~.
\end{equation}
Let us denote the difference between $Q_{U_1}$ and $Q_L$ by $\Delta(\L_{\eta, N_B}, N_S )$.
\begin{equation}
\Delta(\L_{\eta, N_B}, N_S ) = Q_{U_1}(\L_{\eta, N_B},N_S) - Q_L(\L_{\eta, N_B},N_S).
\end{equation}
Then the difference simplifies as
\begin{align}
& \Delta(\L_{\eta, N_B}, N_S )\nonumber \\
&= 1- \ln(Y+1+(Y+1-\eta)N_S) +\ln([(1+D)^2-((1-\eta)N_S-Y)^2]/4) + O(1/N_S)~.\\
&= 1 -\ln(Y+1+(Y+1-\eta)N_S) + \ln([1+N_S(1-\eta+2Y)+Y+D]/2) + O(1/N_S)\\
&= 1+ \ln([1+N_S(1-\eta+2Y)+Y+D]/[2(Y+1+(Y+1-\eta)N_S)]) +O(1/N_S).
\end{align}
The second equality follows from the definition of $D^2$. Next, we show that
\begin{equation}
\ln([1+N_S(1-\eta+2Y)+Y+D]/[2(Y+1+(Y+1-\eta)N_S)])  \rightarrow 0
\end{equation}
as $N_S \rightarrow \infty$, and hence we get the desired result. Consider the following expression and take the limit ${N_S \rightarrow \infty}$:
\begin{align}
& \lim_{N_S \rightarrow \infty}  \frac{1+N_S(1-\eta+2Y)+Y+D}{2(Y+1+(Y+1-\eta)N_S)}\\
&=\lim_{N_S \rightarrow \infty}  \frac{1/N_S+(1-\eta+2Y)+Y/N_S+\sqrt{((1+\eta)+(Y+1)/N_S)^2-4\eta -4\eta/N_S}}{2((Y+1)/N_S+(Y+1-\eta))}\\
&\rightarrow \frac{(1-\eta+2Y)+1-\eta}{2(Y+1-\eta)} = 1~. 
\end{align}
Hence, $\lim_{N_S \rightarrow \infty} \Delta(\L_{\eta, N_B}, N_S ) = 1$. After incorporating the $1/\ln 2 $ factor, which was omitted earlier for simplicity, we find that the difference between the upper and lower bounds approaches $1/\ln 2$ ($\approx$ 1.45 bits) as $N_S \rightarrow \infty$. 

Now, we show that the difference $\Delta(\L_{\eta, N_B}, N_S )$ is a monotone increasing function with respect to input mean photon number $N_S \geq 0$. Let $\U^{\eta'}_{A\to B_1E_1}$ and $\V^{G}_{B_1 \to B_2E_2}$ denote Stinespring dilations of a pure-loss channel $\L_{\eta', 0}:A\to B_1$ and a quantum limited amplifier channel $\A_{G,0}:B_1 \to B_2$, respectively.  
For the energy-constrained quantum capacity of a pure-loss  channel, the thermal state as an input is optimal for any fixed energy or input mean photon number constraint $N_S$ \cite{Mark2012tradeoff}. Moreover, the lower bound in \eqref{eq:ql} is obtained for a thermal state with mean photon number $N_S$ as input to the channel. 
Then the action of a thermal channel $\L_{\eta, N_B}$ on an input state $\theta(N_S)$ can be expressed as
\begin{equation}
 \L_{\eta, N_B}(\theta(N_S)) = \tr_{E_1E_2}\{ (\operatorname{id}_{E_1} \otimes  \V^G_{B_1 \to B_2 E_2} )\circ \U^{\eta'}_{A \to B_1 E_1} (\theta(N_S))\}~.
\end{equation}
Consider the following state:
\begin{align}
 \omega_{B_2E_1E_2} =  (\operatorname{id}_{E_1} \otimes  \V^G_{B_1 \to B_2 E_2} )\circ \U^{\eta'}_{A \to B_1 E_1} (\theta(N_S))~.
\end{align}
Since the data-processing bound $Q_{U_1}(\L_{\eta, N_B},N_S)$ is equal to the quantum capacity of a pure-loss channel with transmissivity $\eta'$, which in turn is equal to coherent information for this case, \eqref{eq:qu1} can also be represented as
\begin{align}
Q_{U_1}(\L_{\eta, N_B},N_S) =  H(B_2E_2)_{\omega} - H(E_1)_{\omega}~. \label{eq:qu1-b2e2e1}
\end{align}
Similarly, the lower bound can be expressed as
\begin{align}
Q_L(\L_{\eta, N_B}, N_S) = H(B_2)_{\omega} - H(E_1E_2)_{\omega}. \label{eq:ql-b2e2e1}
\end{align}
Hence the difference between \eqref{eq:qu1-b2e2e1} and \eqref{eq:ql-b2e2e1} is given by
\begin{align}
\Delta(\L_{\eta, N_B},N_S) =  H(E_2\vert B_2)_{\omega}+ H(E_2\vert E_1)_{\omega}~.\label{eq:sum-of-two-ce}
\end{align}
Now, our aim is to show that the conditional entropies in \eqref{eq:sum-of-two-ce} are monotone increasing functions of $N_S$. We employ  displacement covariance of the channels, and note that this argument is similar to that used  in the proof of Proposition \ref{thm:Ud-optimization}.
Let 
\begin{align}
&\sigma_{B_2E_1E_2}(\alpha) = [D(\sqrt{\eta G}\alpha)\otimes I \otimes  D(\sqrt{\eta (G-1)}\alpha) ]~\omega_{B_2E_1E_2}~[D^{\dagger}(\sqrt{\eta G}\alpha)\otimes I \otimes D^{\dagger}(\sqrt{\eta (G-1)}\alpha)],\\
&\tau_{B_2E_1E_2}(\alpha)= [I \otimes D(\sqrt{1-\eta }\alpha)\otimes   D(\sqrt{\eta (G-1)}\alpha)]~\omega_{B_2E_1E_2} ~[I \otimes D^{\dagger}(\sqrt{1-\eta }\alpha)\otimes D^{\dagger}(\sqrt{\eta (G-1)}\alpha)]~.
\end{align}
Let $N'_S - N_S \geq 0$, and consider the following chain of inequalities:
\begin{align}
H(E_2\vert B_2)_{\omega} + H(E_2 \vert E_1)_{\omega} &= \int d^2\alpha ~q_{(N'_S-N_S)}(\alpha)~[H(E_2\vert B_2)_{\omega} + H(E_2 \vert E_1)_{\omega}]\\
&=\int d^2\alpha~q_{(N'_S-N_S)}(\alpha)~[H(E_2\vert B_2)_{\sigma(\alpha)} + H(E_2 \vert E_1)_{\tau(\alpha)}]\\
&= \int d^2\alpha~q_{(N'_S-N_S)}(\alpha)~[H(E_2\vert B_2)_{(\V^G\circ\U^{\eta'})(D(\alpha)\theta(N_S)D^{\dagger}(\alpha))}]\nonumber \\
&\qquad+ \int d^2\alpha~q_{(N'_S-N_S)}(\alpha)~ [H(E_2 \vert E_1)_{(\V^G\circ\U^{\eta'})(D(\alpha)\theta(N_S)D^{\dagger}(\alpha))}]\\
&\leq H(E_2\vert B_2)_{(\V^G\circ\U^{\eta'})(\theta(N'_S))}+ H(E_2\vert E_1)_{(\V^G\circ\U^{\eta'})(\theta(N'_S))}~.
\end{align}
The first equality follows by placing a probability distribution in front, and the second follows from invariance of the conditional entropy under local unitaries. The third equality follows because the channel is covariant with respect to displacement operators, as reviewed in \eqref{eq:covariance-gaussian}. The last inequality follows from concavity of conditional entropy, and from the fact that a thermal state with a higher mean photon number can be realized by random Gaussian displacements of a thermal state with a lower mean photon number, as reviewed in \eqref{eq:displaced-thermal-is-thermal}. 

Hence, the difference between the data-processing bound in \eqref{eq:qu1} and the lower bound in \eqref{eq:ql}  attains its maximum value in the limit  $ N_S \to \infty$.
\end{proof}

\begin{figure*}
	\begin{center}
		\subfloat[]{\includegraphics[width=.42\columnwidth]{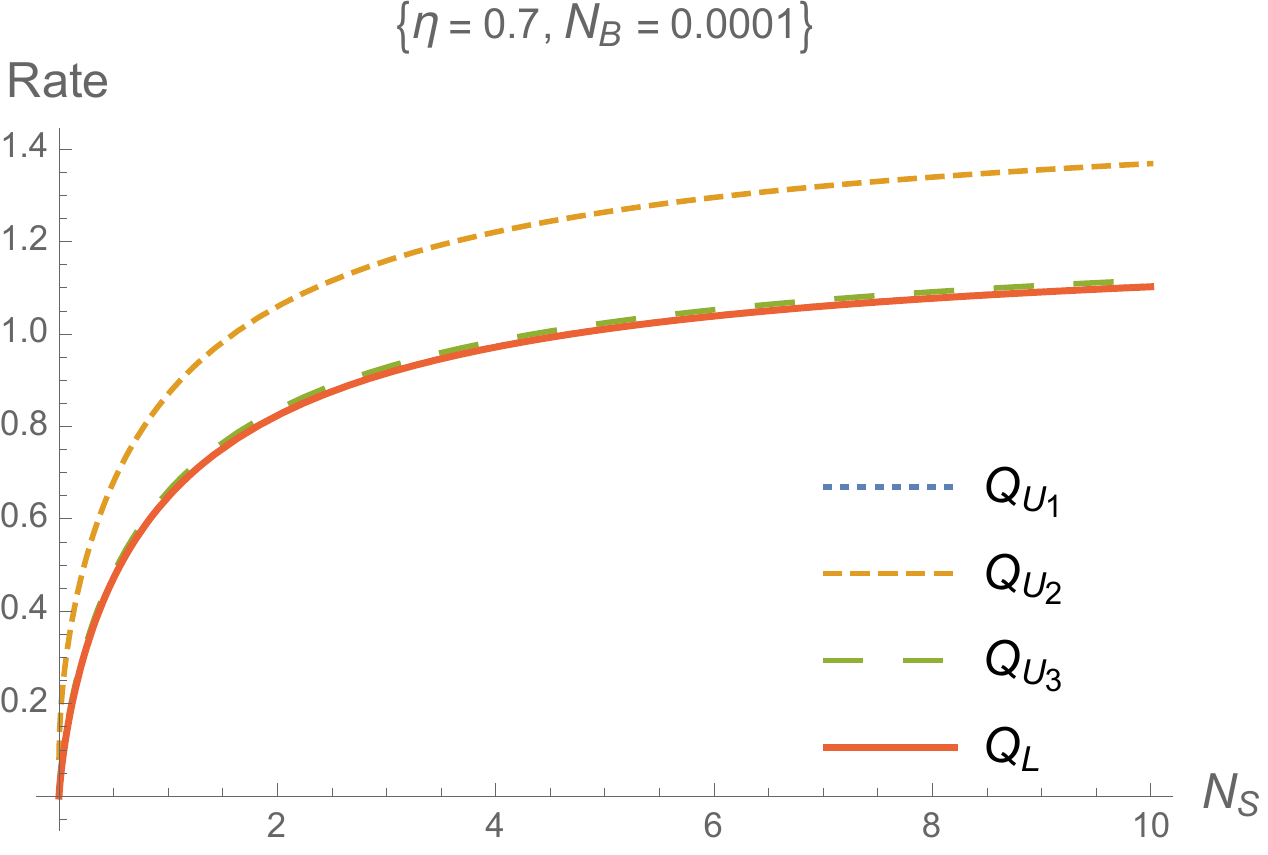}}\qquad
		\qquad\subfloat[]{\includegraphics[width=.42\columnwidth]{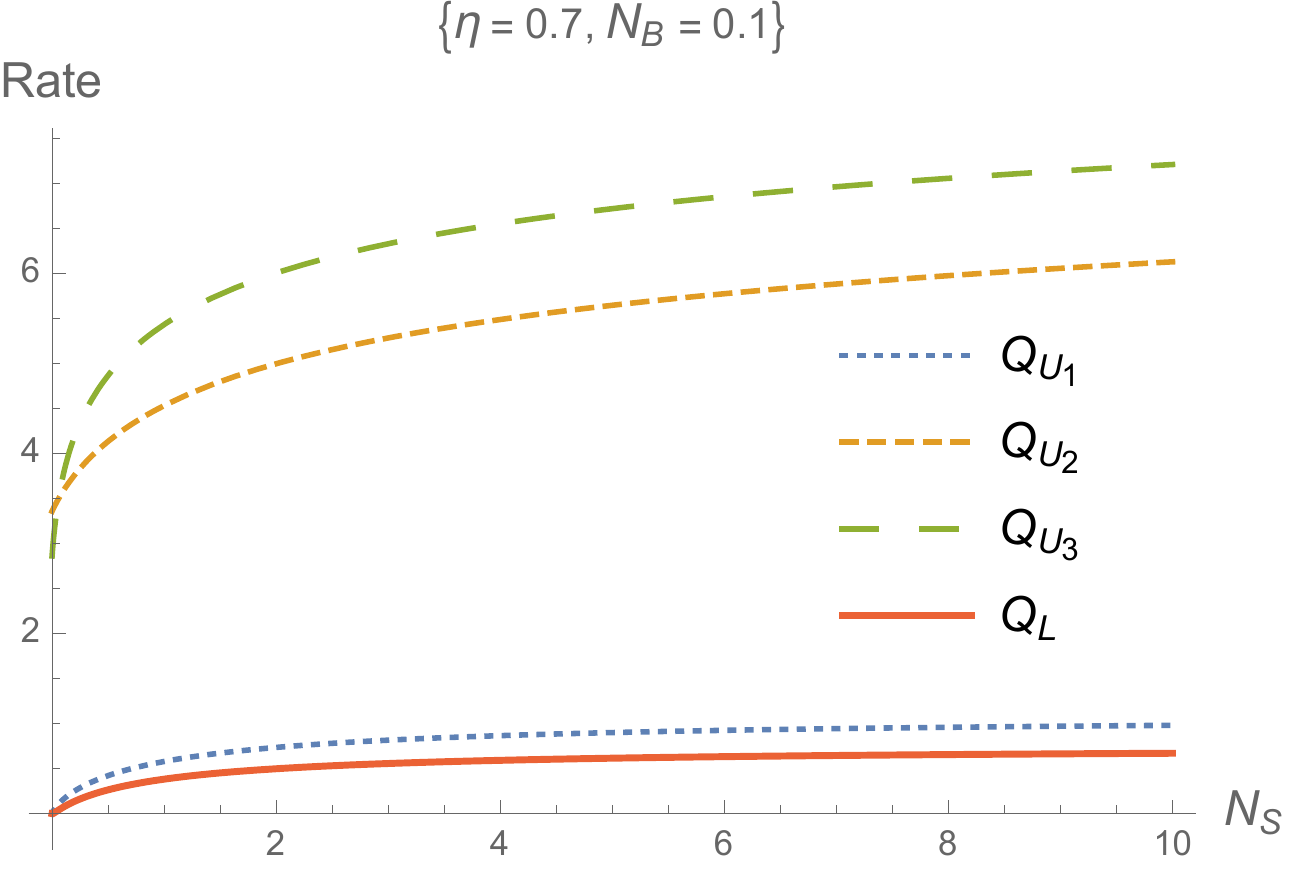}}\newline%
		\subfloat[]{\includegraphics[width=.42\columnwidth]{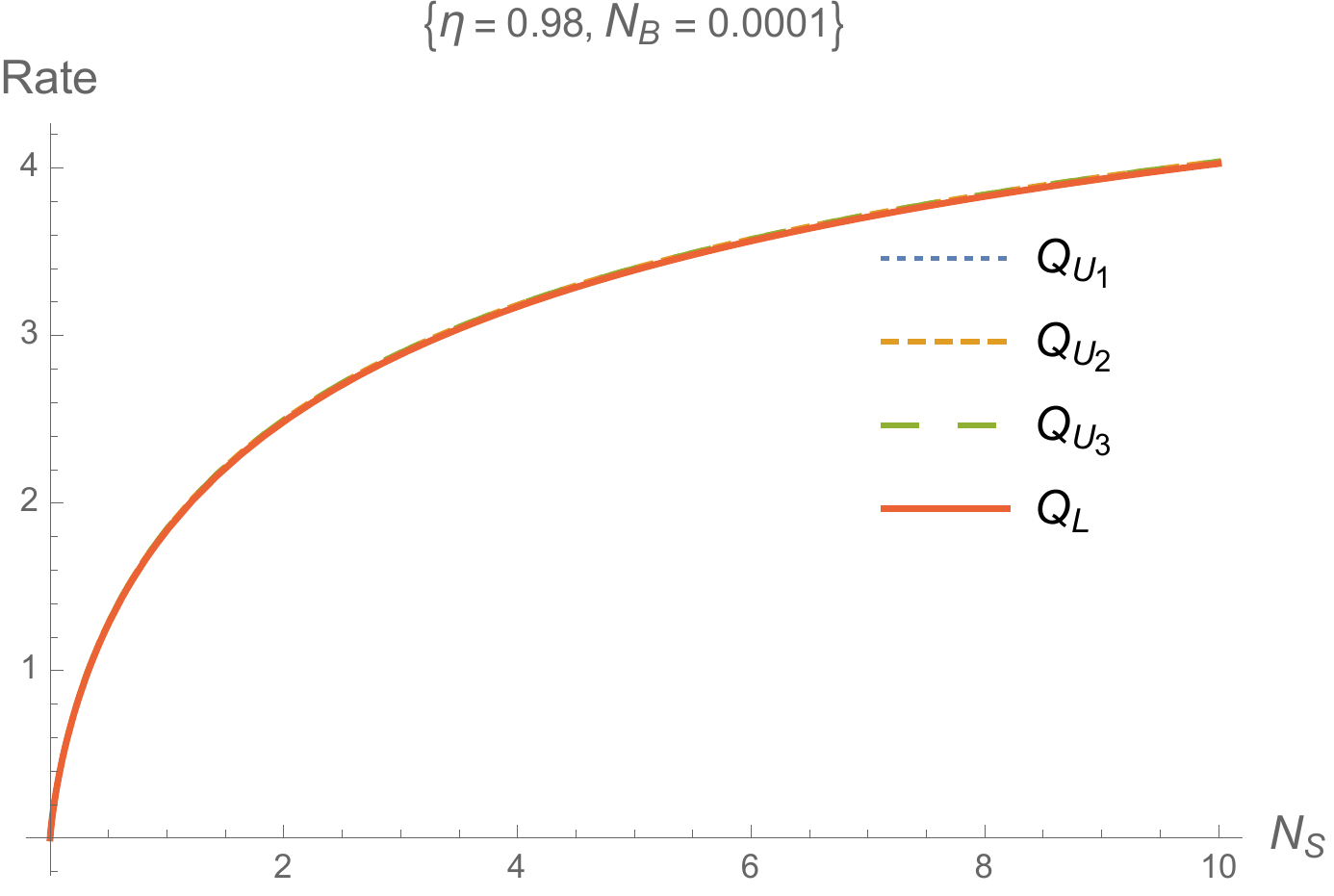}}\qquad
		\qquad\subfloat[]{\includegraphics[width=.42\columnwidth]{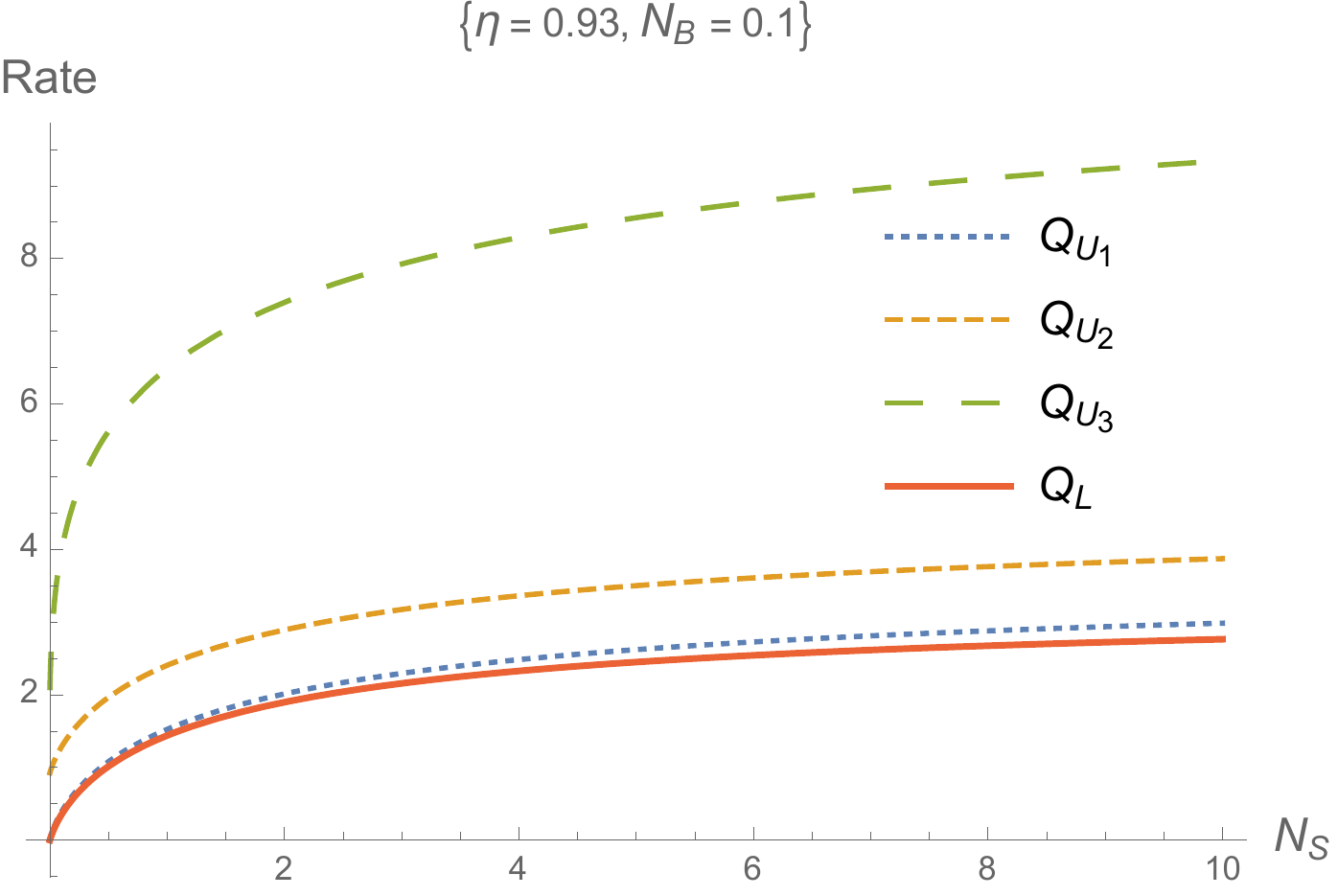}}
	\end{center}
	\caption{The figures plot the data-processing bound $(Q_{U_1})$, the $\varepsilon$-degradable bound $(Q_{U_2})$, the $\varepsilon$-close-degradable bound $(Q_{U_3})$ and the lower bound $(Q_{L})$ on energy-constrained quantum capacity of thermal channels. In each figure, we select certain values of $\eta$ and $N_B$, with the choices indicated above each figure. In all the cases, the data-processing bound $Q_{U_1}$ is close to the lower bound $Q_{U_L}$. In (a), for medium transmissivity and low thermal noise, the $\varepsilon$-close-degradable bound is close to the data-processing bound, and they are tighter than the $\varepsilon$-degradable bound. In (b), for medium transmissivity and high thermal noise, only the data-processing bound is close to the lower bound. Also the $\varepsilon$-degradable bound is tighter than the $\varepsilon$-close-degradable bound. In (c), for high transmissivity and low thermal noise, all upper bounds are very near to the lower bound. In (d), for high transmissivity and high noise, the $\varepsilon$-degradable bound is tighter  than the $\varepsilon$-close-degradable bound.}%
	\label{fig:bounds-quantum-cap}%
\end{figure*}

\bigskip
Next, we perform numerical evaluations to see how close the three different upper bounds are to the lower bound $Q_L$ in \eqref{eq:ql}. Since there is a free parameter $\varepsilon'$ in both the $\varepsilon$-degradable bound in \eqref{eq:qu2} and the $\varepsilon$-close-degradable bound in \eqref{eq:qu3}, we optimize these bounds with respect to $\varepsilon'$ \cite{Mathematica}. In Figure \ref{fig:bounds-quantum-cap}, we plot the data-processing bound $Q_{U_1}$, the $\varepsilon$-degradable bound $Q_{U_2}$, the $\varepsilon$-close-degradable bound $Q_{U_3}$ and the lower bound $Q_L$ versus $N_S$ for certain values of the transmissivity $\eta$ and thermal noise $N_B$. In particular, we find that the data-processing bound is close to the lower bound $Q_L$ for both low and high thermal noise. This is related to Theorem~\ref{thm:1.45bits}, as the data-processing bound can  be at most 1.45 bits larger than the lower bound $Q_L$. In Figure \ref{fig:bounds-quantum-cap}(a), we plot for medium transmissivity and low thermal noise. We find that the $\varepsilon$-close-degradable bound is very near to the data-processing bound and is tighter than the $\varepsilon$-degradable bound. In Figure \ref{fig:bounds-quantum-cap}(b), we plot for medium transmissivity and high thermal noise. We find that the $\varepsilon$-degradable bound is tighter than the $\varepsilon$-close degradable bound. In Figure \ref{fig:bounds-quantum-cap}(d), we plot for high transmissivity and high thermal noise. In Figure \ref{fig:bounds-quantum-cap}(c), we plot for high transmissivity and low thermal noise. We find that all upper bounds are very near to the lower bound $Q_L$. From Figures 
\ref{fig:bounds-quantum-cap}(a) and \ref{fig:bounds-quantum-cap}(c), it is evident that in the low-noise regime, there is a strong limitation on any potential super-additivity of coherent information of a thermal channel. Similar results were obtained on quantum and private capacities of low-noise quantum channels in \cite{Felix17}. It is important to stress that the upper bound $Q_{U_3}$ can serve as a good bound only for low values of the thermal noise $N_B$, as the technique to calculate this bound requires the closeness of a thermal channel with a pure-loss channel (discussed in Theorem~\ref{thm:thermal-pureloss}), and the closeness parameter is equal to $N_B/(N_B+1)$. 

In Figure \ref{fig:eps-degradable-best}, we plot all the upper bounds and the lower bound $Q_L$ versus $N_S$, for high transmissivity and high thermal noise. In Figure \ref{fig:eps-degradable-best}(a), we find that the $\varepsilon$-degradable bound is tighter than all other bounds for high values of $N_S$. In Figure \ref{fig:eps-degradable-best}(b), we plot for the same parameter values, but for low values of $N_S$.
It is evident that for low input mean photon number, the data-processing bound is tighter than the $\varepsilon$-degradable bound.

The plots suggest that our upper bounds based on the notion of approximate degradability are good for the case of high input mean photon number. We suspect that these bounds can be further improved for the case of low input mean photon number by considering the energy-constrained diamond norm \cite{Sh17,Wetal17}. To address this question, we consider the generalized channel divergences of quantum Gaussian channels in Section \ref{sec:generalized-channel-divergence} and argue about their optimization.
	
\begin{figure*}
	\begin{center}
		\subfloat[]{\includegraphics[width=.70\columnwidth]{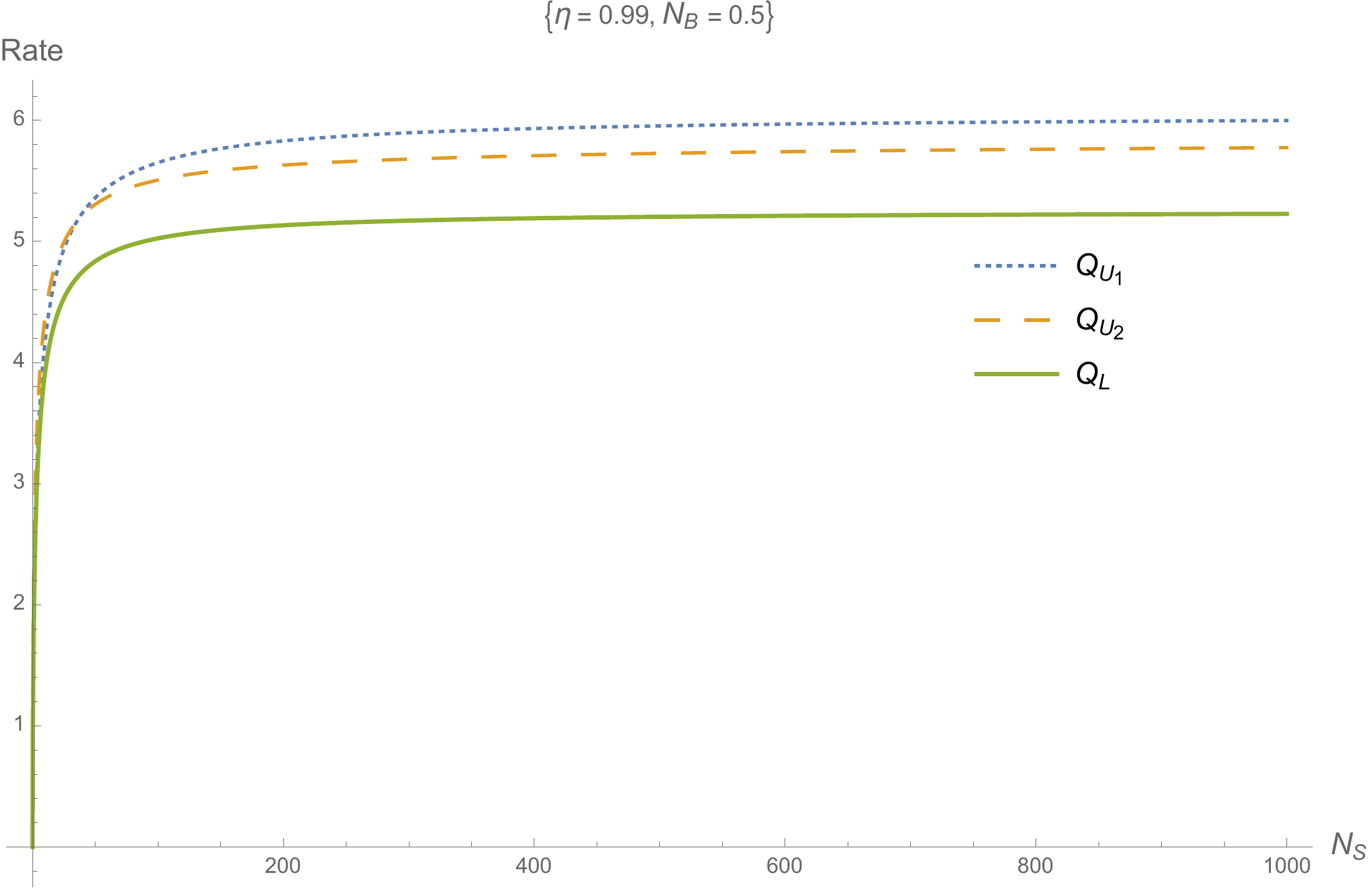}}\qquad
		\qquad\subfloat[]{\includegraphics[width=.70\columnwidth]{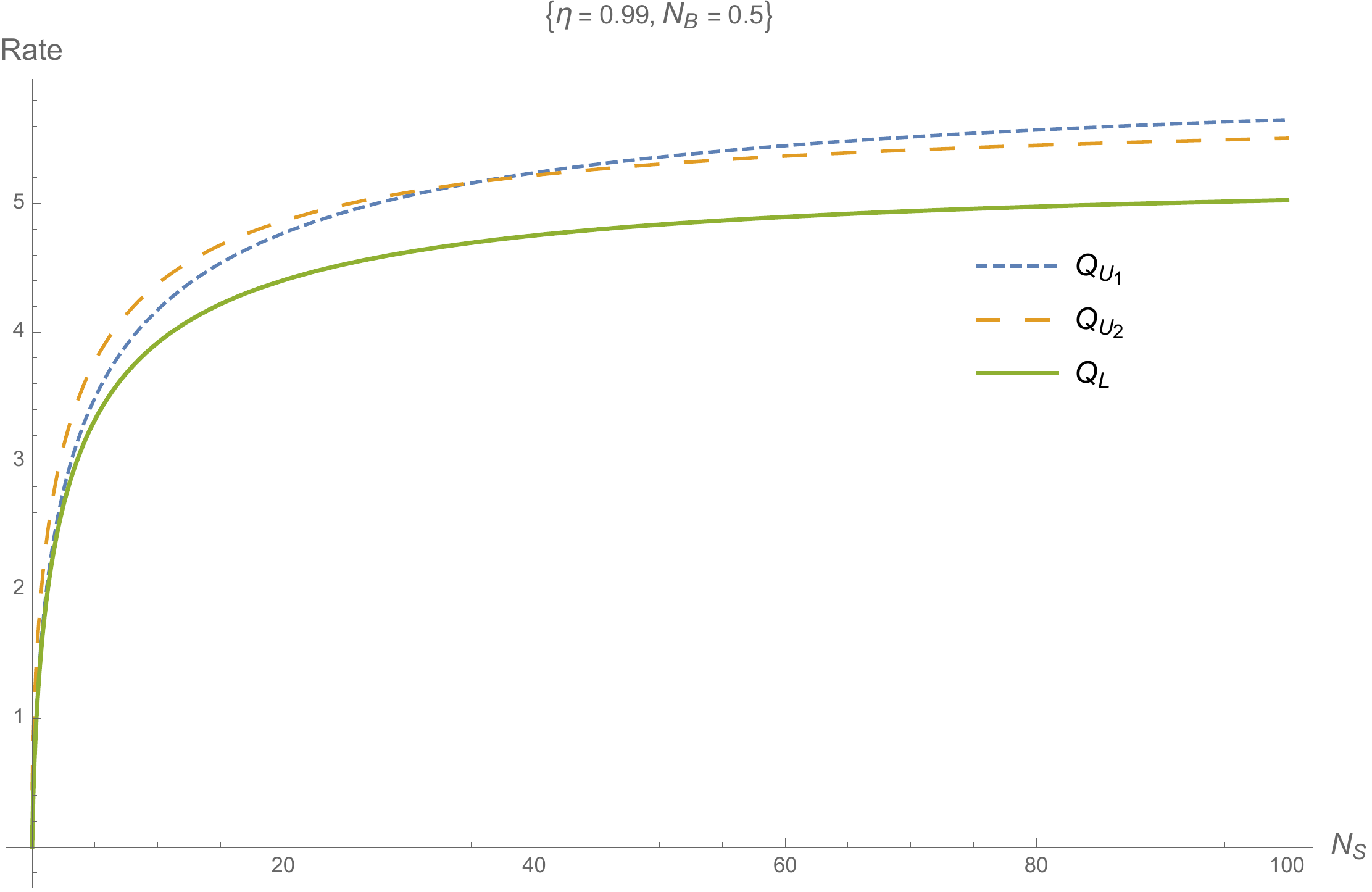}}
	\end{center}
	\caption{The figures plot the data-processing bound $(Q_{U_1})$, the $\varepsilon$-degradable bound $(Q_{U_2})$,  and the lower bound $(Q_{L})$ on energy-constrained quantum capacity of thermal channels (the $\varepsilon$-close-degradable bound $(Q_{U_3})$ is not plotted because it is much higher than the other bounds for all parameter values considered). In each figure, we select $\eta= 0.99$ and $N_B=0.5$. In (a), the $\varepsilon$-degradable upper bound is tighter than all other upper bounds. In (b), for low values of $N_S$, the data-processing bound is tighter than the $\varepsilon$-degradable bound.}%
	\label{fig:eps-degradable-best}%
\end{figure*}

\section{Upper bounds on energy-constrained private capacity of bosonic thermal channels} 

\label{sec:upper-bound-private-cap}

In this section, we provide three different upper bounds on the energy-constrained private capacity of a thermal channel. These upper bounds are derived very similarly as in Section \ref{sec:upper-bound-quantum-cap}. We call these different bounds the data-processing bound, the $\varepsilon$-degradable bound, and the $\varepsilon$-close-degradable bound, and denote them by $P_{U_1}$, $P_{U_2}$, and $P_{U_3}$, respectively. 

\subsection{Data-processing bound on the energy-constrained private capacity of bosonic thermal channels}

\begin{theorem}\label{thm:data-proc-pri-cap}
	An upper bound on the private capacity of a thermal channel $\L_{\eta, N_B}$ with transmissivity $\eta\in\lbrack1/2,1\rbrack$, environment photon number $N_B \geq 0$, and input mean photon number constraint $N_S \geq 0$  is given by
	\begin{align}
	P(\L_{\eta, N_B}, N_S )  &\leq
	\max\{0, P_{U_1}(\L_{\eta, N_B}, N_S )\}\\
	P_{U_1}(\L_{\eta, N_B}, N_S ) & \equiv g(\eta' N_{S})-g[(1-\eta')N_{S}]~, \label{eq:upper-bound-data-proc-pri-cap}
	\end{align}
	with $\eta' = \eta/((1-\eta)N_B+1)$.
\end{theorem}
\begin{proof}
	A proof follows from arguments similar to those in the proof of Theorem~\ref{thm:qu1}. Since a pure-loss channel is a degradable channel \cite{WPG07,GSE08}, its energy-constrained private capacity is the same as its energy-constrained quantum capacity \cite{MH16}.
\end{proof}

\begin{remark}
Applying Remarks \ref{rem:unconstrained-cap} and \ref{rem:unconstrained-qcap-thermal}, we find the following data-processing bound $P_{U_1}(\L_{\eta, N_B})$ on the unconstrained private capacity
$P(\L_{\eta, N_B})$
of a thermal channel $\L_{\eta, N_B}$: 
\begin{align}
P(\L_{\eta, N_B}) \leq P_{U_1}(\L_{\eta, N_B}) = \log_{2}(\eta/(1-\eta)) - \log_2(N_B+1)~.
\end{align}
\end{remark}

\subsection{$\varepsilon$-degradable bound on the energy-constrained private capacity of bosonic thermal channels}

\label{sec:eps-deg-pri-cap}

\begin{theorem}
	An upper bound on the private capacity of a thermal channel $\L_{\eta, N_B}$ with transmissivity $\eta \in [1/2,1]$, environment photon number $N_B \geq 0$, and input mean photon number constraint $N_S\geq 0$ is given by
	\begin{multline}
	P(\L_{\eta, N_B} , N_S) \leq P_{U_2}(\L_{\eta, N_B} , N_S) \equiv g(\eta N_S+(1-\eta)N_B) -g(\zeta_{+}) -g(\zeta_{-}) \\ +(6\varepsilon' + 12\delta) g([(1-\eta)N_S+(1+\eta)N_B]/\delta)
	+ 3 g(\varepsilon') + 6 h_2(\delta)~,
	\end{multline}
	with 
	\begin{align}
	&\varepsilon = \sqrt{1-\eta^2/\left(\eta^2 + N_B(N_B+1)\lbrack1+3 \eta^2 - 2\eta (1+\sqrt{2\eta -1})\rbrack\right)}~,\\
	&\zeta_{\pm} = \frac{1}{2}\left(-1+\sqrt{[ (1+2N_B)^2 - 2\varrho + (1+2\vartheta)^2 \pm 4(\vartheta -N_B)\sqrt{[1+N_B+\vartheta]^2-\varrho}]/2}\right)~,\\
	&\varrho = 4N_B(N_B+1)(2-1/\eta)~,\\
	&\vartheta = \eta N_B+(1-\eta)N_S~,
	\end{align}
	$\varepsilon' \in (\varepsilon,1]$, and $\delta = (\varepsilon' - \varepsilon)/(1+ \varepsilon')$.
\end{theorem}
\begin{proof}
A proof follows from arguments similar to those in the proof of Theorem~\ref{thm:qu2}. The final result is obtained using Theorem~\ref{thm:pu2gen}.
\end{proof}

\subsection{$\varepsilon$-close-degradable bound on the energy-constrained private capacity of bosonic thermal channels}

\label{sec:eps-close-pri-cap}

\begin{theorem}\label{thm:eps-close-pri-cap-bound}
	An upper bound on the private capacity of a thermal channel $\L_{\eta, N_B}$ with transmissivity $\eta \in \lbrack1/2,1\rbrack$, environment photon number $N_B \geq 0$, and input mean photon number constraint $N_S\geq 0$  is given by
	\begin{multline}
	P(\L_{\eta, N_B},N_S) \leq P_{U_3}(\L_{\eta, N_B},N_S) \equiv g(\eta N_S)-g[(1-\eta)N_S] \\
	+ (8\varepsilon' + 16\delta) g[(\eta N_S +(1-\eta)N_B)/\delta] + 4g(\varepsilon') + 8h_2(\delta)~, \label{eq:eps-close-pri-cap-bound}
	\end{multline} 
	with $\varepsilon = N_B/(N_B+1)$, $\varepsilon' \in (\varepsilon,1]$, and $\delta = (\varepsilon' - \varepsilon)/(1+ \varepsilon')$.
\end{theorem}
\begin{proof}
A proof follows from  arguments similar to those in the proof of Theorem~\ref{thm:qu3}. The final result is obtained using Theorem~\ref{thm:eps-close-priv-gen-b}.
\end{proof}

\bigskip

\section{Lower bound on energy-constrained private capacity of bosonic thermal channels}

\label{sec:lower-bound-private-cap}

In this section, we establish an improvement on the best known lower bound \cite{Mark2012tradeoff} on the energy-constrained private capacity of bosonic thermal  channels, by using displaced thermal states as input to the channel. We note that a similar effect has been observed in \cite{RGK05} for the finite-dimensional case.

The energy-constrained private information of a channel $\N$, 
as defined in \eqref{eq:private-information}, can also be written as
\begin{equation}
P^{(1)}(\N, G, W) \equiv \sup_{\bar{\rho}_{\E_A}: \tr\{G\bar{\rho}_{\E_A}\}\leq W}[H(\N(\bar{\rho}_{\E_A})) -H(\hat{\N}(\bar{\rho}_{\E_A}))-\int dx~ p_X(x)[ H(\N(\rho^x_A))- H(\hat{\N}(\rho^x_A))]]~, \label{eq:private-info-lower}
\end{equation}
where $\bar{\rho}_{\E_A} \equiv \int dx~p_X(x) \rho^x_A$ is an average state of the ensemble 
$\E_{A} \equiv \{p_X(x), \rho^x_A   \}$
and $\hat{\N}$ denotes a complementary channel of $\N$. If the energy-constrained private information is calculated for coherent-state inputs, then for each element of the ensemble, the following equality holds $H(\N(\rho^x_A))= H(\hat{\N}(\rho^x_A))$. Hence, the entropy difference $ H(\N(\bar{\rho}_{\E_A})) -H(\hat{\N}(\bar{\rho}_{\E_A}))$ is an achievable rate, which is the same as the energy-constrained coherent information.

However, we show that displaced thermal-state inputs provide an improved lower bound for certain values of the transmissivity $\eta$, low thermal noise $N_B$, and both low and high input mean photon number $N_S$. We start with the following ensemble of displaced thermal states,
\begin{equation}
\E \equiv \{p_{N^1_S}(\alpha), D(\alpha)~\theta(N^2_S)~ D(-\alpha)  \},\label{eq:displaced-thermal-state}
\end{equation}
chosen according to the Gaussian probability distribution 
\begin{equation}
p_{N^1_S}(\alpha) = \frac{1}{\pi N^1_S} \exp(-\vert \alpha \vert^2/N^1_S),\label{eq:gaussian-prob-dist}
\end{equation}
where $D(\alpha)$ denotes the displacement operator, $\theta(N^2_S)$ denotes the thermal state with mean photon number $N^2_S$, and $N^1_S$ and $N^2_S$ are chosen such that $N^1_S + N^2_S = N_S$, which is the mean number of photons input to the channel. By employing \eqref{eq:displaced-thermal-is-thermal}, the average of this ensemble is a thermal state with mean photon number $N_S$, i.e.,
\begin{equation}
\bar{\rho}_{\E} = \int d^2\alpha ~p_{N^1_S}(\alpha) ~D(\alpha) ~\theta(N^2_S)~D(-\alpha) = \theta(N_S)~.
\end{equation}
 Hence, this ensemble meets the constraint that the average number of photons input to the channel is equal to $N_S$.

\begin{figure*}[ptb]
	\begin{center}
		\subfloat[]{\includegraphics[width=.42\columnwidth]{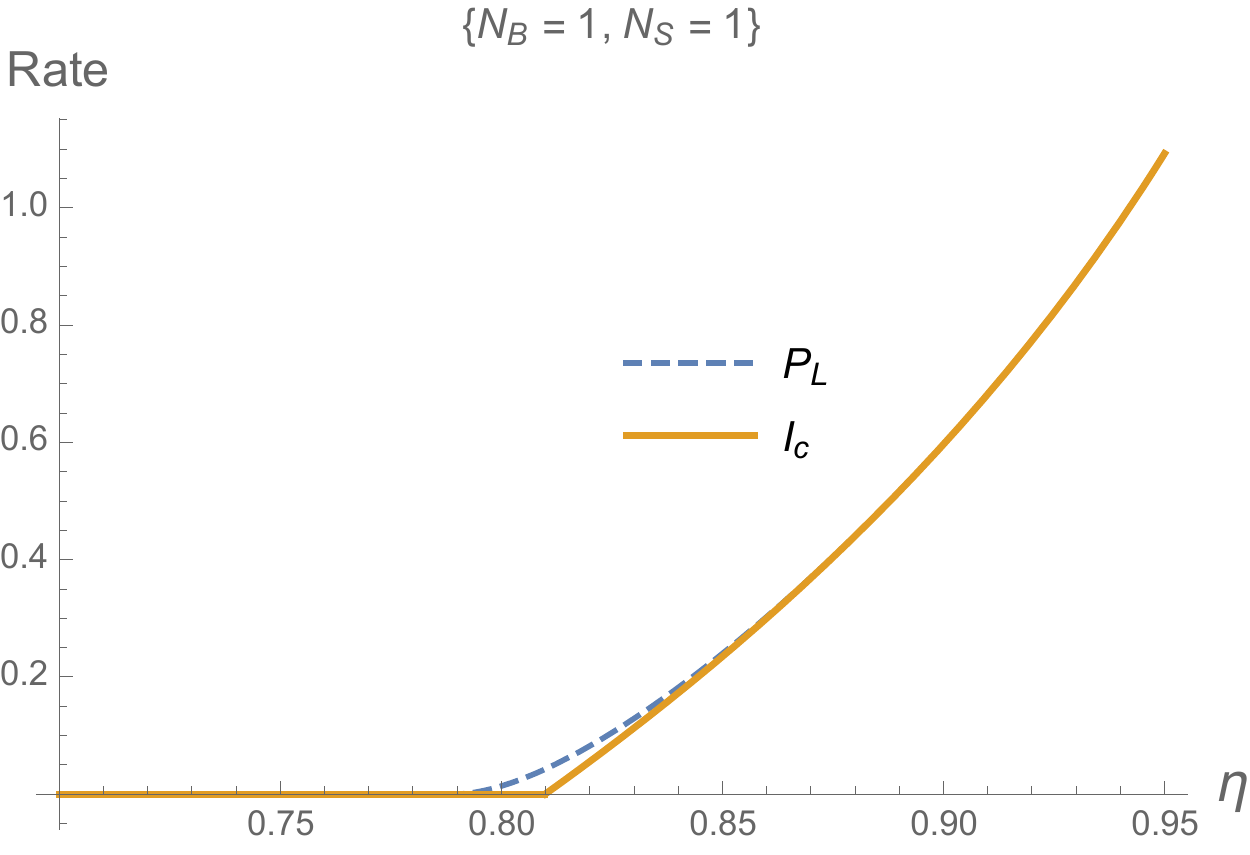}}\qquad
		\qquad\subfloat[]{\includegraphics[width=.42\columnwidth]{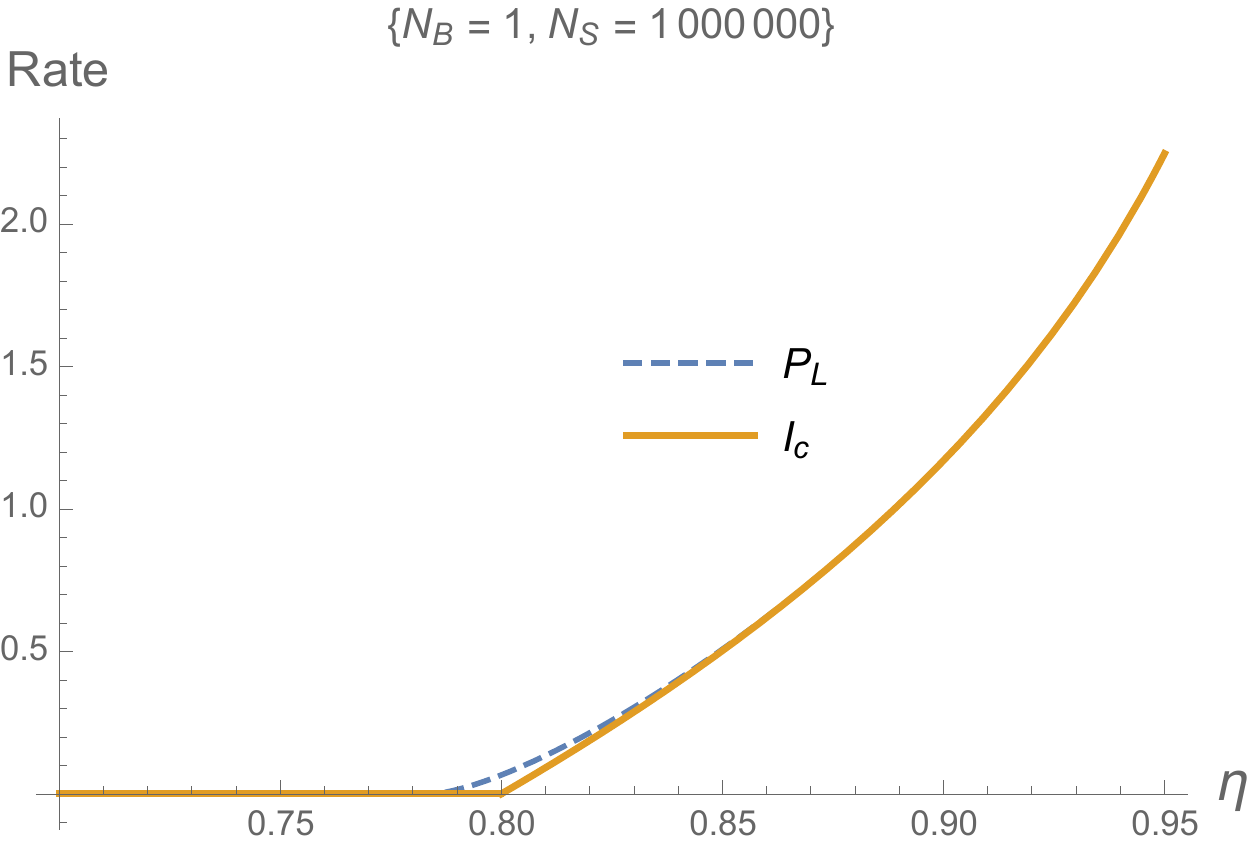}}%
	\end{center}
	\caption{The figures plot the optimized value of the lower bound on the private information $P_L(\L_{\eta, N_B}, N_S)$ (dashed line)  and coherent information $I_c(\L_{\eta, N_B}, N_S)$ (solid line) of a thermal channel versus transmissivity parameter $\eta$. In each figure, we select certain values of thermal noise $N_B$ and input mean photon number $N_S$, with the choices indicated above each figure. In all the cases, there is an improvement in the achievable rate of private communication for certain values of the transmissivity $\eta$.}%
	\label{fig:lower-bound-private-cap}%
\end{figure*}
 
 After the action of the channel on one of the states in the ensemble, the entropy of the output state is given by
 \begin{align}
H(\L_{\eta, N_B} (D(\alpha)~\theta(N^2_S)~D(-\alpha)))& = H(D(\sqrt{\eta}\alpha)\L_{\eta, N_B}(\theta(N^2_S))D(-\sqrt{\eta}\alpha))\\
& = H(\L_{\eta, N_B}(\theta(N^2_S)))~,
 \end{align}
where the first equality follows because thermal channel is covariant with respect to displacement operators, as reviewed in \eqref{eq:covariance-gaussian}. The second equality follows because $D(\sqrt{\eta}\alpha)$ is a unitary operator and entropy is invariant under the action of a unitary operator. Since $ H(\L_{\eta, N_B}(\theta(N^2_S)))$ is independent of the Gaussian probability distribution in (\ref{eq:gaussian-prob-dist}), we have that
\begin{equation}
\int d^2\alpha~ p_{N^1_S}(\alpha) H(\L_{\eta, N_B}(\theta(N^2_S))) =H(\L_{\eta, N_B}(\theta(N^2_S))) .
\end{equation}
Similar arguments can be made for the output states at the environment mode.

Hence, a lower bound on the energy-constrained private information in \eqref{eq:private-info-lower} for the bosonic thermal  channel is as follows:
\begin{align}
& P^{(1)}(\L_{\eta, N_B}, N_S)\nonumber \\
&  \geq H(\L_{\eta, N_B}(\theta(N_S))) - H(\hat{\L}_{\eta, N_B}(\theta(N_S))) - \lbrack H(\L_{\eta, N_B}(\theta(N^2_S)))-H(\hat{\L}_{\eta, N_B}(\theta(N^2_S)))\rbrack\\ 
& = I_c(\L_{\eta, N_B}, N_S) - I_c(\L_{\eta, N_B}, N^2_S) \equiv P_L(\L_{\eta, N_B}, N_S) \label{eq:private-information-final-exp}~,
\end{align}
where $\hat{\L}_{\eta, N_B}$ denotes the complementary channel of $\L_{\eta, N_B}$, and we denote the lower bound in \eqref{eq:private-information-final-exp} on the private information  by $P_L(\L_{\eta, N_B},N_S)$. The first inequality follows from \eqref{eq:private-information}. Here,  $I_c(\L_{\eta, N_B}, N_S)$ denotes the coherent information of the channel for the thermal state with mean photon number $N_S$ as input to the channel. $I_c(\L_{\eta, N_B}, N_S)$ has the same form as \eqref{eq:ql}, i.e., 
\begin{multline}
I_c(\L_{\eta, N_B}, N_S)=g(\eta N_S + (1-\eta) N_B) - g([D+(1-\eta)N_S - (1-\eta)N_B-1]/2)\\
- g([D-(1-\eta)N_S + (1-\eta)N_B-1]/2)~, \label{eq:coherent-info-private-cap}
\end{multline}
where $D^2 \equiv [(1+\eta)N_S+(1-\eta)N_B+1]^2 - 4\eta N_S(N_S+1)$. Similarly, $I_c(\L_{\eta, N_B}, N^2_S)$ is defined by replacing $N_S$ in \eqref{eq:coherent-info-private-cap} with $N^2_S$.

We optimize the lower bound in \eqref{eq:private-information-final-exp} on the private information $P_L(\L_{\eta, N_B}, N_S)$  with respect to $N^2_S$ for a fixed value of $N_S$ \cite{Mathematica}. In Figure \ref{fig:lower-bound-private-cap}, we plot the optimized value of the lower bound in \eqref{eq:private-information-final-exp} on the private information $P_L(\L_{\eta, N_B}, N_S)$ (dashed line) and the coherent information 
in \eqref{eq:coherent-info-private-cap} $I_c(\L_{\eta, N_B}, N_S)$ (solid line) of the thermal channel versus the transmissivity parameter $\eta$, for low thermal noise $N_B$ and for both low and high input mean number of photons $N_S$. We find that a larger rate for private communication can be achieved by using displaced thermal states as input to the channel instead of coherent states, for certain values of the transmissivity $\eta$.

\bigskip

\section{Upper bounds on energy-constrained quantum and private capacities of quantum amplifier channels} \label{sec:bounds-amp}

Using methods similar to those from Sections \ref{sec:upper-bound-quantum-cap} and \ref{sec:upper-bound-private-cap}, we now establish three different upper bounds on the energy-constrained quantum and private capacities of a noisy amplifier channel.

\subsection{Data-processing bound on energy-constrained quantum and private capacities of quantum amplifier channels}
\label{sec:data-proc-amp}

  In this section, we provide an upper bound using Theorem~\ref{thm:concat-loss-amp} below, which states that any phase-insensitive single-mode bosonic Gaussian channel can be decomposed as a pure-amplifier channel followed by a pure-loss channel, if the original channel is not entanglement breaking. This theorem was independently proven in \cite{RMG18, NAJ18} (see also \cite{notesSWAT17} in this context). 
  
  Before we state the theorem, let us recall that the action of a phase-insensitive channel $\N$ on the covariance matrix $\Gamma$ of a single-mode, bosonic quantum state is given by
  \begin{align} \label{eq:act-amp}
  \Gamma \longmapsto \tau~\Gamma + \nu I_2,
  \end{align}
  where $\nu$ is the variance of an additive noise, $I_2$ is the $2\times 2$ identity matrix, and $\tau$ and $\nu$ satisfy the conditions in \eqref{eq:cptp-cond-1}--\eqref{eq:cptp-cond}. Moreover, as mentioned previously, a phase-insensitive channel $\N$ is  entanglement-breaking 
  \cite{HSR03,Holevo2008}
  if 
  \begin{equation}\label{eq:ent-breaking-cond-phase-i}
  \tau+1\leq\nu~.
  \end{equation}
  
  \begin{theorem}
  \label{thm:concat-loss-amp}
  	Any single-mode, phase-insensitive bosonic Gaussian channel $\mathcal{N}$ that is not entanglement-breaking (i.e., satisfies $\tau + 1 > \nu$) can be decomposed as the concatenation of a quantum-limited amplifier channel $\A_{G,0}$ with gain $G>1$ followed by a pure-loss channel $\L_{\eta, 0}$ with transmissivity $\eta \in (0,1]$, i.e., 
  	\begin{equation} \label{eq:concat-LA}
  	\N = \L_{\eta, 0} \circ \A_{G,0},
  	\end{equation}
where $\eta = (\tau + 1 - \nu)/2$ and $G = \tau / \eta$.	
  \end{theorem}
\begin{proof}
The action of a quantum-limited amplifier channel $\A_{G,0}$ with gain $G$ followed by a pure-loss channel $\L_{\eta, 0}$ with transmissivity $\eta$, on the convariance matrix $\Gamma$ is given by
\begin{equation} \label{eq:act-concat}
\eta (G~\Gamma + [G-1] I_2) + [1-\eta] I_2 .
\end{equation}
By comparing \eqref{eq:act-amp} and \eqref{eq:act-concat}, we find that it is necessary for the following equalities to hold
\begin{align}
 \eta G & = \tau, \\
 \eta(G-1) + 1 - \eta & =  \nu .
\end{align}
Solving these equations for $\eta$ and $G$ in terms of $\tau$ and $\nu$ then gives $\eta = (\tau + 1 - \nu)/2$ and $G = \tau / \eta$.
By the assumption that $\mathcal{N}$ is not entanglement breaking, which is that
$\tau+1 > \nu$, we find that
\begin{equation}
\eta = (\tau + 1 - \nu)/2 > 0.
\end{equation}
Now applying the conditions in \eqref{eq:cptp-cond-1} and \eqref{eq:cptp-cond} for the channel $\N$ to be a CPTP map, we find that
\begin{equation}
 \eta = (\tau + 1 - \nu)/2 \leq 
 (\tau + 1 - |1-\tau|)/2 =
 \left\{
 \begin{array}{ll}
  \tau & \text{ for } \tau \in [0,1)\\
  1 & \text{ for } \tau \geq 1 
 \end{array}\right.
\end{equation}
By the fact that $G = \tau / \eta$, the above implies that $G > 1$, so that the decomposition in \eqref{eq:concat-LA} is valid under the stated conditions.
\end{proof}

We now apply Theorem~\ref{thm:concat-loss-amp} and a data-processing argument  to a noisy amplifier channel $\A_{G, N_B}$ with gain $G> 1$, environment photon number $N_B \geq 0$,
for which $\tau = G$ and $\nu = (G-1)(2 N_B + 1)$. This channel is entanglement-breaking when $(G-1) N_B \geq 1$ \cite{Holevo2008}.

\begin{theorem}\label{thm:qu1-amp}
An upper bound on the energy-constrained quantum and private capacities of a noisy amplifier channel $\A_{G, N_B}$ with gain $G> 1$ and  environment photon number $N_B \geq 0$, such that $(G-1) N_B < 1$, and input photon number constraint $N_S\geq 0$, is given by
\begin{equation}
Q(\A_{G, N_B}, N_S) , P(\A_{G, N_B}, N_S) \leq \operatorname{max}\{0, Q_{U_1}(\A_{G, N_B}, N_S) \},
\label{eq:qu1-amp}
\end{equation} 
where
\begin{align}
 Q_{U_1}(\A_{G, N_B}, N_S) & \equiv g(G' N_S + G' - 1) - g[(G'-1)(N_S + 1)], \\
G' & = G/(1+N_B(1-G)).
\end{align}
\end{theorem}
\begin{proof}
	An upper bound on the energy-constrained quantum and private capacities can be established by using \eqref{eq:concat-LA} and a data-processing argument. We find that 
	\begin{align}
	Q(\A_{G, N_B}, N_S) &= Q(\L_{\eta, 0} \circ \A_{G',0}, N_S) \\
	& \leq Q(\A_{G',0}, N_S) \\
	& = \operatorname{max}\{0, g(G'~N_S + G' - 1) - g[(G'-1)(N_S + 1)]\}~.
	\end{align}
	The first inequality follows from the definition and data processing---the energy-constrained capacity of $\L_{\eta, 0} \circ \A_{G',0}$ cannot exceed that of $\A_{G',0}$. The second equality follows from the formula for the energy-constrained quantum capacity of a quantum-limited amplifier channel with gain $G'$ and input mean photon number $N_S$ \cite{QW16}. Since a quantum-limited amplifier channel is a degradable channel \cite{CG06,WPG07}, its energy-constrained private capacity is the same as its energy-constrained quantum capacity. 
\end{proof}

\begin{remark}
	Applying Remark \ref{rem:unconstrained-cap}, we find the following data-processing bound $Q_{U_1}(\A_{G, N_B})$ on the unconstrained quantum and private capacities of amplifier channels for which $(G-1)N_B<1$:
	\begin{align}
	Q(\A_{G, N_B}), P(\A_{G, N_B}) \leq Q_{U_1}(\A_{G, N_B}) &= \sup_{N_S: N_S\in[0,\infty]} Q_{U_1}(\A_{G, N_B}, N_S)\\
	& = \lim_{N_S \to \infty} Q_{U_1}(\A_{G, N_B}, N_S) \\
	&= \log_{2}(G/(G-1)) - \log_2(N_B+1)~. \label{eq:unconstrained-cap-amp}
	\end{align}
	The second equality follows from the monotonicity of $Q_{U_1}(\A_{G, N_B}, N_S)$ with respect to $N_S$, which in turn  follows from the fact that the first derivative of $Q_{U_1}(\A_{G, N_B}, N_S)$ with respect to $N_S$ goes to zero as $N_S \to \infty$, and the second derivative is always negative.
	
	The bound
	\begin{equation}
	Q(\A_{G, N_B}), P(\A_{G, N_B}) \leq
	\log_2\!\left(\frac{G^{N_B+1}}{G-1}\right) - g(N_B) \label{eq:PLOBWTB-amp-bnd}
	\end{equation}
	was given in \cite{PLOB17,WTB17}.
	From a comparison of \eqref{eq:unconstrained-cap-amp} with \eqref{eq:PLOBWTB-amp-bnd}, we find that the bound given in \eqref{eq:PLOBWTB-amp-bnd} is always tighter than \eqref{eq:unconstrained-cap-amp}. Both the bounds in \eqref{eq:unconstrained-cap-amp} and  \eqref{eq:PLOBWTB-amp-bnd} converge to the true unconstrained quantum and private capacity in the limit as $N_B \to 0$, but \eqref{eq:PLOBWTB-amp-bnd} is tighter for $N_B > 0$.
\end{remark}

\begin{remark}
	The data-processing bound $Q_{U_1}(\A_{G, N_B}, N_S)$ on the energy-constrained quantum capacity of amplifier channels places a strong restriction on the channel parameters $G$ and $N_B$. Since the quantum capacity of a quantum-limited amplifier channel with gain $G'$ is non-zero only for $G' \neq  \infty$, the energy-constrained quantum capacity of an amplifier channel will be non-zero only for 
	\begin{align}
 1\leq G < (1+N_B)/N_B~,
	\end{align} 
	which is same as the condition given in \cite{dp2006} and is equivalent to the condition $(G-1)N_B < 1$, that the channel is not entanglement breaking.
\end{remark}
We now study the closeness of the data-proccessing bound $Q_{U_1}(\A_{G, N_B}, N_S)$ when compared to a known lower bound. In particular, we use the following lower bound on the energy-constrained quantum and private capacities of an amplifier channel \cite{HW01, MH16} and denote it by $Q_L(\A_{G, N_B},N_S)$:
\begin{multline}
Q(\A_{G, N_B}, N_S) \geq  Q_L(\A_{G, N_B}, N_S ) \equiv g(G N_S + (G-1) (N_B+1)) \\
- g([D+(G-1)(N_S+N_B+1)-1]/2)  
- g([D-(G-1)(N_S+N_B+1)-1]/2), \label{eq:ql-amp}
\end{multline}
where
\begin{equation}
D^2 \equiv [(1+G)N_S+(G-1)(N_B+1)+1]^2 - 4G N_S(N_S+1).
\end{equation}
\begin{theorem}
Let $\A_{G, N_B}$ be an amplifier channel with gain $G > 1$ and environment photon number $N_B\geq 0$, such that $(G-1)N_B < 1$, and input photon number constraint $N_S\geq 0$. Then the following relation holds between the data-processing bound $Q_{U_1}(\A_{G, N_B}, N_S)$ in \eqref{eq:qu1-amp} and the lower bound $Q_L(\A_{G, N_B},N_S)$ in \eqref{eq:ql-amp} on the energy-constrained quantum and private capacities of an amplifier channel:
\begin{align}
Q_L(\A_{G, N_B},N_S) \leq Q_{U_1}(\A_{G, N_B}, N_S) \leq Q_L(\A_{G, N_B}, N_S) + 1/\ln2~.
\end{align}
\end{theorem}
\begin{proof}
	A proof follows from arguments similar to those in the proof of Theorem \ref{thm:1.45bits}.
\end{proof}

\subsection{$\varepsilon$-degradable bound on energy-constrained quantum and private capacities of amplifier channels}

\label{sec:eps-deg-q-cap-amp}

In this section, we provide an upper bound on the energy-constrained quantum and private capacities of a quantum amplifier channel $\A_{G, N_B}$ using the idea of $\varepsilon$-degradability. We first construct an approximate degrading channel $\D$ by following arguments similar to those in Section \ref{sec:eps-deg-q-cap}. Furthermore, we introduce a particular channel that simulates the serial concatenation of the amplifier channel $\A_{G, N_B}$ and the approximate degrading channel $\D$. We finally provide an upper bound on the energy-constrained quantum and private capacities of an amplifier channel by using all these tools and invoking Theorem \ref{thm:qcbound-eps-approx}.

Similar to Section \ref{sec:eps-deg-q-cap}, we first establish an upper bound on the diamond distance between the complementary channel of the amplifier channel and the concatenation of the amplifier channel followed by a particular approximate degrading channel. Let $\T$ and $\T'$ represent transformations of two-mode squeezers with parameter $G$ and $(2G-1)/G$, respectively. In the Heisenberg picture, the unitary transformation corresponding to $\T$ and $\T'$ follow from \eqref{eq:amp-unitary}.

Consider the following action of the noisy amplifier channel $\A_{G, N_B}$ on an input state $\phi_{RA}$:
\begin{equation} \label{eq:amp-ch}
(\operatorname{id}_R\otimes \A_{G, N_B})(\phi_{RA}) = \tr_{E_1E_2}\{ \T_{AE'\to BE_2}(\phi_{RA}\otimes \psi_{\operatorname{TMS}}(N_B)_{E'E_1})  \}~, 
\end{equation}
where $R$ is a reference system and $\psi_{\operatorname{TMS}}(N_B)_{E'E_1}$ is a two-mode squeezed vacuum state with parameter $N_B$, as defined in \eqref{eq:tms}. It is evident from \eqref{eq:amp-ch} that the output of the noisy amplifier channel $\A_{G, N_B}$ is system $B$, and the outputs of the complementary channel $\hat{\A}_{G, N_B}$ are systems $E_1$ and $E_2$.

Consider a two-mode squeezer $\mathcal{T}'$ with parameter $(2G-1)/G$, such that the output of the amplifier channel $\A_{G, N_B}$ becomes an environmental input for $\mathcal{T}'$. We consider one mode of the two-mode squeezed vacuum state $\psi_{\operatorname{TMS}}(N_B)_{FE'_1}$ as an input for $\T'$, so that the subsystem $E'_1$ mimics $E_1$. We denote our choice of degrading channel by $\D_{(2G-1)/G, N_B}: \T(B) \to \T(E'_1)\otimes \T(E'_2)$. More formally, $\D_{(2G-1)/G, N_B}$ has the following action on the output state $\A_{G, N_B}(\phi_{RA})$:
\begin{equation}
\label{eq:deg-map-amp}
(\operatorname{id}_R \otimes[\D_{(2G-1)/G,N_B}\circ\A_{G, N_B}])(\phi_{RA}) = \tr_{G}\{\T'_{BF\to E'_2G}(\A_{G, N_B}(\phi_{RA})\otimes \psi_{\operatorname{TMS}}(N_B)_{FE'_1})  )  \}~.
\end{equation}

Now, similar to Section \ref{sec:eps-deg-q-cap}, we introduce a particular channel that simulates the action of $\D_{(2G-1)/G} \circ \A_{G, N_B}$ on an input state $\phi_{RA}$. We denote this channel by $\Lambda$, and it has the following action on an input state $\phi_{RA}$:
\begin{equation}\label{eq:simulating-ch-amp}
(\operatorname{id}_R \otimes \Lambda)(\phi_{RA}) = \tr_{B}\{\T_{AE' \rightarrow BE_2} (\phi_{RA}\otimes \omega(N_B)_{E'E_1}) \}~,
\end{equation}
where $\omega(N_B)_{E'E_1}$ represents a noisy version of a two-mode squeezed vacuum state with parameter $N_B$, and is same as \eqref{eq:noisy-tms-cov}, except $\eta$ is replaced by $G$. Similar to \eqref{eq:equiv}, the following equivalence holds for any quantum input state $\phi_{RA}$:
\begin{equation}
(\operatorname{id}_R \otimes [ \D_{(2G-1)/G, N_B}\circ\A_{G, N_B}])(\phi_{RA}) = (\operatorname{id}_R \otimes \Lambda)(\phi_{RA})~.
\end{equation}
Thus, the channels $\D_{(2G-1)/G,N_B}\circ\A_{G, N_B}$ and $\Lambda$ are indeed the same. 

Similar to Theorem \ref{thm:eps-diamond}, we now establish an upper bound on the diamond distance between the complementary channel of a noisy amplifier channel and the concatenation of the amplifier channel followed by the degrading channel in \eqref{eq:deg-map-amp}.

\begin{theorem}
	Fix $G> 1$. Let $\A_{G, N_B}$ be an amplifier channel with gain $G$, and let $\D_{(2G-1)/G, N_B}$ be a degrading channel as defined in \eqref{eq:deg-map-amp}. Then
	\begin{equation}
	\frac{1}{2} \Vert \hat{\A}_{G, N_B} - \D_{(2G-1)/G, N_B}\circ \A_{G, N_B} \Vert_{\diamond} \leq \sqrt{1 - G^2/\kappa(G, N_B)}, 
	\end{equation}
	with 
	\begin{equation}
	\kappa(G, N_B) = G^2 + N_B(N_B+1)[1+ 3G^2 -2G(1+\sqrt{2G-1})].
	\end{equation}
\end{theorem}
\begin{proof}
A proof follows from arguments similar to those in the proof of Theorem \ref{thm:eps-diamond}.
\end{proof}

\begin{theorem}\label{them:qu2-amp}
	An upper bound on the energy-constrained quantum capacity of a noisy amplifier channel $\A_{G, N_B}$ with gain $G> 1$, environment photon number $N_B$, such that $(G-1)N_B < 1$, and input mean photon-number constraint $N_S \geq 0$ is given by
	\begin{multline}
Q(\A_{G, N_B}, N_S) \leq Q_{U_2}(\A_{G, N_B}, N_S) \equiv g(G N_S + (G-1)N_B)-g(\zeta_{+}) -g(\zeta_{-}) \\ +(2\varepsilon' + 4\delta) g([(G-1)N_S+(1+G)N_B]/\delta)
+ g(\varepsilon') + 2 h_2(\delta)~,
\end{multline}
with 
\begin{align}
&\varepsilon = \sqrt{1-G^2/\left(G^2 + N_B(N_B+1)\lbrack1+3 G^2 - 2G (1+\sqrt{2G -1})\rbrack\right)}~,\\
&\zeta_{\pm} = \frac{1}{2}\left(-1+\sqrt{[ (1+2N_B)^2 - 2\varrho + (2\vartheta-1)^2 \pm 4(\vartheta -N_B-1)\sqrt{[N_B+\vartheta]^2-\varrho}]/2}\right)~,\\
&\varrho = 4N_B(N_B+1)(2G-1)/G~,\\
&\vartheta = G(1+N_B )+(G-1) N_S~,
\end{align}
$\varepsilon' \in (\varepsilon,1]$, and $\delta = (\varepsilon' - \varepsilon)/(1+ \varepsilon')$.
\end{theorem}
\begin{proof}
	A proof follows from arguments similar to those in the proof of Theorem \ref{thm:qu2}. 
\end{proof}
\bigskip 

\begin{theorem}\label{them:qu2-private-amp}
	An upper bound on the
	energy-constrained
	 private capacity of a noisy amplifier channel $\A_{G, N_B}$ with with gain $G> 1$, environment photon number $N_B$, such that $(G-1)N_B < 1$, and input mean photon-number constraint $N_S \geq 0$ is given by
	\begin{multline}
	P(\A_{G, N_B}, N_S) \leq P_{U_2}(\A_{G, N_B}, N_S) \equiv g(G N_S + (G-1)N_B)-g(\zeta_{+}) -g(\zeta_{-}) \\ +(6\varepsilon' + 12\delta) g([(G-1)N_S+(1+G)N_B]/\delta)
	+ 3g(\varepsilon') + 6 h_2(\delta)~,
	\end{multline}
	with 
	\begin{align}
	&\varepsilon = \sqrt{1-G^2/\left(G^2 + N_B(N_B+1)\lbrack1+3 G^2 - 2G (1+\sqrt{2G -1})\rbrack\right)}~,\\
	&\zeta_{\pm} = \frac{1}{2}\left(-1+\sqrt{[ (1+2N_B)^2 - 2\varrho + (2\vartheta-1)^2 \pm 4(\vartheta -N_B-1)\sqrt{[N_B+\vartheta]^2-\varrho}]/2}\right)~,\\
	&\varrho = 4N_B(N_B+1)(2G-1)/G~,\\
	&\vartheta = G(1+N_B )+(G-1) N_S~,
	\end{align}
	$\varepsilon' \in (\varepsilon,1]$, and $\delta = (\varepsilon' - \varepsilon)/(1+ \varepsilon')$.
\end{theorem}
\begin{proof}
	A proof follows from arguments similar to those in the proof of Theorem \ref{thm:qu2}. The final result is obtrained using Theorem \ref{thm:pu2gen}.
\end{proof}

\subsection{$\varepsilon$-close-degradable bound on energy-constrained quantum and private capacities of amplifier channels}
\label{sec:eps-close-q-cap-amp}

In this section, we first establish an upper bound on the diamond distance between a noisy amplifier channel and a quantum-limited amplifier channel. Since a quantum-limited amplifier channel is a degradable channel, an upper bound on the energy-constrained quantum capacity of a noisy amplifier channel directly follows from Theorem \ref{thm:qcbound-eps-close}.
\begin{theorem} \label{thm:amp-ql-amp}
	If a noisy amplifier channel $\A_{G, N_B}$  and a quantum-limited amplifier channel $\A_{G,0}$ have the same gain $G >  1$, then
	\begin{equation}
	\frac{1}{2}\left\Vert \A_{G, N_B} - \A_{G,0} \right\Vert_{\diamond} \leq \frac{N_B}{N_B+1}~.
	\end{equation}
\end{theorem}
\begin{proof}
	A proof follows from arguments similar to those in the proof of Theorem  \ref{thm:thermal-pureloss}.
\end{proof}
\bigskip

\begin{theorem}\label{thm:qu3-amp}
	An upper bound on the energy-constrained quantum capacity of a noisy amplifier channel $\A_{G, N_B}$ with gain $G> 1$, environment photon number $N_B$, such that $(G-1)N_B < 1$, and input mean photon-number constraint $N_S \geq 0$ is given by
	\begin{multline} \label{eq:qu3-amp}
	Q(\A_{G, N_B},N_S) \leq Q_{U_3}(\A_{G, N_B},N_S) \equiv g(G N_S + G-1)-g[(G-1)(N_S+1)] \\
	+ (4\varepsilon' + 8\delta) g[(GN_S +(G-1)N_B)/\delta] + 2g(\varepsilon') + 4h_2(\delta)~, 
	\end{multline} 
	with $\varepsilon = N_B/(N_B+1)$, $\varepsilon' \in (\varepsilon,1]$ and $\delta = (\varepsilon' - \varepsilon)/(1+ \varepsilon')$.
\end{theorem}
\begin{proof}
	A proof follows from arguments similar to those in the proof of Theorem \ref{thm:qu3}. 
\end{proof}

\begin{theorem}\label{thm:qu3-private-amp}
	An upper bound on the energy-constrained private capacity of a noisy amplifier channel $\A_{G, N_B}$ with gain $G> 1$, environment photon number $N_B$, such that $(G-1)N_B < 1$, and input mean photon-number constraint $N_S \geq 0$ is given by
	\begin{multline} \label{eq:qu3-private-amp}
	P(\A_{G, N_B},N_S) \leq P_{U_3}(\A_{G, N_B},N_S) \equiv g(G N_S + G-1)-g[(G-1)(N_S+1)] \\
	+ (8\varepsilon' + 16\delta) g[(GN_S +(G-1)N_B)/\delta] + 4g(\varepsilon') + 8h_2(\delta)~, 
	\end{multline} 
	with $\varepsilon = N_B/(N_B+1)$, $\varepsilon' \in (\varepsilon,1]$ and $\delta = (\varepsilon' - \varepsilon)/(1+ \varepsilon')$.
\end{theorem}
\begin{proof}
	A proof follows from arguments similar to those in the proof of Theorem \ref{thm:qu3}. The final result is obtrained using Theorem \ref{thm:eps-close-priv-gen-b}.
\end{proof}
\bigskip

\section{Data-processing bound on energy-constrained quantum and private capacities of additive-noise channels} \label{sec:bound-additive-noise}

In this section, we provide an upper bound on the energy-constrained quantum and private capacities of an additive-noise channel using Theorem \ref{thm:qu1}. Note that we only consider $\bar{n} \in (0,1)$ because the additive-noise channel is not entanglement breaking in this interval \cite{Holevo2008}.

\begin{theorem}\label{thm:bound-additive-noise}
	An upper bound on the energy-constrained quantum and private capacities of an additive-noise channel $\N_{\bar{n}}$ with noise parameter $\bar{n}\in(0,1)$, and input mean photon number constraint $N_S$ is given by
	\begin{equation}
	Q(\N_{\bar{n}}, N_S) ,P(\N_{\bar{n}}, N_S)\leq \operatorname{max}\{ 0, Q_{U_1}(\N_{\bar{n}}, N_S)\},
		\end{equation}
	where
		\begin{equation}
	Q_{U_1}(\N_{\bar{n}},N_S) \equiv g(N_S/(\bar{n}+1)) - g(\bar{n}N_S/(\bar{n}+1)). \label{eq:constrained-cap-additive-n}
		\end{equation}
\end{theorem} 
\begin{proof}
	A proof follows from the fact that an additive noise channel can be obtained from a thermal noise channel in the limit $\eta\to 1$ and $N_B\to\infty$, with $(1-\eta)N_B \to \bar{n}$ \cite{GGLMS04}, as well by applying the continuity results for these capacities from \cite[Theorem~3]{S17} (see also \cite{Wetal17}). By taking these limits in \eqref{eq:qu1}, we obtain the desired result.   
\end{proof} 
 
 \begin{remark}
 	Applying Remarks \ref{rem:unconstrained-cap} and \ref{rem:unconstrained-qcap-thermal}, and Theorem \ref{thm:bound-additive-noise}, we find the following data-processing bound $Q_{U_1}(\N_{\bar{n}})$ on the unconstrained quantum and private capacities of additive-noise channels
	for $\bar{n} \in (0,1)$:
\begin{equation}
Q_{U_1}(\N_{\bar{n}}) = \log_2(1/\bar{n})~. \label{eq:unconstrained-cap-additive-n}
\end{equation}
 \end{remark}

\begin{remark}
	From Theorem \ref{thm:1.45bits}, it follows that the data-processing upper bound $Q(\N_{\bar{n}}, N_S)$ can be at most 1.45 bits larger than a known lower bound on the energy-constrained quantum and private capacities of an additive-noise channel. 
\end{remark}

\begin{remark}
The following bound was given in 
\cite{PLOB17,WTB17} for  $\bar{n} \in (0,1)$:
\begin{equation}
Q(\N_{\bar{n}}) ,P(\N_{\bar{n}})
\leq
\frac{\bar{n}-1}{\ln 2} + \log_2(1/\bar{n}). \label{eq:add-noise-PLOTWTB-bnd}
\end{equation}
	From a comparison of \eqref{eq:unconstrained-cap-additive-n} with the bound in \eqref{eq:add-noise-PLOTWTB-bnd}, we find that the bound in \eqref{eq:add-noise-PLOTWTB-bnd} is always tighter than \eqref{eq:unconstrained-cap-additive-n}.
\end{remark}

 \section{Recent developments}\label{sec:recent-dev}
 
 In this section, we first recall a recent result of \cite{RMG18} on the unconstrained quantum capacity of a thermal channel. After that, we extend these results here to obtain new bounds on the energy-constrained quantum and private capacities of a thermal channel and an additive-noise channel. Finally, we compare these new bounds with our previous bounds. 
 
 \subsection{\cite{RMG18} bounds for the unconstrained quantum capacity of a thermal channel}
 
 Recently, the following upper bound on the unconstrained quantum capacity of a thermal channel for which
 $\eta > (1-\eta)N_B$ was introduced in \cite[Eq.~(40)]{RMG18}:
 \begin{equation}
 Q_{U_1}(\L_{\eta, N_B}) = \max\left\{0, \log_2\!\left(
 \frac{\eta - (1-\eta)N_B} {(1-\eta)(N_B+1)}\right) \right\}~. \label{eq:rmg}
 \end{equation}
This bound was obtained by using the decomposition
$\L_{\eta, N_B} = 
\L_{\eta', 0}\circ\A_{G,0}$
from Theorem \ref{thm:concat-loss-amp} (found independently in \cite{RMG18}) and the bottleneck inequality $Q(\L_{\eta', 0}\circ\A_{G,0}) \leq \min\{Q(\L_{\eta', 0}), Q(\A_{G,0}) \}$ for the unconstrained quantum capacity. Note that \eqref{eq:rmg} is slightly tighter than \eqref{eq:unconstrained-qcap-thermal} for all parameter regimes. These findings were independently discovered in \cite{NAJ18}.

\subsection{Further extension to the energy-constrained quantum and private capacities of thermal channels}

We now introduce a new upper bound on the energy-constrained quantum and private capacities of thermal channels (independently discovered in \cite{NAJ18} as well).  In the energy-constrained scenario, one cannot directly apply the bottleneck inequality in order to obtain a bound for the finite-energy case, due to an important physical consideration discussed below. However, we introduce a method to tackle this issue and establish an upper bound in the following theorem:

\begin{theorem}\label{thm:qu4}
An upper bound on the
energy-constrained
quantum and private capacities of a thermal channel $\L_{\eta, N_B}$ with transmissivity $\eta\in[1/2,1]$, environment photon number $N_B\geq 0$,
such that $\eta > (1-\eta)N_B$, and input mean photon number constraint $N_S\geq 0$ is given by
\begin{equation}
Q(\L_{\eta, N_B}, N_S),
P(\L_{\eta, N_B}, N_S)
 \leq \max\{0, Q_{U_4} (\L_{\eta, N_B}, N_S)\},
\end{equation}
where
\begin{equation}
Q_{U_4}(\L_{\eta, N_B}, N_S) \equiv g(\eta N_S + (1-\eta)N_B) - g[(1/\eta'-1)(\eta N_S + (1-\eta)N_B)],\label{eq:qu4}
\end{equation}
and $\eta' = \eta-(1-\eta)N_B$.
\end{theorem} 
\begin{proof}
Using Theorem \ref{thm:concat-loss-amp}, a  thermal channel $\L_{\eta, N_B}$ satisfying $\eta > (1-\eta)N_B$ can be decomposed as the concatenation of a quantum-limited amplifier channel $\A_{G, 0}$ followed by a pure-loss channel $\L_{\eta', 0}$, such that 
\begin{align}
& G = \eta/\eta',\\
& \eta' = \eta-(1-\eta)N_B~.
\end{align}
Consider the following chain of inequalities:
\begin{align}
Q(\L_{\eta, N_B}, N_S) &= Q(\L_{\eta', 0} \circ \A_{G, 0}, N_S)\\
& \leq Q(\L_{\eta',0}, G N_S + G-1)\\
& = g(\eta' [G N_S + G-1]) - g[(1-\eta')(G N_S + G-1)]\\
& = g(\eta N_S + (1-\eta)N_B) - g[(1/\eta'-1)(\eta N_S + (1-\eta)N_B)]~. 
\end{align}
The first inequality is a consequence of the following argument: consider an arbitrary encoding and decoding scheme for energy-constrained quantum communication over the thermal channel~$\L_{\eta, N_B}$, which satisfies the mean input photon number constraint $N_S \geq 0$. Due to the decomposition of $\L_{\eta, N_B}$ as $\L_{\eta', 0} \circ \A_{G, 0}$, this encoding, followed by many uses of the pure-amplifier channel 
$\A_{G, 0}$
 can be considered as an encoding for the channel $\L_{\eta',0}$, which also satisfies the mean photon number constraint $G N_S + G-1$, due to the fact that the pure-amplifier channel $\A_{G, 0}$ introduces a gain. Since the energy-constrained quantum capacity of the channel $\L_{\eta',0}$ involves an optimization over all such encodings that satisfies the mean photon number constraint $ G N_S + G-1$, we arrive at the desired inequality. The second equality follows from the formula for the energy-constrained quantum capacity of a pure-loss bosonic channel with transmissivity $\eta'$ and input mean photon number $GN_S + G-1$ \cite{Mark2012tradeoff,MH16}. 
\end{proof}

\bigskip
Now, we conduct a numerical evaluation in order to compare the bound in \eqref{eq:qu4} with our other bounds on the energy-constrained quantum and private capacities of a thermal channel. Since there is a free parameter $\varepsilon'$ in both the $\varepsilon$-degradable bound in \eqref{eq:qu2} and the $\varepsilon$-close-degradable bound in \eqref{eq:qu3}, we optimize these bounds with respect to $\varepsilon'$ \cite{Mathematica}. In Figure \ref{fig:recent-develop}(a), we find that both the data-processing bound $Q_{U_1}$ and $Q_{U_4}$ are close to the lower bound for medium transimissivity and high thermal noise. Moreover, $Q_{U_4}$ is slightly tighter than $Q_{U_1}$ for some parameter regimes. In Figure \ref{fig:recent-develop}(b), we find that the $\varepsilon$-degradable bound is tighter than all other bounds for high transmissivity and high thermal noise. 
\begin{figure*}[ptb]
	\begin{center}
		\subfloat[]{\includegraphics[width=.42\columnwidth]{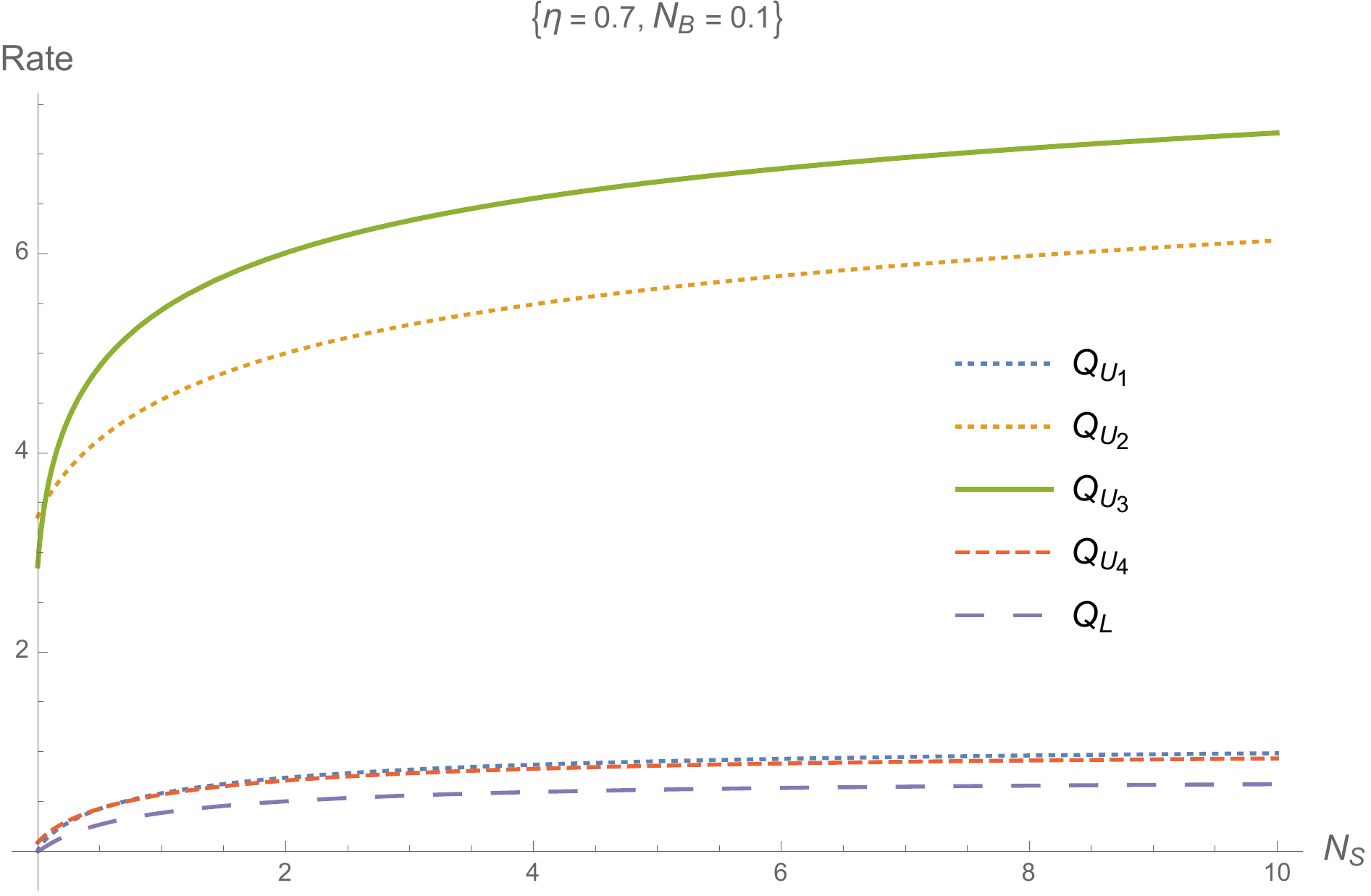}}\qquad
		\qquad\subfloat[]{\includegraphics[width=.42\columnwidth]{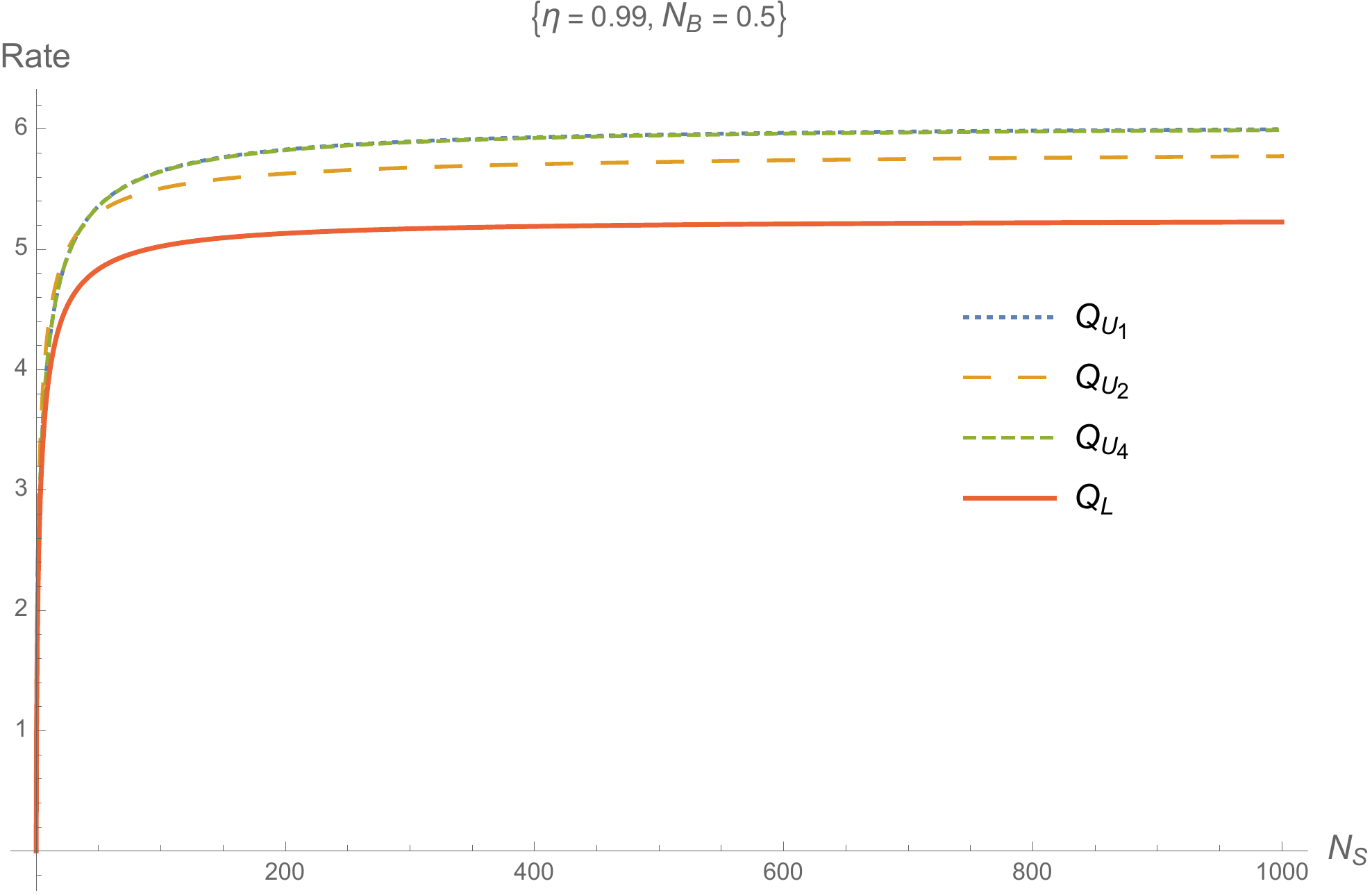}}%
	\end{center}
	\caption{The figures plot the data-processing bound $(Q_{U_1})$, the $\varepsilon$-degradable bound $(Q_{U_2})$, the $\varepsilon$-close-degradable bound $(Q_{U_3})$, the bound $Q_{U_4}$ and the lower bound $(Q_{L})$ on energy-constrained quantum capacity of thermal channels. In each figure, we select certain values of $\eta$ and $N_B$, with the choices indicated above each figure. In (a), for medium transmissivity and high thermal noise, both the data-processing bound and $Q_{U_4}$ is close to the lower bound. In (b), the $\varepsilon$-degradable upper bound is tighter than all other upper bounds ($Q_{U_3}$ is not plotted because it is much higher than the other bounds for all parameter values).}%
	\label{fig:recent-develop}%
\end{figure*}

\begin{remark}
The upper bound $Q_{U_4}(\L_{\eta, N_B}, N_S)$ on the energy-constrained quantum and private capacities of thermal channels places a strong restriction on the channel parameters $\eta$ and $N_B$. Since the quantum and private capacities of a pure-loss channel with $\eta'$ are non-zero only for $\eta' > 1/2$, the energy-constrained quantum and private capacities of a thermal channel will be non-zero only for 
\begin{equation}
1 \geq \eta > \frac{1+2N_B}{2(1+N_B)},
\end{equation}
which is same as the condition given in \cite[Section 4]{dp2006}.
\end{remark}

\subsection{Further extension on the energy-constrained quantum and private capacities of additive-noise channels}
 
In this section, we establish another upper bound on the energy-constrained quantum and private capacities of an additive-noise channel, by using Theorem \ref{thm:qu4}.

\begin{theorem}\label{thm:qu4-add-noise}
	An upper bound on the
	energy-constrained 
	quantum and private capacities of an additive-noise channel $N_{\bar{n}}$ with noise parameter $\bar{n} \in (0,1)$, and input photon number constraint $N_S \geq 0$ is given by
	\begin{equation}
	Q(\N_{\bar{n}}, N_S) , P(\N_{\bar{n}}, N_S) \leq \operatorname{max}\{0, Q_{U_4}(\N_{\bar{n}}, N_S) \},
	\end{equation}
	where
	\begin{equation}
	Q_{U_4}(\N_{\bar{n}}, N_S) \equiv g(N_S + \bar{n}) - g[\bar{n} (N_S+ \bar{n})/(1-\bar{n}) ]. \label{eq:qu4-add-noise}
	\end{equation}
\end{theorem}
\begin{proof}
A proof follows from arguments similar to those in the proof of Theorem \ref{thm:bound-additive-noise}. The final result is obtained using Theorem \ref{thm:qu4}. 	
\end{proof}	

\begin{remark}
	From a comparision of \eqref{eq:qu4-add-noise} and \eqref{eq:constrained-cap-additive-n}, we find that $Q_{U_1}(\N_{\bar{n}}, N_S)$ is tighter than $Q_{U_4}(\N_{\bar{n}}, N_S)$ only for low noise and low input mean photon number. The bound $Q_{U_4}(\N_{\bar{n}}, N_S)$ is tighter than  $Q_{U_1}(\N_{\bar{n}}, N_S)$ for all other parameter regimes. 
\end{remark}

 \begin{remark}
	Applying Remarks \ref{rem:unconstrained-cap} and \ref{rem:unconstrained-qcap-thermal}, and Theorem \ref{thm:qu4-add-noise}, we find the following data-processing bound $Q_{U_4}(\N_{\bar{n}})$ on the unconstrained quantum and private capacities of additive-noise channels:
	\begin{equation}
	Q_{U_4}(\N_{\bar{n}}) = \log_2[(1-\bar{n})/\bar{n}]~. \label{eq:unconstrained-qu4-add-noise}
	\end{equation}
\end{remark}

\begin{remark}
	From a comparison of \eqref{eq:unconstrained-qu4-add-noise} with the bound in \eqref{eq:add-noise-PLOTWTB-bnd}, we find that \eqref{eq:unconstrained-qu4-add-noise} is tighter than \eqref{eq:add-noise-PLOTWTB-bnd} for high noise.
\end{remark}

 \section{On the optimization of generalized channel divergences of quantum Gaussian channels}\label{sec:generalized-channel-divergence}
 	
	In this section, we address the question of computing the energy-constrained diamond norm of several channels of interest that have appeared in our paper. We provide a very general argument, based on some definitions and results in \cite{LKDW17} and phrased in terms of the ``generalized channel divergence'' as a measure of the distinguishability of quantum channels. We find that, among all Gaussian input states with a fixed energy constraint, the two-mode squeezed vacuum state saturating the energy constraint is the optimal state for the energy-constrained generalized channel divergence of two particular Gaussian channels. We describe these results in more detail in what follows.
	
 	We begin by recalling some developments from \cite{LKDW17}:
 	
 	\begin{definition}
 		[Generalized divergence \cite{SW12,WWY13}]A functional $\mathbf{D}:\mathcal{D}(\H)\times
 		\mathcal{D}(\H)\rightarrow\mathbb{R}$ is a \textit{generalized divergence} if it
 		satisfies the monotonicity (data processing) inequality%
 		\begin{equation}
 		\mathbf{D}(\rho\Vert\sigma)\geq\mathbf{D}(\mathcal{N}(\rho)\Vert
 		\mathcal{N}(\sigma)),
 		\end{equation}
 		where $\mathcal{N}$ is a quantum channel.
 	\end{definition}
 	
 	Particular examples of a generalized divergence are the trace distance,
 	quantum relative entropy, and the negative root fidelity.
 	
 	We say that a generalized channel divergence possesses the direct-sum property  on classical--quantum
 	states if the following equality holds:%
 	\begin{equation}
 	\mathbf{D}\!\left(  \sum_{x}p_{X}(x)|x\rangle\langle x|_{X}\otimes\rho
 	^{x}\middle\Vert\sum_{x}p_{X}(x)|x\rangle\langle x|_{X}\otimes\sigma
 	^{x}\right)  =\sum_{x}p_{X}(x)\mathbf{D}(\rho^{x}\Vert\sigma^{x}),
 	\end{equation}
 	where $p_{X}$ is a probability distribution, $\{|x\rangle\}_{x}$ is an
 	orthonormal basis, and $\{\rho^{x}\}_{x}$ and $\{\sigma^{x}\}_{x}$ are sets of
 	states. We note that this property holds for trace distance, quantum relative
 	entropy, and the negative root fidelity.
 	
 	\begin{definition}
 		[Generalized channel divergence \cite{LKDW17}]Given quantum channels $\mathcal{N}%
 		_{A\rightarrow B}$ and $\mathcal{M}_{A\rightarrow B}$, we define the
 		generalized channel divergence as%
 		\begin{equation}
 		\mathbf{D}(\mathcal{N}\Vert\mathcal{M})\equiv\sup_{\rho_{RA}}\mathbf{D}%
 		((\operatorname{id}_{R}\otimes\mathcal{N}_{A\rightarrow B})(\rho_{RA}%
 		)\Vert(\operatorname{id}_{R}\otimes\mathcal{M}_{A\rightarrow B})(\rho_{RA})).
 		\end{equation}
 		In the above definition, the supremum is with respect to all mixed states and
 		the reference system $R$ is allowed to be arbitrarily large. However, as a
 		consequence of purification, data processing, and the Schmidt decomposition,
 		it follows that%
 		\begin{equation}
 		\mathbf{D}(\mathcal{N}\Vert\mathcal{M})=\sup_{\psi_{RA}}\mathbf{D}%
 		((\operatorname{id}_{R}\otimes\mathcal{N}_{A\rightarrow B})(\psi_{RA}%
 		)\Vert(\operatorname{id}_{R}\otimes\mathcal{M}_{A\rightarrow B})(\psi_{RA})),
 		\end{equation}
 		such that the supremum can be restricted to be with respect to pure states and
 		the reference system $R$ isomorphic to the channel input system $A$.
 	\end{definition}
 	
 	Particular cases of the generalized channel divergence are the diamond norm of
 	the difference of $\mathcal{N}_{A\rightarrow B}$ and $\mathcal{M}%
 	_{A\rightarrow B}$ as well as the R\'enyi channel divergence from \cite{CMW14}.
 	
 	Covariant quantum channels have symmetries that allow us to simplify the set
 	of states over which we need to optimize their generalized channel divergence \cite{Hol02}.
 	Let $G$ be a finite group, and for every $g\in G$, let $g\rightarrow U_{A}(g)$
 	and $g\rightarrow V_{B}(g)$ be unitary representations acting on the input and
 	output spaces of the channel, respectively. Then a\ quantum channel
 	$\mathcal{N}_{A\rightarrow B}$\ is covariant with respect to $\left\{  \left(
 	U_{A}(g),V_{B}(g)\right)  \right\}  _{g}$\ if the following relation holds for
 	all input density operators $\rho_{A}$ and group elements $g\in G$:%
 	\begin{equation}
 	\left(  \mathcal{N}_{A\rightarrow B}\circ\mathcal{U}_{A}^{g}\right)  \left(
 	\rho_{A}\right)  =\left(  \mathcal{V}_{B}^{g}\circ\mathcal{N}_{A\rightarrow
 		B}\right)  (\rho_{A}),
 	\end{equation}
 	where%
 	\begin{align}
 	\mathcal{U}_{A}^{g}(\rho_{A})  &  =U_{A}(g)\rho_{A}U_{A}^{\dag}(g),\\
 	\mathcal{V}_{B}^{g}(\sigma_{B})  &  =V_{B}(g)\sigma_{B}V_{B}^{\dag}(g).
 	\end{align}
 	We say that channels $\mathcal{N}_{A\rightarrow B}$ and $\mathcal{M}%
 	_{A\rightarrow B}$ are \textit{jointly covariant} with respect to $\left\{  \left(
 	U_{A}(g),V_{B}(g)\right)  \right\}  _{g\in G}$ if each of them is covariant
 	with respect to $\left\{  \left(  U_{A}(g),V_{B}(g)\right)  \right\}  _{g}$ \cite{2016channel-discrimination,DW17}.
 	
 	The following lemma was established in \cite{LKDW17}:
 	
 	\begin{lemma}
 		[\cite{LKDW17}]\label{lemma:cov-critical-step}Let $\mathcal{N}_{A\rightarrow
 			B}$ and $\mathcal{M}_{A\rightarrow B}$ be quantum channels, and let $\left\{
 		\left(  U_{A}(g),V_{B}(g)\right)  \right\}  _{g\in G}$ denote unitary
 		representations of a group $G$. Let $\rho_{A}$ be a density operator, and let
 		$\phi_{RA}^{\rho}$ be a purification of $\rho_{A}$. Let $\bar{\rho}_{A}$
 		denote the group average of $\rho_{A}$ according to a distribution $p_{G}$,
 		i.e.,%
 		\begin{equation}
 		\bar{\rho}_{A}=\sum_{g}p_{G}(g)\ \mathcal{U}_{A}^{g}(\rho_{A}),
 		\end{equation}
 		and let $\phi_{RA}^{\bar{\rho}}$ be a purification of $\bar{\rho}_{A}$. If the
 		generalized divergence possesses the direct-sum property on classical--quantum states, then the
 		following inequality holds%
 		\begin{multline}
 		\mathbf{D}(\mathcal{N}_{A\rightarrow B}(\phi_{RA}^{\bar{\rho}})\Vert
 		\mathcal{M}_{A\rightarrow B}(\phi_{RA}^{\bar{\rho}}))\\
 		\geq\sum_{g}p_{G}(g)\mathbf{D}\!\left(  \left(  \mathcal{V}_{B}^{g\dag}%
 		\circ\mathcal{N}_{A\rightarrow B}\circ\mathcal{U}_{A}^{g}\right)  (\phi
 		_{RA}^{\rho})\middle\Vert\left(  \mathcal{V}_{B}^{g\dag}\circ\mathcal{M}%
 		_{A\rightarrow B}\circ\mathcal{U}_{A}^{g}\right)  (\phi_{RA}^{\rho})\right)  .
 		\end{multline}
 		
 	\end{lemma}
 	
 	By approximation, the above lemma can be extended to continuous groups for
 	several generalized channel divergences of interest:
 	
 	\begin{lemma}
 		\label{lemma:cov-critical-step-cont}Let $\mathcal{N}_{A\rightarrow B}$ and
 		$\mathcal{M}_{A\rightarrow B}$ be quantum channels, and let $\left\{  \left(
 		U_{A}(g),V_{B}(g)\right)  \right\}  _{g\in G}$ denote unitary representations
 		of a continuous group $G$. Let $\rho_{A}$ be a density operator, and let
 		$\phi_{RA}^{\rho}$ be a purification of $\rho_{A}$. Let $\bar{\rho}_{A}$
 		denote the group average of $\rho_{A}$ according to a measure $\mu(g)$, i.e.,%
 		\begin{equation}
 		\bar{\rho}_{A}=\int d\mu(g)\ \mathcal{U}_{A}^{g}(\rho_{A}),
 		\end{equation}
 		and let $\phi_{RA}^{\bar{\rho}}$ be a purification of $\bar{\rho}_{A}$. If the
 		generalized divergence possesses the direct-sum property on classical--quantum states and is a Borel function, then the
 		following inequality holds%
 		\begin{multline}
 		\mathbf{D}(\mathcal{N}_{A\rightarrow B}(\phi_{RA}^{\bar{\rho}})\Vert
 		\mathcal{M}_{A\rightarrow B}(\phi_{RA}^{\bar{\rho}}))\\
 		\geq\int d\mu(g)\ \mathbf{D}\!\left(  \left(  \mathcal{V}_{B}^{g\dag}%
 		\circ\mathcal{N}_{A\rightarrow B}\circ\mathcal{U}_{A}^{g}\right)  (\phi
 		_{RA}^{\rho})\middle\Vert\left(  \mathcal{V}_{B}^{g\dag}\circ\mathcal{M}%
 		_{A\rightarrow B}\circ\mathcal{U}_{A}^{g}\right)  (\phi_{RA}^{\rho})\right)  .
 		\end{multline}
 		
 	\end{lemma}
 	
 	We can apply this lemma effectively in the context of quantum Gaussian
 	channels. To this end, we consider an
	\textbf{energy-constrained generalized channel
 	divergence} for $W\in\lbrack0,\infty)$ and an energy observable $G$ as follows:%
 	\begin{equation}
 	\mathbf{D}_{G,W}(\mathcal{N}\Vert\mathcal{M})=\sup_{\psi_{RA}%
 		\ :\ \operatorname{Tr}\{G\psi_{A}\}\leq W}\mathbf{D}((\operatorname{id}%
 	_{R}\otimes\mathcal{N}_{A\rightarrow B})(\psi_{RA})\Vert(\operatorname{id}%
 	_{R}\otimes\mathcal{M}_{A\rightarrow B})(\psi_{RA})).
 	\end{equation}
 	In what follows, we specialize this measure even further to the \textbf{Gaussian
 	energy-constrained generalized channel divergence}, meaning that the
 	optimization is constrained to be with respect to Gaussian input states:%
 	\begin{equation}
 	\mathbf{D}_{G,W}^{\mathcal{G}}(\mathcal{N}\Vert\mathcal{M})=\sup_{\psi
 		_{RA}\ :\ \operatorname{Tr}\{G\psi_{A}\}\leq W,\ \psi_{RA}\in\mathcal{G}%
 	}\mathbf{D}((\operatorname{id}_{R}\otimes\mathcal{N}_{A\rightarrow B}%
 	)(\psi_{RA})\Vert(\operatorname{id}_{R}\otimes\mathcal{M}_{A\rightarrow
 		B})(\psi_{RA})),
 	\end{equation}
 	where $\mathcal{G}$ denotes the set of Gaussian states. We then establish the
 	following proposition:
 	
 	\begin{proposition}
\label{prop:gaussian-gen-div}Suppose that channels $\mathcal{N}_{A\rightarrow
B}$ and $\mathcal{M}_{A\rightarrow B}$ are Gaussian, they each take one input
mode to $m$ output modes, and they have the following action on a single-mode,
input covariance matrix $V$:
\begin{align}
V &  \rightarrow XVX^{T}+Y_{\mathcal{N}},\label{eq:Gaussian-chan-N}\\
V &  \rightarrow XVX^{T}+Y_{\mathcal{M}},\label{eq:Gaussian-chan-M}%
\end{align}
where $X$ is an $m\times1$ matrix, $Y_{\mathcal{N}}$ and $Y_{\mathcal{M}}$ are
$m\times m$ matrices such that $\mathcal{N}_{A\rightarrow B}$ and
$\mathcal{M}_{A\rightarrow B}$ are legitimate Gaussian channels. Suppose
furthermore they these channels are jointly phase covariant
(phase-insensitive), in the sense that for all $\phi\in\lbrack0,2\pi)$ and
input density operators $\rho$, the following equality holds%
\begin{align}
\mathcal{N}_{A\rightarrow B}(e^{i\hat{n}\phi}\rho e^{-i\hat{n}\phi}) &
=\left(  \bigotimes\limits_{i=1}^{m}e^{i\hat{n}_{i}\left(  -1\right)  ^{a_{i}%
}\phi}\right)  \mathcal{N}_{A\rightarrow B}(\rho)\left(  \bigotimes
\limits_{i=1}^{m}e^{-i\hat{n}_{i}\left(  -1\right)  ^{a_{i}}\phi}\right)
,\label{eq:phase-covariance-1-chan-N}\\
\mathcal{M}_{A\rightarrow B}(e^{i\hat{n}\phi}\rho e^{-i\hat{n}\phi}) &
=\left(  \bigotimes\limits_{i=1}^{m}e^{i\hat{n}_{i}\left(  -1\right)  ^{a_{i}%
}\phi}\right)  \mathcal{M}_{A\rightarrow B}(\rho)\left(  \bigotimes
\limits_{i=1}^{m}e^{-i\hat{n}_{i}\left(  -1\right)  ^{a_{i}}\phi}\right)
,\label{eq:phase-covariance-1-chan-M}%
\end{align}
where $a_{i}\in\{0,1\}$ for $i\in\{1,\ldots,m\}$ and $\hat{n}_{i}$ is the
photon number operator for the $i$th mode. Then it suffices to restrict the
optimization in the energy-constrained generalized channel divergence as
follows:%
\begin{equation}
\mathbf{D}_{\hat{n},N_{S}}(\mathcal{N}\Vert\mathcal{M})=\sup_{\psi
_{RA}:\operatorname{Tr}\{\hat{n}\psi_{A}\}=N_{S}}\mathbf{D}(\left(
\operatorname{id}_{R}\otimes\mathcal{N}_{A\rightarrow B}\right)  \left(
\psi_{RA}\right)  \Vert\left(  \operatorname{id}_{R}\otimes\mathcal{M}%
_{A\rightarrow B}\right)  \left(  \psi_{RA}\right)
),\label{eq:special-states-opt-gen-ch-div}%
\end{equation}
where $\psi_{RA}=|\psi\rangle\langle\psi|_{RA}$ and%
\begin{equation}
|\psi\rangle_{RA}=\sum_{n=0}^{\infty}\lambda_{n}|n\rangle_{R}|n\rangle
_{A},\label{eq:psi-state-conj-gen-div}%
\end{equation}
for some $\lambda_{n}\in\mathbb{R}^+$ such that $\sum_{n=0}^{\infty}
\lambda_{n} ^{2}=1$ and $\sum_{n=0}^{\infty} n \lambda
_{n} ^{2}=N_{S}$. Furthermore, the Gaussian energy-constrained
generalized channel divergence is achieved by the two-mode squeezed vacuum
state with parameter $N_{S}$, i.e.,
\begin{equation}
\mathbf{D}_{\hat{n},N_{S}}^{\mathcal{G}}(\mathcal{N}\Vert\mathcal{M}%
)=\mathbf{D}(\left(  \operatorname{id}_{R}\otimes\mathcal{N}_{A\rightarrow
B}\right)  \left(  \psi_{\operatorname{TMS}}(N_{S})\right)  \Vert\left(
\operatorname{id}_{R}\otimes\mathcal{M}_{A\rightarrow B}\right)  \left(
\psi_{\operatorname{TMS}}(N_{S})\right)
).\label{eq:TMSV-opt-gauss-gen-ch-div}%
\end{equation}

\end{proposition}

\begin{proof}
This result is an application of Lemma~\ref{lemma:cov-critical-step-cont} and
previous developments in our paper. We first exploit the joint displacement
covariance of the channels $\mathcal{N}_{A\rightarrow B}$ and $\mathcal{M}%
_{A\rightarrow B}$. That is, the fact that channels $\mathcal{N}_{A\rightarrow
B}$ and $\mathcal{M}_{A\rightarrow B}$ have the same $X$ matrix as given in
\eqref{eq:Gaussian-chan-N}--\eqref{eq:Gaussian-chan-M} implies that they are
jointly covariant with respect to displacements; i.e., for all input density
operators $\rho$ and unitary displacement operators $D(\alpha)\equiv \exp(\alpha \hat{a}^\dag - \alpha^* \hat{a})$, the following equalities hold
\begin{align}
\mathcal{N}_{A\rightarrow B}(D(\alpha)\rho D(-\alpha)) &  =\left(
\bigotimes\limits_{i=1}^{m}D(f_{i}(X,\alpha))\right)  \mathcal{N}_{A\rightarrow
B}(\rho)\left(  \bigotimes\limits_{i=1}^{m}D(-f_{i}(X,\alpha))\right)  ,\\
\mathcal{M}_{A\rightarrow B}(D(\alpha)\rho D(-\alpha)) &  =\left(
\bigotimes\limits_{i=1}^{m}D(f_{i}(X,\alpha))\right)  \mathcal{M}_{A\rightarrow
B}(\rho)\left(  \bigotimes\limits_{i=1}^{m}D(-f_{i}(X,\alpha))\right)  ,
\end{align}
where $f_{i}$ for $i\in\{1,\ldots,m\}$ are functions depending on the entries
of the matrix $X$ and $\alpha$. Let $\phi_{RA}$ be an arbitrary pure state such that
$\operatorname{Tr}\{\hat{n}\phi_{A}\}=N_{1}\leq N_{S}$. Consider the following
additive-noise Gaussian channel acting on an input state $\rho_{A}$:%
\begin{equation}
\mathcal{A}(\rho_{A})=\int d^{2}\alpha\ p_{  N_2
}(\alpha)\ D(\alpha)\rho_{A}D(-\alpha),\label{eq:thermal-bigger-thermal}%
\end{equation}
where  $p_{N_2}(\alpha)=\exp
\{-\left\vert \alpha\right\vert ^{2}/N_2\}/\pi N_2$ is a complex, centered Gaussian
probability density function with variance $N_2 \equiv N_S - N_1 \geq0$. Applying this channel to $\phi_{A}$
increases its photon number from $N_{1}$ to $N_{S}$:%
\begin{equation}
\operatorname{Tr}\{\hat{n}\mathcal{A}(\phi_{A})\}=N_{S},
\end{equation}
which follows because%
\begin{align}
\operatorname{Tr}\{\hat{n}\mathcal{A}(\phi_{A})\}  & =\operatorname{Tr}%
\{\hat{a}^{\dag}\hat{a}\int d^{2}\alpha\ p_{N_2
}(\alpha)\ D(\alpha)\phi_{A}D(-\alpha)\}\\
& =\int d^{2}\alpha\ p_{N_2  }(\alpha)\operatorname{Tr}%
\{D(-\alpha)\hat{a}^{\dag}\hat{a}D(\alpha)\phi_{A}\}\\
& =\int d^{2}\alpha\ p_{N_2  }(\alpha)\operatorname{Tr}%
\{D(-\alpha)\hat{a}^{\dag}D(\alpha)D(-\alpha)\hat{a}D(\alpha)\phi_{A}\}\\
& =\int d^{2}\alpha\ p_{N_2  }(\alpha)\operatorname{Tr}%
\{[  \hat{a}^{\dag}+\alpha^{\ast}]  [  \hat{a}+\alpha]
\phi_{A}\}\\
& =\int d^{2}\alpha\ p_{N_2  }(\alpha)\left[
\operatorname{Tr}\{\hat{a}^{\dag}\hat{a}\phi_{A}\}+\alpha\operatorname{Tr}%
\{\hat{a}^{\dag}\phi_{A}\}+\alpha^{\ast}\operatorname{Tr}\{\hat{a}\phi
_{A}\}+\left\vert \alpha\right\vert ^{2}\operatorname{Tr}\{\phi_{A}\}\right]
\\
& =N_{1}+0+0+N_2=N_{S}.
\end{align}
The first three equalities use definitions, cyclicity of trace, and the fact
that $D(\alpha)D(-\alpha)=I$. The fourth equality uses the well known
identities (see, e.g., \cite{AS17})%
\begin{equation}
D(-\alpha)\hat{a}D(\alpha)=\hat{a}+\alpha,\qquad D(-\alpha)\hat{a}^{\dag
}D(\alpha)=\hat{a}^{\dag}+\alpha^{\ast}.
\end{equation}
The second-to-last equality follows because $p_{N_2
}(\alpha)$ is a probability density function with mean zero and variance
$N_2$ (we have explicitly indicated what each of the four terms
evaluate to in the following line). Let $\varphi_{RA}$ denote a purification
of the state $\mathcal{A}(\phi_{A})$. We can then exploit the joint covariance
of the channels with respect to displacements, the relation in
\eqref{eq:thermal-bigger-thermal}, and
Lemma~\ref{lemma:cov-critical-step-cont} to conclude that%
\begin{equation}
\mathbf{D}(\mathcal{N}_{A\rightarrow B}(\varphi_{RA})\Vert\mathcal{M}%
_{A\rightarrow B}(\varphi_{RA}))\geq\mathbf{D}(\mathcal{N}_{A\rightarrow
B}(\phi_{RA})\Vert\mathcal{M}_{A\rightarrow B}(\phi_{RA})),
\end{equation}
for all $N_{1}\leq N_{S}$. As a consequence of this development, we find that
it suffices to restrict the optimization of the energy-constrained,
generalized channel divergence to pure bipartite states $\varphi_{RA}$ that
meet the energy constraint with equality (i.e., $\operatorname{Tr}\{\hat
{n}\varphi_{A}\}=N_{S}$).

Now we exploit the joint phase covariance of the channels. Let $\varphi_{RA}$
be a pure bipartite state that meets the energy constraint with equality.
Consider that%
\begin{equation}
\overline{\varphi}_{A}\equiv\frac{1}{2\pi}\int_{0}^{2\pi}d\phi\ e^{i\hat
{n}\phi}\varphi_{A}e^{-i\hat{n}\phi}=\sum_{n=0}^{\infty}|n\rangle\langle
n|\varphi_{A}|n\rangle\langle n|.
\end{equation}
That is, the state after phase averaging is diagonal in the number basis, and
furthermore, the resulting state $\overline{\varphi}_{A}$ has the same photon
number $N_{S}$ as $\varphi_{A}$ because%
\begin{align}
\operatorname{Tr}\{\hat{n}\overline{\varphi}_{A}\}  & =\frac{1}{2\pi}\int
_{0}^{2\pi}d\phi\ \operatorname{Tr}\{\hat{n}e^{i\hat{n}\phi}\varphi
_{A}e^{-i\hat{n}\phi}\}\\
& =\frac{1}{2\pi}\int_{0}^{2\pi}d\phi\ \operatorname{Tr}\{e^{-i\hat{n}\phi
}\hat{n}e^{i\hat{n}\phi}\varphi_{A}\}\\
& =\frac{1}{2\pi}\int_{0}^{2\pi}d\phi\ \operatorname{Tr}\{\hat{n}\varphi
_{A}\}=\operatorname{Tr}\{\hat{n}\varphi_{A}\}.
\end{align}
Thus, $\overline{\varphi}_{A}=\sum_{n=0}^{\infty}\lambda_{n}^2 |n\rangle\langle
n|_{A}$, for some $\lambda_{n}\in\mathbb{R}^+$ such that $\sum_{n=0}^{\infty
} \lambda_{n} ^{2}=1$ and $\sum_{n=0}^{\infty}
n \lambda_{n} ^{2}=N_{S}$. Let $\xi_{RA}$ denote a pure bipartite
state that purifies $\overline{\varphi}_{A}$. By applying
Lemma~\ref{lemma:cov-critical-step-cont} and the joint phase covariance
relations in
\eqref{eq:phase-covariance-1-chan-N}--\eqref{eq:phase-covariance-1-chan-M}, we
find that the following inequality holds
\begin{equation}
\mathbf{D}(\mathcal{N}_{A\rightarrow B}(\xi_{RA})\Vert\mathcal{M}%
_{A\rightarrow B}(\xi_{RA}))\geq\mathbf{D}\!\left(  \mathcal{N}_{A\rightarrow
B}(\varphi_{RA})\middle\Vert\mathcal{M}_{A\rightarrow B}(\varphi_{RA})\right)
.
\end{equation}
Since all purifications are related by isometries acting on the purifying
system $R$, and since a generalized divergence is invariant under such an
isometry \cite{TWW17}, we find that%
\begin{equation}
\mathbf{D}(\mathcal{N}_{A\rightarrow B}(\xi_{RA})\Vert\mathcal{M}%
_{A\rightarrow B}(\xi_{RA}))=\mathbf{D}(\mathcal{N}_{A\rightarrow B}(\psi
_{RA})\Vert\mathcal{M}_{A\rightarrow B}(\psi_{RA})),
\end{equation}
where $\psi_{RA}$ is a state of the form in \eqref{eq:psi-state-conj-gen-div}.
This concludes the proof of \eqref{eq:special-states-opt-gen-ch-div}.

To conclude \eqref{eq:TMSV-opt-gauss-gen-ch-div}, consider that the thermal
state $\theta(N_{S})$ is the only Gaussian state of a single mode that is
diagonal in the number basis with photon number equal to $N_{S}$. A
purification of the thermal state $\theta(N_{S})$ is the two-mode squeezed
vacuum $\psi_{\operatorname{TMS}}(N_{S})$ with parameter $N_{S}$. So this
means that, for a fixed photon number $N_{S}$, the two-mode squeezed vacuum
with parameter $N_{S}$ is optimal among all Gaussian states with reduced state
on the channel input having the same photon number.
\end{proof}

\bigskip

We note here that joint phase covariance of two otherwise arbitrary channels
implies that states of the form in \eqref{eq:psi-state-conj-gen-div} with mean
photon number  of their reduced states $\leq N_{S}$ are optimal, while joint
displacement covariance of two otherwise arbitrary channels implies that
states with mean photon number  of their reduced states $=N_{S}$ are optimal.
In Proposition~\ref{prop:gaussian-gen-div}, we chose to present the interesting case of
Gaussian channels in which both kinds of joint covariance hold simultaneously. The
aforementioned result regarding jointly phase-covariant channels was concluded
in \cite{N11} for a special case by employing a different argument and considering the special
case of fidelity and Chernoff-information divergences, as well as the
discrimination of pure-loss channels. It is worthwhile to note that our
argument is different, relying mainly on channel symmetries and data
processing, and thus applies in far more general situations than those considered in \cite{N11}.

Proposition~\ref{prop:gaussian-gen-div} applies to the various settings and
channels that we have considered in this paper for $\varepsilon$-degradable
and $\varepsilon$-close degradable bosonic thermal channels. Thus, we can
conclude in these situations that the Gaussian energy-constrained generalized
channel divergence is achieved by the two-mode squeezed vacuum state.

Particular generalized channel divergences of interest are the
energy-constrained diamond norm \cite{Sh17,Wetal17} and the
energy-constrained, channel version of the $C$-distance
\cite{R02,R03,GLN04,R06},  respectively defined as
\begin{align}
\left\Vert \mathcal{N}-\mathcal{M}
\right\Vert _{\Diamond,G,W}  &  \equiv\sup_{\psi_{RA}\ :\ \operatorname{Tr}
\{G\psi_{A}\}\leq W}\left\Vert (\operatorname{id}_{R}\otimes\mathcal{N}
_{A\rightarrow B})(\psi_{RA})-(\operatorname{id}_{R}\otimes\mathcal{M}
_{A\rightarrow B})(\psi_{RA})\right\Vert _{1},\\
C_{G,W}(\mathcal{N},\mathcal{M})  &
\equiv\sup_{\psi_{RA}\ :\ \operatorname{Tr}\{G\psi_{A}\}\leq W}\sqrt
{1-F((\operatorname{id}_{R}\otimes\mathcal{N}_{A\rightarrow B})(\psi
_{RA}),(\operatorname{id}_{R}\otimes\mathcal{M}_{A\rightarrow B})(\psi_{RA}
))},
\end{align}
where $F$ denotes the quantum fidelity.
Proposition~\ref{prop:gaussian-gen-div}\ implies that the
Gaussian-constrained  versions of these quantities reduce to the following for
channels satisfying  the assumptions stated there:
\begin{align}
\left\Vert \mathcal{N}-\mathcal{M}\right\Vert _{\Diamond,\hat{n},N_{S}}^{\mathcal{G}}  &  =\left\Vert
(\operatorname{id}_{R}\otimes\mathcal{N}_{A\rightarrow B})(\psi
_{\operatorname{TMS}}(N_{S}))-(\operatorname{id}_{R}\otimes\mathcal{M}
_{A\rightarrow B})(\psi_{\operatorname{TMS}}(N_{S}))\right\Vert _{1},\\
C_{\hat{n},N_{S}}^{\mathcal{G}}(\mathcal{N},\mathcal{M}
)  &  =\sqrt{1-F((\operatorname{id}_{R}\otimes\mathcal{N}
_{A\rightarrow B})(\psi_{\operatorname{TMS}}(N_{S})),(\operatorname{id}
_{R}\otimes\mathcal{M}_{A\rightarrow B})(\psi_{\operatorname{TMS}}(N_{S})))}.
\end{align}
We note that the latter quantity is readily expressed as a closed formula in
terms of the Gaussian specification of the channels $\mathcal{N}_{A\rightarrow
B}$ and $\mathcal{M}_{A\rightarrow B}$ in
\eqref{eq:Gaussian-chan-N}--\eqref{eq:Gaussian-chan-M}\ and the parameter
$N_{S}$ by employing the general formula for the fidelity of zero-mean
Gaussian states from \cite{PS00}. One could also employ the formulas from
\cite{SLW17} or \cite{PhysRevA.71.062320,K06} to compute Gaussian,
energy-constrained channel divergences based on R\'enyi relative entropy or
quantum relative entropy, respectively.

The result in \eqref{eq:special-states-opt-gen-ch-div} already significantly
reduces the set of states that we need to consider in computing a given
energy-constrained, generalized channel divergence for channels satisfying the conditions of Proposition~\ref{prop:gaussian-gen-div}. However, it is a very
interesting open question to determine whether, under the conditions given in
Proposition~\ref{prop:gaussian-gen-div}, the energy-constrained generalized
channel divergence is always achieved by the two-mode squeezed vacuum state
(if the restriction to Gaussian input states is lifted). Divergences of
interest in applications are the trace distance, fidelity, quantum relative
entropy, and R\'{e}nyi relative entropies. All of these measures lead to a
very interesting suite of Gaussian optimizer questions, which we leave for
future work. If there is a positive answer to this question, then we would
expect to see, in the low-photon-number regime, significant improvements of
the $\varepsilon$-degradable and $\varepsilon$-close degradable upper bounds
on the capacities of the thermal channel.

\section{Conclusion}

\label{sec:conclusion}

In this paper, we established several bounds on the energy-constrained quantum and private capacities of single-mode, phase-insensitive bosonic Gaussian channels. The energy-constrained bounds imply bounds for the corresponding unconstrained capacities.

In particular, we began by proving several different upper bounds on the energy-constrained quantum capacity of thermal channels. We discussed the closeness of these three upper bounds with a known lower bound. In particular, we have shown that the $\varepsilon$-close degradable bound works well only in the low-noise regime and that the data-processing upper bound is close to a lower bound for both low and high thermal noise. We also discussed an interesting case in which the $\varepsilon$-degradable bound is tighter than all other upper bounds. Also, our results establish strong limitations on any potential superadditivity of coherent information of a thermal channel in the low-noise regime.

Similarly, we established several different upper bounds on the energy-constrained private capacity of thermal channels.
We have also shown an improvement in the achievable rates of private communication through quantum thermal channels by using displaced thermal states as inputs to the channel.

Additionally, we proved several different upper bounds on the energy-constrained quantum and private capacities of quantum amplifier channels. We also established a data-processing upper bound on the energy-constrained quantum and private capacities of additive-noise channels. 

We also found that the data-processing bound can be at most 1.45 bits larger than a known lower bound on the energy-constrained quantum and private capacities of all phase-insensitive Gaussian channels. 

Building on recent developments in \cite{RMG18},  we proved even more bounds on the energy-constrained quantum and private capacities of the aforementioned channels.

Since thermal noise is present in almost all communication and optical systems, our results have implications for quantum computing and quantum cryptography. The knowledge of bounds on quantum capacity can be useful to quantify the performance of distributed quantum computation between remote locations, and private communication rates are connected to the ability to generate secret key.

We finally used the generalized channel divergence from \cite{LKDW17} to address the question of  optimal input states for the energy-bounded diamond norm and other related divergences. In particular, we showed that for two Gaussian channels that are jointly phase and displacement covariant, the Gaussian energy-constrained generalized channel divergence is achieved by a two-mode squeezed vacuum state that saturates the energy constraint. It is an interesting open question to determine whether, among all input states, the two-mode squeezed vacuum is the optimal input state for several energy-constrained, generalized channel divergences of interest. Here, we have reduced the optimization to be as given in \eqref{eq:special-states-opt-gen-ch-div}.

\bigskip
\textbf{Acknowledgements.}
We thank Vittorio Giovannetti, Kenneth Goodenough, Saikat Guha, Felix Leditzky, Iman Marvian,
Ranjith Nair,
Ty Volkoff, Christian Weedbrook, and Andreas Winter for discussions related to this paper. We also acknowledge the catalyzing role of the open problems session at Beyond i.i.d.~2015 (Banff International Research Research Station, Banff, Canada) in which the question of applying approximate degradability to bosonic channel capacities was raised. KS acknowledges support from the Department of Physics and Astronomy at LSU and the National Science Foundation under Grant No.~1714215. MMW
 thanks NICT for hosting him during
Dec.~2015 and acknowledges support from the Office of Naval Research and the National Science Foundation.
SA acknowledges support from the ARO, AFOSR, DARPA, NSF, and NGAS.
MT acknowledges CREST Japan Science and Technology Agency, Grant Number JPMJCR1772, and the ImPACT Program of Council for Science, Technology, and Innovation, Japan.

\appendix

	\bibliographystyle{alpha}
	\bibliography{Ref}
\end{document}